\documentclass[12pt]{article}
\usepackage{amssymb,amsmath,amsthm,latexsym,fullpage,tabularx,color,bbm}
\usepackage{enumerate,hyperref}
\usepackage{graphicx}
\usepackage{subcaption}
\usepackage[title]{appendix}
\usepackage{float}
\usepackage{cleveref}
\usepackage{colortbl}

\numberwithin{equation}{section}

\def\Nset{\mathbb{N}}
\def\Zset{\mathbb{Z}}
\def\Rset{\mathbb{R}}
\def\lebesgue{\mathrm{Leb}}

\def\gammaz{\gamma_Z}
\def\cz{c_Z}

\def\rme{\mathrm{e}}
\def\rmd{\mathrm{d}}
\def\pr{\mathbb{P}}
\def\esp{\mathbb{E}}
\def\cov{\mathrm{cov}}
\def\cum{\mathrm{cum}}
\def\var{\mathrm{var}}
\def\corr{\mathrm{corr}}
\def\convprob{\stackrel{P}\longrightarrow}
\def\convdistr{\stackrel{(d)}\longrightarrow}
\def\wrt{w.r.t.}

\def\constant{\mathrm{cst}}

\def\R{R_{\tilde{t}, \tilde{t}+\tilde{H}}}
\def\RV{RV_{\tilde{t}-\tilde{H}, \tilde{t}}}
\def\tR{\tilde{R}_{\tilde{t}, \tilde{t}+\tilde{H}}}

\definecolor{Gray}{gray}{0.85}

\newcommand\1[1]{\mathbbm{1}_{#1}}
\newcommand\ind[1]{\mathbbm{1}{\{#1\}}}
\newcommand\ddfrac[2]{\frac{\displaystyle #1}{\displaystyle #2}}

\newtheorem{theorem}{Theorem}[section]
\newtheorem{lemma}[theorem]{Lemma}

\newtheorem{corollary}[theorem]{Corollary}

\newtheorem{remark}[theorem]{Remark}

\begin{document}

{\renewcommand{\baselinestretch}{1.5}
\begin{titlepage}
    \title{
     \vspace{0in}
   Long-Horizon Return Predictability from Realized Volatility in Pure-Jump Point Processes
     \vspace{.05in}
    }
    \author{Meng-Chen Hsieh \footnote{Norm Brodsky College of Business Administration, Department of Information Systems, Analytics and Supply Chain Management, Rider University, Sweigart Hall 313, 2083 Lawrenceville Road, Lawrenceville, NJ 08648. E-mail: {\em mehsieh@rider.edu}} ~ \ \ Clifford Hurvich \footnote{Department of
    Technology, Operations and Statistics, Leonard N. Stern
    School of Business, New York University, Suite 8-52, 44 West 4th
    St., New York, NY 10012. E-mail: {\em churvich@stern.nyu.edu}}~ \ \ Philippe Soulier \footnote{D\'epartement de mathématiques
      et informatique, Universit\'e Paris Nanterre, 200 avenue de la R\' epublique, 92001 Nanterre
      cedex, France, E-mail: {\em philippe.soulier@parisnanterre.fr} } }
    \vspace{.05in}

  \date{\normalsize \today}
  \end{titlepage}
}
\maketitle

\begin{abstract}
We develop and justify methodology to consistently test for long-horizon return predictability based on realized variance. To accomplish this, we propose a parametric transaction-level model for the continuous-time log price process based on a pure-jump point process. The model determines the
returns and realized variance at any level of aggregation with properties shown to be consistent with the stylized facts in the empirical finance literature. Under our model, the long-memory parameter propagates unchanged from the transaction-level drift to the calendar-time returns and the realized variance, leading endogenously to a balanced predictive regression equation.
We propose an asymptotic framework using power-law aggregation in the predictive regression. Within this framework, we propose a hypothesis test for long-horizon return
predictability which is asymptotically correctly sized and consistent.
\end{abstract}

\newpage

\section{Introduction}
Research in return predictability has been an active area for decades.
Fama and French (1988), Campbell and Schiller (1987, 1988), Mishkin (1990, 1992), Boudoukh and Richardson (1993)
found return predictability based on long-run aggregated financial variables such as the dividend yield, price-dividend ratio, and functions of interest rates.
In these studies, a predictive regression with the overlapped aggregated return as the response and the similarly overlapped aggregated or non-aggregated financial variable as the predictor, was implemented. As the aggregation horizon increases, long-run return predictability was found to increase.

It has also been well documented in the literature that, even without aggregation, the strong autocorrelation of the predictors employed in these empirical studies induces  bias in the OLS estimator of the predictive coefficient.
Stambaugh (1999) considered a predictive regression model with an AR(1) predictor. He obtained an expression for the bias of
the OLS estimator of the predictive coefficient. Amihud and Hurvich (2004) also considered an AR(1) predictor and
introduced the augmented regression method to reduce the bias and studied a bias-corrected hypothesis test for return predictability.
Chen, Deo, and Yi (2013) again considered an AR(1) regressor and proposed the quasi restricted likelihood ratio test (QRLRT) for inference on the predictive coefficient.
They proposed a test that maintained the nominal size uniformly as the AR(1) coefficient approaches $1$
while delivering higher power than the competing procedures.
Boudoukh, Richardson, and Whitelaw (2008) considered a sequence of predictive regressions where overlapped aggregated returns at multiple horizons were regressed on an autocorrelated predictor without aggregation. They showed that the OLS estimators of the predictive coefficients were highly correlated across horizons under the assumption of no return predictability, inflating the size of a joint test that assumes no correlation. They conducted a joint Wald test for the OLS estimators of AR(1) predictors and found much weaker evidence of return predictability.

Furthermore, the overlapping aggregation of an originally-autocorrelated predictor
strengthens the degree of autocorrelation and increases the likelihood of spurious findings of long-run predictability of aggregated returns.
Valkanov (2003) demonstrated that the ordinary $t$-test in an OLS regression with overlapped aggregated autocorrelated regressors tends to over-reject the null hypothesis of no return predictability. 

Kostakis, Magdalinos, and Stamatogiannis (2018) considered a predictive regression on short-memory, integrated, and local-to-unity predictors. They proved that the long-horizon OLS estimator in a predictive regression is inconsistent when the predictor is integrated or is studied in a local-to-unity framework.
Using a procedure that robustified against the unknown persistence of the predictor, they found far weaker evidence of return predictability with increasing aggregation horizon compared with most previous empirical literature.

The literature described above considers regressors that either have short memory, are integrated, or are viewed in a local-to-unity framework.
This literature does not consider regressors that are fractionally integrated. Nevertheless, it has been established that predictors related to volatility, such as realized variance or VIX, are indeed fractionally integrated. For example, Andersen et al (2001) found that realized stock return volatility is well described by a long-memory process.
Bandi and Perron (2006) found a fractionally-cointegrated relationship between realized and implied volatility, suggesting that both series are long-memory processes.
Sizova (2013) considered a long-memory predictor in a predictive regression and proved that under certain assumptions as the level of aggregation increases, the population correlation between future returns and lagged realized variance converges to a constant which is a function of the long-memory parameter of the predictor. When the value of the long-memory parameter is zero, this population correlation converges to zero. Hence her theorem establishes that, under certain assumptions, a short-memory regressor does not lead to long-run return predictability as the level of aggregation increases.

The association between long-run future returns and current and past variance has been extensively studied in the financial economics literature.
Bansal and Yaron (2004), Bollerslev, Tauchen and Zhou (2009), and Drechsler and Yaron (2011) considered the impact of long-run economic uncertainty on consumption fluctuations, and thus the effect on the expected long-run returns.
They proved that the variance premium, defined as the difference between the variance of returns measured under the risk-neutral and physical probability measures, was
a good proxy for the latent economic shocks and was correlated with the expected excess returns. Hence these models imply that the variance premium predicts long-run returns.
In an empirical study, Bandi and Perron (2008) found that the long-run excess market return had a stronger correlation with past realized variance than it had with the classical dividend yield or the consumption-to-wealth ratio proposed by Lettau and Ludvigson (2001).
Hence following Bandi and Perron (2008), we will use realized variance as a predictor in our predictive regressions.

In spite of the long-memory property of the realized variance, Bandi and Perron (2008) did not model market variance as a fractionally integrated process. As they pointed out, if stock returns have short memory, then a predictive regression with a long-memory regressor would be unbalanced.
Indeed, rudimentary analysis of stock returns suggests that they have short memory although a weak long-memory component would be hard to detect. Researchers have adopted two frameworks to reconcile the apparently unbalanced nature of this  predictive regression. One is to impose assumptions on the functional relationship between the time series of returns and regressors as done by Sizova (2013),
who assumed that the predictor is in the domain of attraction of a fractional Brownian motion, and that the error term is additive and has shorter memory than the predictor,
so that both sides of the predictive regression equation have long memory with the same memory parameter.
The other is to directly model the continuous-time log price process, which then determines the returns and regressor at any level of aggregation.
Our work falls into the second framework. Under our model, the long-memory parameter propagates unchanged from the transaction level drift to the calendar-time returns and the realized variance with the same memory parameter, leading endogenously to a balanced predictive regression equation.

In the world of ultra-high-frequency data, the actual transaction‐level stock price is naturally viewed as a pure-jump process, that is, a marked point process where the points are the transaction times and the marks are the prices. The stock price observed in continuous time is a step function as opposed to a diffusion process. Engle and Russell (1998) studied the time intervals, i.e. durations, between successive transactions and proposed the ACD (Autoregressive Conditional Duration) model for the durations. Since their seminal work, several studies have been conducted to investigate the propagation of transaction-level properties of pure-jump processes under aggregation to lower-frequency time series in  discrete time.
Deo, Hsieh, and Hurvich (2010) studied the intertrade durations, counts, (i.e. the number of trade), squared returns, and realized variance of 10 NYSE stocks.  They found the presence of long memory in all of theses series. In light of this stylized fact, they proposed the LMSD (Long Memory Stochastic Duration) model for the durations. Deo et al (2009) provided sufficient conditions for the propagation of the long-memory parameter of durations to the corresponding counts and realized variance.

Cao, Hurvich, and Soulier (2017) studied the effect of drift in pure-jump transaction-level models for asset prices in continuous time, driven by a point process.
Under their model, the drift is proportional to the driving point process itself, i.e. the cumulative number of transactions.
This link reveals a mechanism by which long memory of the intertrade durations leads to long memory in the returns, with the same memory parameter.
Our proposed price model follows Cao, Hurvich, and Soulier's (2017) framework.  The calendar-time return derived by our model has a component that is proportional to the counts. Under this model, both calendar-time returns and realized variance are long-memory processes, which leads to a balanced predictive regression.
\footnote{Note that Bollerslev, Sizova, and Tauchen (2012) also considered generalizing their framework by modeling the volatility of the consumption growth rate as a long-memory process in the continuous-time framework of Comte and Renault (1996). Though their generalized model also leads to a balanced predictive regression (since the components of their variance premium regressor are fractionally-cointegrated and the response is a short memory return series), it would follow from Theorem 3 of Sizova (2013) if its assumptions held that the model of Bollerslev, Sizova, and Tauchen (2012) would not lead to long-run return predictability.}

Our paper makes the following contributions to the existing literature.
We propose a parametric transaction-level model for the log price.
Our model for the log price implies properties of the calendar-time returns and realized variance that are consistent with the stylized facts.
We propose an estimation procedure for the model parameters that is easy to implement.
For assessing return predictability, we propose to aggregate $\tilde{H} = \tilde{T}^{\kappa}$ ($\kappa \in (0,1)$) calendar-time returns in contrast to Sizova's choice of $\tilde{H}=\theta \tilde{T}$ ($\theta \in (0,1)$), where $\tilde{T}$ is the number of available calendar-time returns.
Sizova (2013) used the hypothesis test of Valkanov (2003) for long-run return predictability under the linear aggregation framework $\tilde{H}=\theta \tilde{T}$, but did not establish the consistency of this test.
Within the power-law framework $\tilde{H} = \tilde{T}^{\kappa}$, we propose a hypothesis test for long-run return predictability based on the sample correlation between the future aggregated returns and the past realized variance. We establish a central limit theorem for the test statistic under the null hypothesis of no long-run return predictability.
Our test is consistent, unlike the one used by Sizova (2013). We also provide simulations on a parametric bootstrap approach to testing for long-run return predictability under the power-law framework, and find that the test has a high power while maintaining the nominal size. We discuss the applicability of Sizova's (2013) assumptions for certain models in Section \ref{sec:comparison}. 

The paper is organized as follows. In Section \ref{sec:model}, we present our proposed price model.  In Section \ref{sec:return predictability}, we demonstrate return predictability by evaluating the correlation between the aggregated return and the lagged aggregated variance.
In Section \ref{sec:hypothesis_test_ret_predictability}, we propose a theoretically-based hypothesis test for long-run return predictability and use simulations to evaluate the size and the power of the bootstrap test. We compare our work with Sizova (2013) in Section \ref{sec:comparison} and conclude the paper in Section \ref{sec:conclusion}.
In Appendix Section \ref{sec:model_estimates}, we propose formulas for estimating the model parameters. In Section \ref{sec:sim_model}, we describe a procedure for simulating our model. The proofs of our theorems are provided in Section \ref{sec:proof}.

\section{The Model}\label{sec:model}
In this section, we propose our transaction-level model based on a pure-jump point process.
We then obtain the corresponding calendar-time return series and its properties.

\subsection{Cox Process Driven by Fractional Brownian Motion}\label{cox_process}
We start with the Cox process which is a key ingredient in our model.
The points of any point process $N$ consist of the sequence $\{T_{i},i\in\Zset\}$ (here, the transaction arrival times) such that
$T_i < T_{i+1}$ for all $i\in\Zset$ and $T_0 \leq 0 < T_1$.
For all measurable sets $A\subset \Rset$, we define the point process $N$ by
\begin{align*}
  N(A) = \sum_{i\in\Zset} \1{A}(T_i) \; .
\end{align*}
Let $\lambda(t), t \in \Rset$ be a non-negative stochastic intensity function, and
define the random measure $\Lambda$ by
\begin{equation}
  \label{intensity}
  \Lambda(A) = \int_A \lambda(t) \rmd t \; ,
\end{equation}
for all measurable sets $A$.
A Cox process, $N$, with mean measure $\Lambda$, is defined as follows: conditionally on $\Lambda$, $N$ is a Poisson process with mean measure $\Lambda$.

\noindent We now make further assumptions on the stochastic intensity function which guarantee long memory in returns and volatility when $d > 0$, where $d$ is the long-memory parameter. Let
\begin{equation}
\lambda(t)=\lambda \rme^{Z_H(t)},\; t \in \Rset, \; \lambda > 0\;,
\end{equation}
where $Z_H$ is a Gaussian stationary process with mean zero,
variance one and Hurst index $1/2 \le H <1$ (where $H=\frac{1}{2}+d$), which implies that the lag-$r$ autocovariance function
$\gammaz$ of $Z_H$ satisfies
\begin{align}
  \label{eq:Z-lrd}
  \lim_{r\to\infty} r^{2-2H} \gammaz(r) = \cz > 0 \;, \;\; H \in (1/2,1) \;.
\end{align}
Equation (\ref{eq:Z-lrd}) implies that the autocovariance of $Z_H$ is positive for
large lags.
The unconditional mean of $N(A)$ is equal to
\begin{align}
  \label{eq:esp-N}
  \esp[N(A)] =  \esp[ \esp[N(A)|\Lambda]] = \lambda \rme^{1/2} \lebesgue(A)   \; .
\end{align}
An example of such a process $Z_H$ satisfying (\ref{eq:Z-lrd}) is fractional Gaussian noise defined as
the increment process
\begin{equation}\label{def:frac_gaussian_noise}
Z_H(t) = B_H(ct) - B_H(ct-1)\;,
\end{equation}
where $B_H(t)$ is fractional Brownian motion and $c > 0$ is a time-scaling constant to ensure that our model is invariant to the choice of the time unit.
Under Equation (\ref{def:frac_gaussian_noise}), which we assume for the remainder of the paper, $Z_H$ has autocovariance function given by
\begin{equation}\label{eq:frac_Gaussian_gammaz}
 \gammaz(r) = \frac{1}{2} \left[ |cr+1|^{2H} - 2|cr|^{2H} + |cr-1|^{2H} \right]\;,
\end{equation}
$\forall r \in \mathbb{R}$, $\gammaz(r) \ge 0$ (see Lemma \ref{lem:gammaz_fun}).

\noindent For the short-memory case of $H=\frac{1}{2}$ ($d=0$), if $|r| \ge \frac{1}{c}$, $\gammaz(r)=0$; if $|r| < \frac{1}{c}$, $\gammaz(r)= 1-c|r|$. (see Lemma \ref{lem:exp_gammaz_d0}).
We show in Theorem \ref{thm:convergence_rho_shortM} that when $H=\frac{1}{2}$ there is no long-run predictability of aggregated future returns.

\subsection{Log Price Model}
Let $N(t) = N((0,t])$ denote the number of transactions in $(0,t]$.
For all $t \ge 0$, the log price process is defined as
\begin{equation}\label{eq:log price model_1}
\log P(t) = \sum_{k=1}^{N(t)} \tilde{e}_k \;,
\end{equation}
where
\begin{equation}\label{eq:efficient_shock}
 \tilde{e}_k = \mu + e_k \;,
\end{equation}
$\mu > 0$ is the drift term and $\{e_k\}$ are the efficient shocks,
which are assumed to be $i.i.d$ random variables with mean zero and variance $\sigma^2_e$, independent of $N$.
The efficient shock $e_k$ is assumed to reflect the true value of a stock and thus causes a permanent change to the $\log$ price.
Under (\ref{eq:efficient_shock}), Equation (\ref{eq:log price model_1}) can be expressed as
\begin{equation}\label{eq:log price model_2}
 \log P(t) = \sum_{k=1}^{N(t)}\left\{\mu + e_k \right\} = \mu N(t) + \sum_{k=1}^{N(t)} e_k \;.
\end{equation}
This together with Equation (\ref{eq:esp-N}) imply that $\esp[\log P(t)] = \mu\esp[N(t)] = \mu (\lambda\rme^{\frac{1}{2}}t) > 0$. Hence the expected $\log$ price is a linear function of $t$, which reflects the long-term growth of stock price as observed in the empirical finance literature.
For example, the average real annual returns of the Standard and Poor $500$ Index over the 90-year period from 1889 to 1978 is about $7\%$ (see Mehra and Prescott(1985)).

\noindent Though we will not pursue it in this paper, the model in Equation (\ref{eq:log price model_2}) could be generalized to multiple drift terms and point processes, for example
\begin{equation}\label{eq: price_model_two_shocks}
 \log P(t) = \sum_{k=1}^{N_1(t)}\left\{\mu_1 + e_{1,k}\right\} + \sum_{k=1}^{N_2(t)}\left\{\mu_1 + e_{2,k}\right\} = \mu_1 N_1(t) + \mu_2 N_2(t) + \sum_{k=1}^{N_1(t)} e_{1,k} + \sum_{k=1}^{N_2(t)} e_{2,k}\;,
\end{equation}
where $N_1(t)$ and $N_2(t)$ are mutually independent counting processes; $\{e_{1,k}\}$ and $\{e_{2,k}\}$ are i.i.d. efficient shocks independent of $N_1(t)$ and $N_2(t)$, both with means zero and the same standard deviation $\sigma^2_e$,
$\mu_1 > 0$, $\mu_2 < 0$, and $\mu_1 + \mu_2 > 0$. Here $N_1(t)$ and $N_2(t)$ can be thought of representing the numbers of buy and sell transactions.
As a result, the observed price reflects the net effect of buy and sell transactions. An active period for $N_1$ puts upward pressure on the price, while an active period for $N_2$ exerts downward pressure on the price. The assumption $\mu_1 + \mu_2 > 0$ reflects the long-term upward trend in the log prices.
The positive drift term in Equation (\ref{eq:log price model_2}) could represent the overall effect of $\mu_1$ and $\mu_2$ from (\ref{eq: price_model_two_shocks}). Thus, we consider model (\ref{eq:log price model_2}) as a simplified version of model (\ref{eq: price_model_two_shocks}). As we show in Lemma \ref{lem:cov_ret_pred} and Corollary \ref{cor:return_predicatability_extended_model}, both models exhibit properties of long-term return predictability. Henceforth, unless specified explicitly, we will focus on this simplified model (\ref{eq:log price model_2}).

\subsection{Return series and its properties}
The calendar-time return series $\{r_t\}^{\infty}_{t=1}$ measured at fixed clock-time intervals of width $\Delta t$ is
$r_t = \log P(t\Delta t) - \log P((t-1)\Delta t)$.
Note that $\Delta t$, for example, can be one minute, five minutes, 30 minutes, or one day.
Here we take $\Delta t=1$ (one day) for notational simplicity.
Thus the return series becomes $\{r_t\}^{\infty}_{t=1}$, which is given by
\begin{equation}\label{eq:return process}
  r_t = \log P(t) -\log P(t-1) = \mu \Delta N(t) + \sum_{k=N(t-1)+1}^{N(t)} e_k \;,
\end{equation}
where $\Delta N(t) = N(t) - N(t-1)$ is the number of transactions within the interval $(t-1,t]$.
We refer to the series $\{\Delta N(t)\}^{\infty}_{t=1}$ as the counts.

\noindent Next we present properties of the calendar-time returns $\{r_t\}$.

\begin{lemma}\label{lem:autocov_ret}
If $d > 0$, for all $L \in \Nset$, the lag-L autocovariance function of the return process $\{r_t\}$ is given by
\begin{equation}\label{eq:cov_ro_ret_L}
\cov(r_0, r_L) = \mu^2 \lambda^2 \rme^1 \int_{s=-1}^0 \int_{t=-1}^0 \left\{\rme^{\gammaz(L+t-s)}-1  \right\}\rmd t\rmd s \;.
\end{equation}
As $L \to \infty$,
\[
\cov(r_0, r_L) \sim \mu^2 \lambda^2 \rme \;\gammaz(L)\;.
\]
\end{lemma}
\begin{corollary}\label{cor:var_rt}
The variance of $r_t$ can be represented as
\begin{equation}\label{eq:var_r0}
\var(r_t)
= \mu^2\lambda^2\rme^1\int_{s=-1}^0\int_{t=-1}^0\left\{\rme^{\gammaz(s-t)}-1\right\}\rmd t \rmd s +\lambda \rme^{1/2}(\mu^2+\sigma_e^2)\;.
\end{equation}
\end{corollary}

\begin{theorem}\label{thm:long memory squared returns}
Let $\tilde{\gamma}(L)$ be the lag-L autocovariance function of the squared returns.  If $d > 0$ then
\[
 \tilde{\gamma}(L) =\cov(r_t^2\;,  r_{t+L}^2) \;.
\]
 Then
\[
 \lim_{L \rightarrow \infty}\frac{ \tilde{\gamma}(L)}{\gammaz(L)} = C \;,
\]
where C is a positive constant. 
Hence $\{r^2_t\}$ is a long-memory process.
\end{theorem}

\noindent Lemma \ref{lem:autocov_ret} shows that the returns $\{r_t\}$ have memory parameter $d$. In this paper, we consider returns which can be obtained directly from prices rather than excess returns, which cannot. The long memory of returns may appear to contradict empirical analysis, but the long memory may be hard to detect when embedded in noise. In Table \ref{tab:est_long_memory} and Figure \ref{fig:long_memory_signal}, we present
the estimated long memory parameter as well as the ACF and the log-log periodogram plots of returns simulated from the log price model with $d=0.35$.
The estimated long memory parameters based on different bandwidths are all very close to zero and the lags in the ACF plot appear to be statistically insignificant
\footnote{Indeed, $excess$ returns appear to have short memory in spite of the fact that in theory they would have memory parameter equal to $one$ if the risk free rate has a unit root as is found empirically.}.

\noindent Theorem \ref{thm:long memory squared returns} implies that the realized variance is a long-memory process.
Therefore, the calendar-time return and the realized variance derived from our model are long memory processes with the same memory parameter, which leads endogenously to a balanced predictive regression equation.

\begin{figure}[H]
\begin{subfigure}[b]{0.5\textwidth}
  \centering
  \includegraphics[width=1.0\linewidth]{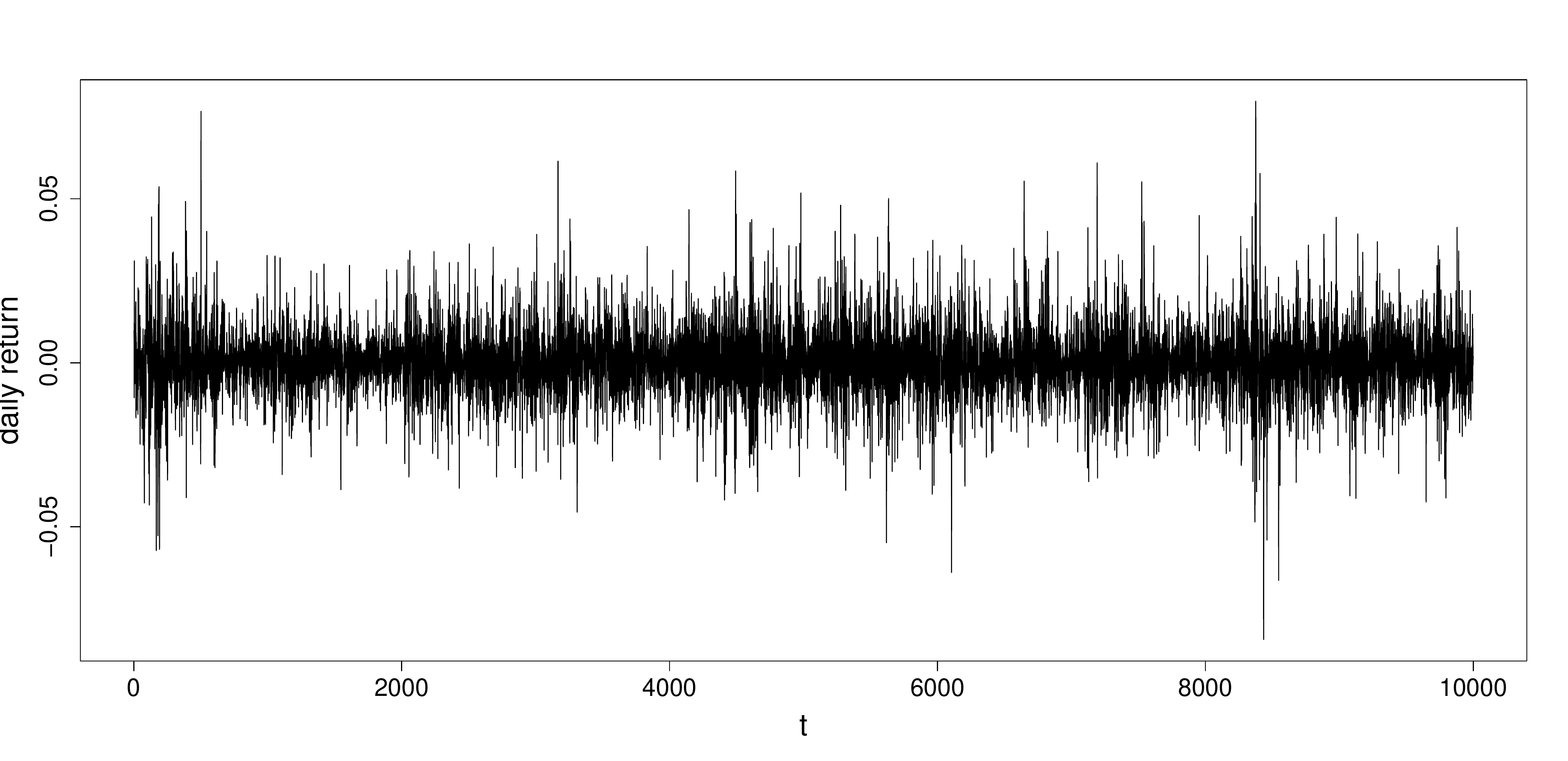}
  \caption{simulated returns}
  \end{subfigure}
\hspace{8pt}
\begin{subfigure}[b]{0.5\textwidth}
  \centering
  \includegraphics[width=1.0\linewidth]{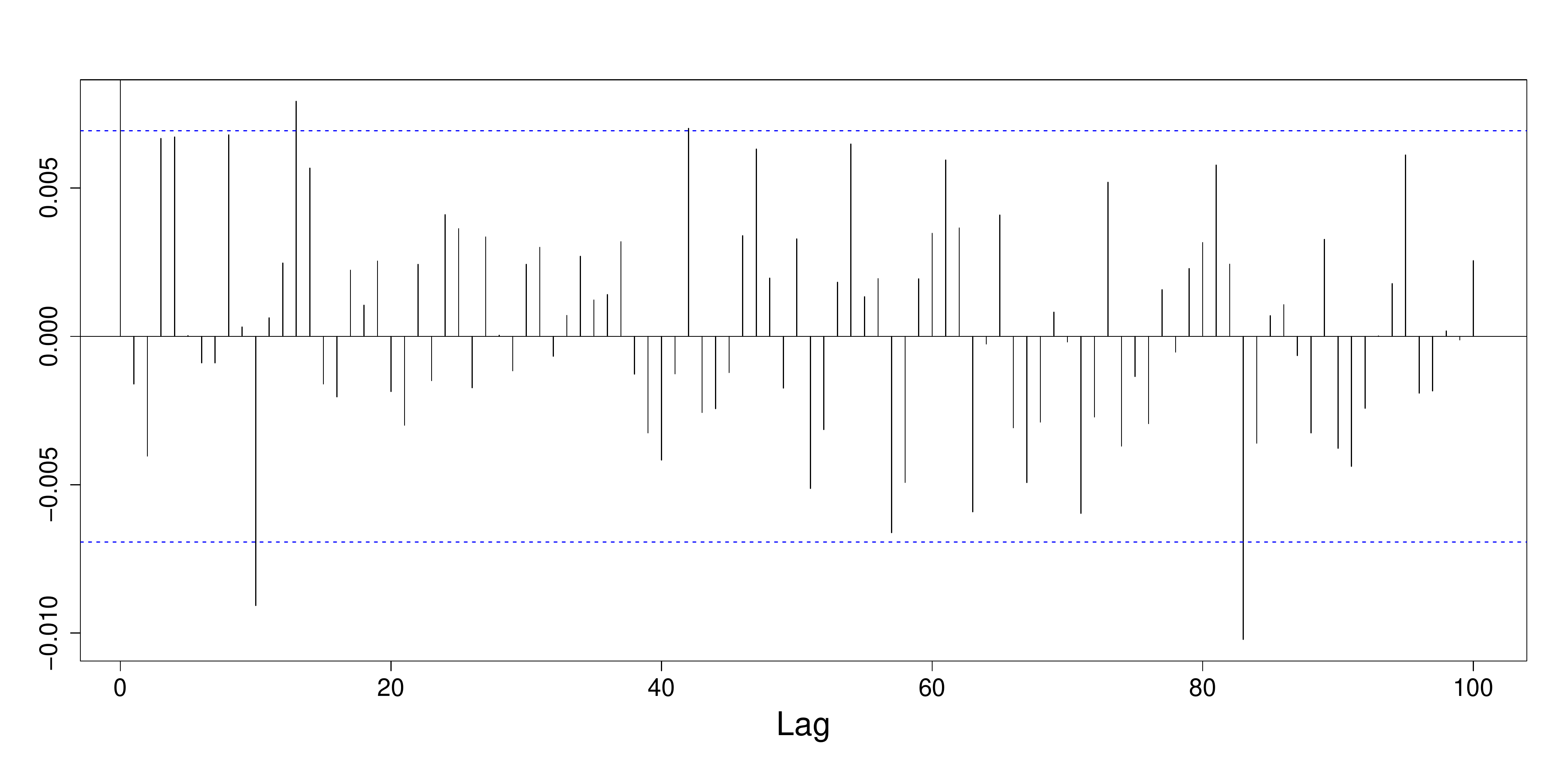}
  \caption{ACF of simulated returns}
 \end{subfigure}
\hspace{8pt}
\begin{subfigure}[b]{0.5\textwidth}
  \centering
  \includegraphics[width=1.0\linewidth]{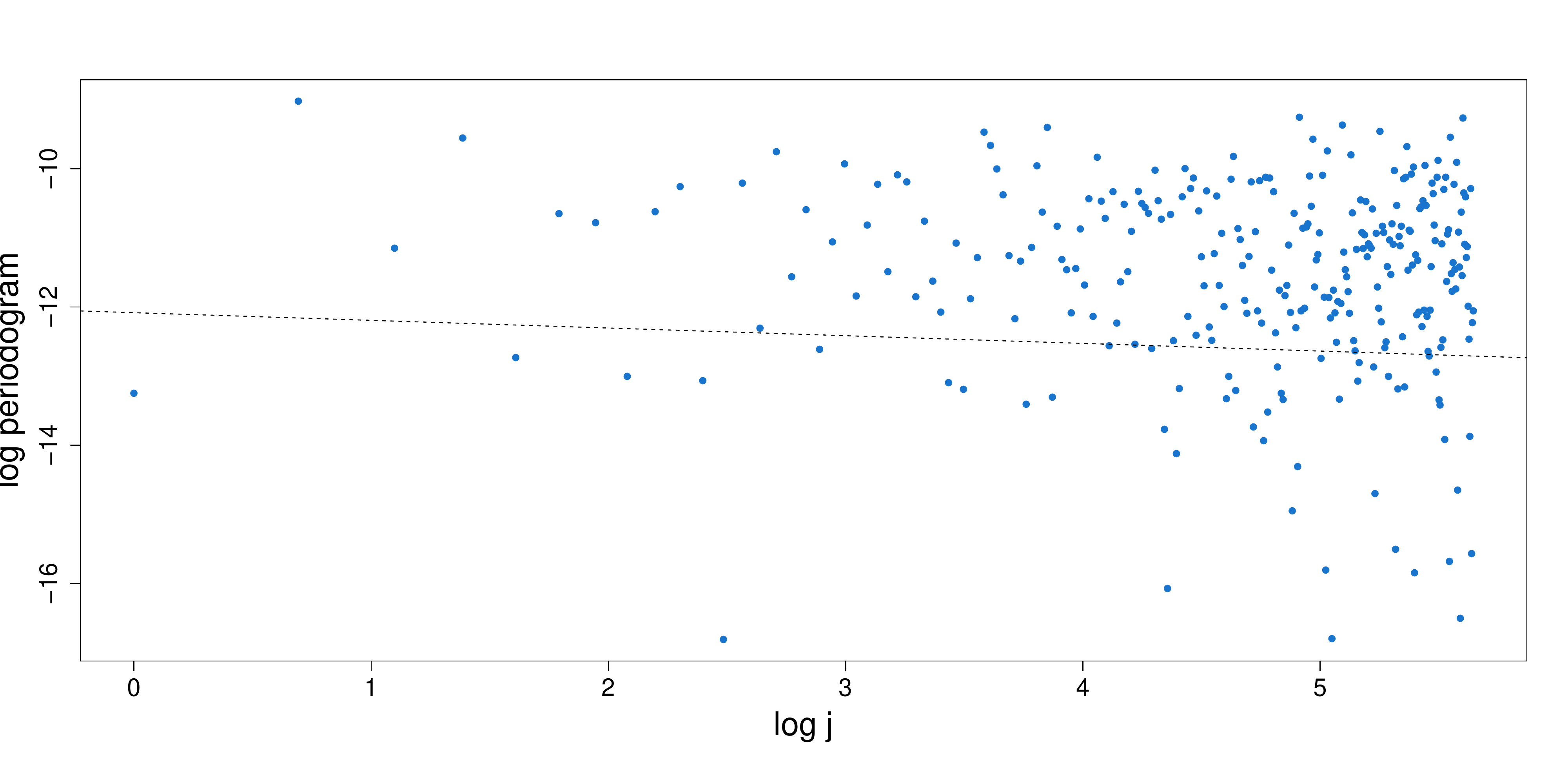}
  \caption{log-log periodogram}
  \end{subfigure}
\caption{One realization of $80,000$ simulated returns from the log-price model. Model parameters: $d=0.35$, $\mu =0.000001419188$, $\lambda = 128.2085$, $c=1$, and $\sigma_e=0.0007289$.}
\label{fig:long_memory_signal}
\end{figure}

\begin{table}[H]
\caption{Estimated long memory parameter $\hat{d}_{GPH}$ (a log-log regression of the periodogram versus frequency) based on one realization of $80000$ simulated daily returns. Model parameters: $\textbf{d=0.35}$, $\mu =0.000001419188$, $\lambda = 128.2085$, $c=1$, and $\sigma_e=0.0007289$.}
\center
\begin{tabular}{|c|c|c|c|}\hline
$d=0.35 $ & $m=n^{0.3}$  & $m=n^{0.5}$ & $m=n^{0.7}$  \\ \hline
 $\hat{d}_{GPH}$  & -0.0397 & 0.0555 & 0.0204 \\ \hline
\end{tabular}
\label{tab:est_long_memory}
\end{table}

\noindent In Appendices \ref{sec:model_estimates} and \ref{sec:sim_model}, we provide method-of-moments estimators of the model parameters and a methodology for simulating realizations of the log price model. We show that the proposed estimators are consistent in Theorem \ref{thm:consistent_estimators} and we evaluate the performance of these estimators in simulations.

\section{Return Predictability}\label{sec:return predictability}
{We use the aggregated past realized variance as the regressor to predict the aggregated future return.
For concreteness, we assume that $\tilde{t}=1, 2, \cdots, \tilde{T} $ is measured in months, there are $m$ trading days per month, and $r_t$ is the return for the $t^{th}$ day. The aggregated return over the next $\tilde{H}$ months is given by
\begin{equation}\label{eq:long term ret}
R_{\tilde{t}, \tilde{t}+\tilde{H}} = \sum_{t=\tilde{t}m+1}^{(\tilde{t}+\tilde{H})m}r_t \;,
\end{equation}
The realized variance over the past $\tilde{H}$ months is given by
\begin{equation}\label{eq:def_real_var}
RV_{\tilde{t}-\tilde{H}, \tilde{t}}  = \sum_{t=(\tilde{t}-\tilde{H})m+1}^{\tilde{t}m}r_t^2 \;.
\end{equation}
To evaluate return predictability over the next $\tilde{H}$ months, one can
compute the covariance between the future return and the past realized variance, $\cov(R_{\tilde{t}, \tilde{t}+\tilde{H}}, RV_{\tilde{t}-\tilde{H}, \tilde{t}})$ or the correlation
\begin{equation}\label{eq:corr_return_realV}
\corr\left(R_{\tilde{t},\tilde{t}+\tilde{H}}, RV_{\tilde{t}-\tilde{H}, \tilde{t}}\right) = \frac{\cov(R_{\tilde{t}, \tilde{t}+\tilde{H}}, RV_{\tilde{t}-\tilde{H}, \tilde{t}})}{\sqrt{\var(R_{\tilde{t}, \tilde{t}+\tilde{H}})\var(RV_{\tilde{t}-\tilde{H}, \tilde{t}})}}\;.
\end{equation}
Suppose first that $d > 0$. In Lemma \ref{lem:cov_ret_pred}, we show that $\cov(R_{\tilde{t}, \tilde{t}+\tilde{H}}\;, RV_{\tilde{t}-\tilde{H}, \tilde{t}})$ can be expressed as
\begin{equation}\label{eq:cov_ret_pred}
\cov\left(R_{\tilde{t}, \tilde{t}+\tilde{H}}, RV_{\tilde{t}-\tilde{H}, \tilde{t}}  \right) = \sum_{L=1}^{\tilde{H}m}\cov(r_L, r^2_0)L +\sum_{L=\tilde{H}m+1}^{2\tilde{H}m-1}\left(2\tilde{H}m-L\right)\cov(r_L, r^2_0)\;,
\end{equation}
and that $\cov(r_L, r^2_0) > 0$, for every positive integer L. Hence $\cov\left(R_{\tilde{t}, \tilde{t}+\tilde{H}}, RV_{\tilde{t}-\tilde{H}, \tilde{t}}  \right)$ is positive and there is return predictability. As the horizon $\tilde{H}$ increases, the correlation between $R_{\tilde{t}, \tilde{t}+\tilde{H}}$ and $RV_{\tilde{t}-\tilde{H}, \tilde{t}}$ increases, and it converges to a function of the long-memory parameter $d$:
\begin{theorem}\label{thm:convergence_rho}
If  $0 < d < \frac{1}{2}$,  as $\tilde{H} \to \infty$,
\[
\corr\left(R_{\tilde{t},\tilde{t}+\tilde{H}},\; RV_{\tilde{t}-\tilde{H},\tilde{t}}\right) \rightarrow  \left(2^{2d}-1\right) \;.
\]
\end{theorem}

\noindent This theorem has the same form as Theorem 3 in Sizova (2013), though our assumptions differ from hers.
On the other hand,  if $d=0$, i.e. short memory, then there is no long-term return predictability as shown in the following theorem:
\begin{theorem}\label{thm:convergence_rho_shortM}
If $d=0$, as $\tilde{H} \to \infty$,
\[
 \corr\left(R_{\tilde{t},\tilde{t}+\tilde{H}},\; RV_{\tilde{t}-\tilde{H},\tilde{t}}\right) \rightarrow 0 \;.
\]
\end{theorem}
\noindent Hence as pointed out by Sizova (2013), it is the presence of the long memory component, strengthened by the aggregation of the returns and the squared returns that causes the long-term return predictability. From now on, we use the notation $\rho$ to represent $\corr\left(R_{\tilde{t},\tilde{t}+\tilde{H}}, RV_{\tilde{t}-\tilde{H},\tilde{t}}\right)$ and $\hat{\rho}$ to represent the sample correlation between these quantities as estimated by}
\begin{eqnarray}\label{eq:def_rho_hat}
&& \hat{\rho}_{\tilde{H}, \tilde{T}}
= \ddfrac{\frac{1}{\tilde{T}-2\tilde{H}}\sum_{\tilde{t}=\tilde{H}}^{(\tilde{T}-\tilde{H})}\left(R_{\tilde{t},\tilde{t}+\tilde{H}}-\overline{R_{\tilde{t},\tilde{t}+\tilde{H}}}\right)
\left(RV_{\tilde{t}-\tilde{H},\tilde{t}}-\overline{RV_{\tilde{t}-\tilde{H},\tilde{t}}}\right)}
     {\sqrt{\frac{1}{\tilde{T}-2\tilde{H}}\sum_{\tilde{t}=\tilde{H}}^{(\tilde{T}-\tilde{H})}\left(R_{\tilde{t},\tilde{t}+\tilde{H}}-\overline{R_{\tilde{t},\tilde{t}+\tilde{H}}}\right)^2} \sqrt{\frac{1}{\tilde{T}-2\tilde{H}}\sum_{\tilde{t}=\tilde{H}}^{(\tilde{T}-\tilde{H})}\left(RV_{\tilde{t}-\tilde{H},\tilde{t}}-\overline{RV_{\tilde{t}-\tilde{H},\tilde{t}}}\right)^2
     }} \notag \\
\end{eqnarray}
where $\overline{R_{\tilde{t},\tilde{t}+\tilde{H}}}$ and $\overline{RV_{\tilde{t}-\tilde{H},\tilde{t}}}$ are the sample means of $\R$ and $\RV$.

\section{Statistical Inference on Return Predictability}\label{sec:hypothesis_test_ret_predictability}
Theorem \ref{thm:convergence_rho_shortM} establishes that when $d=0$, $\rho \to 0$ as $\tilde{H} \to \infty$ and thus there is no long-run return predictability.
Hence we wish to test the hypotheses of no long-run return predictability,
\begin{eqnarray*}
H_0: &\rho = 0 \;\; (d=0) \\
H_a: &\rho > 0 \;\; (d > 0) \;.
\end{eqnarray*}

We consider two different asymptotic frameworks for the degree of aggregation $\tilde{H}$.
The first framework is linear growth, $\tilde{H} = \theta \tilde{T}$, for $\theta \in (0,1)$, as used by Sizova (2013).
The second framework is a power law, $\tilde{H} = \tilde{T}^{\kappa}$, for $\kappa \in (0,1)$.
Under the linear aggregation framework, we will consider using $\hat{\rho}_{\tilde{H}, \tilde{T}}$ as the test statistic without further normalization.
Under the power-law framework, we will consider the test statistic $\sqrt{\tilde{T}^{1-\kappa}}\hat{\rho}_{\tilde{H}, \tilde{T}}$, or $\sqrt{\tilde{T}^{1-\kappa}}\tilde{\rho}_{\tilde{H}, \tilde{T}}$, where $\tilde{\rho}_{\tilde{H}, \tilde{T}}$ is a slightly modified version of $\hat{\rho}_{\tilde{H}, \tilde{T}}$ given by (\ref{eq:def_tilde_rho_hat}).

Under the linear growth framework, Sizova (2013) proved that $\hat{\rho}_{\tilde{H}, \tilde{T}}$ converges in distribution to a functional of fractional Brownian Motion subject to several assumptions (e.g.\;the dynamics of (14) and (15) as well as Assumptions 1 and 2 of that paper for both $d=0$ and $d> 0$).
Under these dynamics and assumptions, the resulting asymptotic distribution $F_\rho(A^d(\tau),B^d(\tau))$ for $d=0$ could be used to
obtain critical values for a hypothesis test based on $\hat{\rho}_{\tilde{H}, \tilde{T}}$. Though such a test is asymptotically correctly sized, it is inconsistent as implied by Sizova's (2013) Theorem 4 since the test statistic converges in distribution under both the null hypothesis and the alternative hypothesis. For our model, Figures \ref{fig:asym_dist_rho_hat_Sizova_d0} and \ref{fig:asym_dist_rho_hat_Sizova_d035} show the sampling distribution of $\hat{\rho}_{\tilde{H}, \tilde{T}}$ under the linear growth framework for $\tilde{H}$ based on simulated returns with $d=0$ and $d > 0$.
These results suggest that $\hat{\rho}_{\tilde{H}, \tilde{T}}$ converges in distribution when $d=0$ and $d > 0$, and thus that the power of a hypothesis test for $d=0$ based on $\hat{\rho}_{\tilde{H}, \tilde{T}}$ will not go to $1$ in this framework. Further support for this conclusion is provided by Tables \ref{tab:variance_rho_hat_sizova_d0} and \ref{tab:variance_rho_hat_sizova_d035}, where we calculate the variance of $\hat{\rho}_{\tilde{H}, \tilde{T}}$ for several values of $\tilde{T}$ under the linear growth framework for our model with $d=0$ and $d > 0$. The variance of $\hat{\rho}_{\tilde{H}, \tilde{T}}$ remains essentially constant. Apparently, the power of a correctly-sized test based on $\hat{\rho}_{\tilde{H}, \tilde{T}}$ would not approach $1$ as $\tilde{T} \to \infty$ under our model in the linear growth framework. Next we obtain the critical values of the correctly-sized test based on the simulated sampling distribution of $\hat{\rho}_{\tilde{H}, \tilde{T}}$ from our model under $d=0$ \footnote{Here we do not consider using the critical values obtained from the asymptotic null distribution in Sizova's (2013) Theorem 4 for the test. This is because our return model implies an endogenous relationship between the return and the realized variance, whereas Sizova's (2013) Theorem 4 assumes dynamic (as in her Eq.(15)) under which the realized variance is an exogenous variable.}.  Tables \ref{tab:size_H_T_sizova} and \ref{tab:test_power_sizova} show the size and the power of the hypothesis test based on $\hat{\rho}_{\tilde{H}, \tilde{T}}$ under the linear growth framework with nominal size of $5\%$. This test is approximately correctly sized but has low power even as the sample size increases, as explained above.

We next consider the power-law framework, which we show is more promising than the linear framework for testing for long-run return predictability.
The density plots of $\hat{\rho}_{\tilde{H}, \tilde{T}}$ under the power-law framework shown in Figures \ref{fig:asym_dist_rho_hat_power_law_d0} and \ref{fig:sample_dist_rho_hat_power_law_d035} and Tables \ref{tab:variance_rho_hat_d0} and \ref{tab:variance_rho_hat_d035} indicate that as $\tilde{T}$ increases, the variance of $\hat{\rho}_{\tilde{H}, \tilde{T}}$ decreases. Furthermore, the mean of $\hat{\rho}_{\tilde{H}, \tilde{T}}$ increases with $\tilde{T}$ when $d > 0$, perhaps approaching $2^{2d}-1$, though this is not clear based on the sample sizes considered. This suggests the possibility to construct a consistent test for long-run return predictability based on a rescaled version of $\hat{\rho}_{\tilde{H}, \tilde{T}}$. To determine an appropriate rescaling factor, we study the behavior of the variance of
$\hat{\rho}_{\tilde{H}, \tilde{T}}$ as $\tilde{T}$ increases under the null hypothesis, $d=0$. To study the convergence rate of the variance of $\hat{\rho}_{\tilde{H}, \tilde{T}}$, we created a scatter plot for the logged variance of $\hat{\rho}_{\tilde{H}, \tilde{T}}$ and $\log {\tilde{T}}$ presented in Figure \ref{fig:log_log_plot}.
For each value of $\kappa$, the scatter plot appears to be linear. To confirm this, we regress the logged variance of $\hat{\rho}_{\tilde{H}, \tilde{T}}$ on $\log \tilde{T}$ and present the fitted equations in Table \ref{tab:fit_reg_var_rho_hat} . The fitted slope coefficients are very close to the values of $\kappa -1$. Hence we conjecture that the variance of $\hat{\rho}_{\tilde{H}, \tilde{T}}$ is proportional to $\tilde{T}^{\kappa-1}$. This leads to the rescaled test statistic $\sqrt{\tilde{T}^{1-\kappa}}\hat{\rho}_{\tilde{H}, \tilde{T}}$.

Density plots in Figure \ref{fig:asym_dist_adj_rho_hat_power_law_d0} as well as Shapiro-Wilk normality test results in Table $\ref{tab:k-s_normality_test}$ for $d=0$
indicate that $\sqrt{\tilde{T}^{1-\kappa}}\hat{\rho}_{\tilde{H}, \tilde{T}}$ may be asymptotically normal under the null hypothesis.
Therefore, from now on we will focus on the power-law aggregation framework. We will also restrict attention to hypothesis testing under our model.
We consider two options for hypothesis testing under the power law framework: (i) bootstrap method and (ii) asymptotic method.

In the bootstrap approach, given one realization of $\log P(t)$, $t \in (0,T]$,
we compute the observed value of $\sqrt{\tilde{T}^{1-\kappa}}\hat{\rho}_{\tilde{H}, \tilde{T}}$, $\sqrt{\tilde{T}^{1-\kappa}}\hat{\rho}_{\tilde{H}, \tilde{T}}^{obs}$, and estimate the model parameters using the formulas in Section \ref{sec:model_estimates}. We then set $d=0$ and use the other estimated parameters to simulate 1000 replications of $\sqrt{\tilde{T}^{1-\kappa}}\hat{\rho}_{\tilde{H}, \tilde{T}}$. We use the $95^{th}$ percentile among the 1000 replications, $\sqrt{\tilde{T}^{1-\kappa}}{\hat{\rho}}_{95\%}$, as the critical value of the test.  Thus, we compare $\sqrt{\tilde{T}^{1-\kappa}}{\hat{\rho}_{\tilde{H}, \tilde{T}}}^{obs}$ with ${\hat{\rho}}_{95\%}$ and reject the null hypothesis if $\sqrt{\tilde{T}^{1-\kappa}}\hat{\rho}_{\tilde{H}, \tilde{T}}^{obs} > \sqrt{\tilde{T}^{1-\kappa}}\hat{\rho}_{95\%}$.

Davidson and MacKinnon (1999, 2006) showed that subject to some assumptions, if the model parameter estimator is consistent under the null hypothesis, the rejection rate of a parametric bootstrap test converges to the nominal size, where the convergence rate depends on the order of convergence of the consistent estimator. In Theorem \ref{thm:consistent_estimators}, we prove that our model estimators are consistent under $d=0$. We find empirically that the rejection rates of our bootstrap hypothesis test (reported in Table \ref{tab:bootstrap_size}) are almost never significantly different from the nominal size. Furthermore, the power (see Table \ref{tab:test_power_power_law}) increases with the sample size and the power is highest for smaller values of $\kappa$. Nevertheless, we do not have a theoretical justification for the bootstrap hypothesis test in our framework.

Next, we develop an asymptotically correctly-sized consistent test based on $\sqrt{\tilde{T}^{1-\kappa}}\tilde{\rho}_{\tilde{H}, \tilde{T}}$, where $\tilde{\rho}_{\tilde{H}, \tilde{T}}$ is a slightly modified version of $\hat{\rho}_{\tilde{H}, \tilde{T}}$ (defined in (\ref{eq:def_tilde_rho_hat})).
When $d=0$, our return model implies that the daily return $r_t$ is $h$-dependent, where the lag $h$ is determined by the time-scaling constant $c$ (See (\ref{eq:def_h}) for the definition of $h$). To take advantage of this fact (see Remark \ref{rmk:mean_cov_explode} below),
we propose a modified sample correlation coefficient $\tilde{\rho}_{\tilde{H}, \tilde{T}}$, where in the aggregation of returns, we skip the first $h$ daily returns from the first day of month $\tilde{t}$ and aggregate the remaining daily returns over the $\tilde{H}$ months to construct the modified aggregated return $\tR$ and then compute the sample correlation between $\tR$ and $\RV$.

\begin{theorem}\label{thm:asym_normal_dist_norm_rho}
Define the modified sample correlation coefficient
\begin{eqnarray}\label{eq:def_tilde_rho_hat}
&& \tilde{\rho}_{\tilde{H}, \tilde{T}}
= \ddfrac{\frac{1}{\tilde{T}-2\tilde{H}}\sum_{\tilde{t}=\tilde{H}+1}^{\tilde{T}-\tilde{H}}\left(\tR -\overline{\tilde{R}_{\tilde{H}}}\right)
\left(RV_{\tilde{t}-\tilde{H},\tilde{t}}-\overline{RV_{\tilde{H}}}\right)}
     {\sqrt{\frac{1}{\tilde{T}-2\tilde{H}}\sum_{\tilde{t}=\tilde{H}+1}^{\tilde{T}-\tilde{H}}\left(\tR-\overline{\tilde{R}_{\tilde{H}}}\right)^2} \sqrt{\frac{1}{\tilde{T}-2\tilde{H}}\sum_{\tilde{t}=\tilde{H}+1}^{\tilde{T}-\tilde{H}}\left(RV_{\tilde{t}-\tilde{H},\tilde{t}}-\overline{RV_{\tilde{H}}}\right)^2
     }}
\end{eqnarray}
where
\begin{equation}\label{eq:tilde_R}
\tR = \sum_{t=\tilde{t}m+h+1}^{(\tilde{t}+\tilde{H})m}r_t
\end{equation}
and
\begin{equation}\label{eq:sample_mean_tilde_R}
 \overline{\tilde{R}_{\tilde{H}}} = \frac{1}{\tilde{T}-2\tilde{H}}\sum_{\tilde{t}=\tilde{H}+1}^{\tilde{T}-\tilde{H}}\tR.
\end{equation}
\\
If $d=0$ and $\tilde{H}=\tilde{T}^{\kappa}$, for $\kappa \in (0, 1/3)$,
\begin{equation}\label{eq:asym_dist_tilde_rho}
\sqrt{\tilde{T}^{1-\kappa}}\; \tilde{\rho}_{\tilde{H}, \tilde{T}} \convdistr N\left(0,\; \frac{S^2}{A_1 A_2} \right)\;,
\end{equation}
where $A_1= \var(r_0) + 2 \sum_{k=1}^h \cov(r_0, r_k)$, $A_2= \var(r^2_0) + 2 \sum_{k=1}^h \cov(r^2_0, r^2_k)$,
and $S^2 = \frac{2}{3}  \sum_{|u|\leq h}\sum_{|v|\leq h} \cov(r_0,r_u)  \cov(r_0^2,r_v^2)$.
\end{theorem}

\begin{remark}
$\cov(r_0, r_k)$, $\cov(r^2_0, r^2_k)$, $\var(r_0)$,
and $\var(r^2_0)$ can be evaluated by (\ref{eq:cov_ro_ret_L}),(\ref{eq: cov_sq_ret_L}), (\ref{eq:var_rt_d0}), and (\ref{eq:var_rt_sq}) under $d=0$.
\end{remark}

\begin{remark}\label{rmk:mean_cov_explode}
The reason we do not attempt to establish a CLT for $\hat{\rho}_{\tilde{H}, \tilde{T}}$ is as follows:
For any $d$, we have shown that $\cov(\R, \RV) > 0$ (see Lemma \ref{lem:cov_ret_pred}). When $d=0$, $\cov(\R, \RV) = \sum_{k=1}^h \cov(r_{k}, r^2_0)k$ (See
(\ref{eq:rho_d0_numerator})).
This suggests that for $\kappa \in (0,1/3)$, $\sqrt{\tilde{T}^{1-\kappa}}\; \hat{\rho}_{\tilde{H}, \tilde{T}}$ diverges.
The modified version $\tR$ is independent of $\RV$ and thus $\cov(\tR, \RV)=0$.
\end{remark}

\begin{table}[H]
\caption{Variance of 1000 realizations of $\hat{\rho}_{\tilde{H}, \tilde{T}}$ using $\tilde{H}=\theta T$ when $d=0$, $\mu =0.000001419188$, $\lambda = 128.2085$, $c=1$, and $\sigma_e=0.0007289$.}
\center
\begin{tabular}{|c|c|c|c|c|c|}\hline
  $\theta=0.0125$   &  $\theta=0.025$  & $\theta=0.05$ & $\theta=0.125$ & $\theta=0.25$ & $\tilde{T}$ \\ \hline
  0.01259	&0.01792	&0.04084	&0.10074	&0.22477 & 131  \\
  0.00785	&0.01933	&0.03776	&0.09905	&0.21818 & 262  \\
  0.00959	&0.01720	&0.03505	&0.09713	&0.21603 & 524  \\
  0.00837	&0.01739	&0.03641	&0.09327	&0.21671 & 1048 \\
  0.00873	&0.01731	&0.03561	&0.10008	&0.21772 & 2097 \\
  0.00844	&0.01738	&0.03828	&0.11019	&0.23044 & 4194 \\ \hline
\end{tabular}\label{tab:variance_rho_hat_sizova_d0}
 \end{table}

\begin{table}[H]
\caption{Variance of 1000 realizations of $\hat{\rho}_{\tilde{H}, \tilde{T}}$ using $\tilde{H}=\theta T$ when $d=0.3545$, $\mu =0.000001419188$, $\lambda = 128.2085$, $c=1$, and $\sigma_e=0.0007289$.}
\center
\begin{tabular}{|c|c|c|c|c|c|}\hline
  $\theta=0.0125$   &  $\theta=0.025$  & $\theta=0.05$ & $\theta=0.125$ & $\theta=0.25$ & $\tilde{T}$ \\ \hline
 0.01488	&0.02223	&0.04795	&0.10958	&0.26234 & 131 \\
 0.01163	&0.02402	&0.04267	&0.12113	&0.27342 & 262 \\
 0.01416	&0.02460	&0.04477	&0.11967	&0.25633 & 524 \\
 0.01203	&0.02347	&0.04523	&0.12059	&0.27081 & 1048 \\
 0.01217	&0.02410	&0.04785	&0.11544	&0.26542 & 2097 \\
 0.01223	&0.02504	&0.04867	&0.12412	&0.26000 & 4194 \\ \hline
\end{tabular}\label{tab:variance_rho_hat_sizova_d035}
 \end{table}

\begin{table}[H]
\caption{Variance of 1000 realizations of $\hat{\rho}_{\tilde{H}, \tilde{T}}$ using $\tilde{H}=T^{\kappa}$ when $d=0$. Model parameters: $\mu =0.000001419188$, $\lambda = 128.2085$, $c=1$, and $\sigma_e=0.0007289$.}
\center
\begin{tabular}{|c|c|c|c|c|}\hline
  $\kappa=0.1$   &  $\kappa=0.3$  & $\kappa=0.5$ & $\kappa=0.7$ & $\tilde{T}$ \\ \hline
0.01259	&0.02289	&0.06194	&0.21163  & 131 \\
0.00591	&0.01330	&0.04637	&0.15999  & 262 \\
0.00271	&0.00959	&0.03061	&0.11927  & 524 \\
0.00146	&0.00517	&0.02218	&0.09242  & 1048 \\
0.00068	&0.00313	&0.01508	&0.07841  & 2097 \\
0.00035	&0.00188	&0.01026	&0.06406  & 4194 \\ \hline
\end{tabular}\label{tab:variance_rho_hat_d0}
 \end{table}

\begin{table}[H]
\caption{Variance of 1000 realizations of $\hat{\rho}_{\tilde{H}, \tilde{T}}$ using $\tilde{H}=T^{\kappa}$ when $d=0.3545$. Model parameters: $\mu =0.000001419188$, $\lambda = 128.2085$, $c=1$, and $\sigma_e=0.0007289$.}
\center
\begin{tabular}{|c|c|c|c|c|}\hline
$\kappa=0.1$   &  $\kappa=0.3$  & $\kappa=0.5$ & $\kappa=0.7$ & $\tilde{T}$ \\ \hline
0.01488	&0.02904	&0.07309	&0.23814 & 131 \\
0.00859	&0.01773	&0.05195	&0.19369 & 262 \\
0.00446	&0.01416	&0.04036	&0.14804 & 524 \\
0.00242	&0.00776	&0.02823	&0.11967 & 1048 \\
0.00120	&0.00491	&0.02129	&0.09154 & 2097 \\
0.00068	&0.00321	&0.01529	&0.08061 & 4194 \\ \hline
\end{tabular}\label{tab:variance_rho_hat_d035}
 \end{table}

\begin{table}[H]
\caption{Size of Sizova's (2013) return predictability hypothesis test: $H_0:\; d=0$ vs. $H_a:\; d > 0$ using $\tilde{H}=\theta \tilde{T}$. Critical values are obtained from 1000 replications of simulated $\hat{\rho}_{\tilde{H}, \tilde{T}}$ from our return model.} Model parameters: $\mu =0.000001419188$, $\lambda = 128.2085$, $c=1$, and $\sigma_e=0.0007289$. The symbol $\ast$ indicates rejection of $size=5\%$.
\center
\begin{tabular}{|c|c|c|c|c|c|}\hline
  $\theta=0.0125$   &  $\theta=0.025$  & $\theta=0.05$ & $\theta=0.125$ & $\theta=0.25$ & $\tilde{T}$ \\ \hline
  0.054 &0.058 &0.063 &0.044 &0.055 &  131  \\
  0.055 &0.045 &0.052 &0.051 &0.063 &  262  \\
  0.054 &0.053 &$0.035^{\ast}$ &0.056 &$0.034^{\ast}$ &  524  \\
  0.057 &0.056 &0.062 &0.046 &$0.036^{\ast}$ &  1048 \\
  0.065 &0.053 &0.059 &0.049 &$0.035^{\ast}$ &  2097 \\
  0.053 &$0.090^{\ast}$ &$0.067^{\ast}$ & 0.060 &0.044 &  4194 \\ \hline
\end{tabular}\label{tab:size_H_T_sizova}
\end{table}

\begin{table}[H]
\caption{Power of return predictability hypothesis test: $H_0:\; d=0$ vs. $H_a:\; d > 0$ using $\tilde{H}=\theta \tilde{T}$.  The power is calculated as the percentage of cases that reject the null hypothesis among the 1000 realizations of return process simulated under $d=0.3545$. Model parameters: $\mu =0.000001419188$, $\lambda = 128.2085$, $c=1$, and $\sigma_e=0.0007289$.}
\center
\begin{tabular}{|c|c|c|c|c|c|}\hline
$\theta=0.0125$   &  $\theta=0.025$  & $\theta=0.05$ & $\theta=0.125$ & $\theta=0.25$ & $\tilde{T}$ \\ \hline
 0.114  & 0.117   & 0.093   & 0.082  &0.072    & 131  \\
 0.165  & 0.131   & 0.085   & 0.076  &0.097    & 262  \\
 0.197  & 0.132   & 0.109   & 0.085  &0.060    & 524  \\
 0.236  & 0.189   & 0.135   & 0.083  &0.072    & 1048 \\
 0.319  & 0.257   & 0.185   & 0.070  &0.065    & 2097 \\
 0.394  & 0.372   & 0.217   & 0.088  &0.048    & 4194 \\ \hline
\end{tabular}\label{tab:test_power_sizova}
\end{table}

\begin{table}[H]
\caption{Size of bootstrap return predictability hypothesis test: $H_0:\; d=0$ vs. $H_a:\; d > 0$ using $\tilde{H}=\tilde{T}^{\kappa}$.  The size is calculated as the percentage of cases that reject the null hypothesis among the 500 realizations of return process simulated under $d=0$. Model parameters: $\mu =0.000001419188$, $\lambda = 128.2085$, $c=1$, and $\sigma_e=0.0007289$. The symbol $\ast$ indicates rejection of $size=5\%$.}
\center
\begin{tabular}{|c|c|c|c|c|}\hline
  $\kappa=0.1$   &  $\kappa=0.3$  & $\kappa=0.5$ & $\kappa=0.7$ & $\tilde{T}$ \\ \hline
 0.055  & 0.062  & 0.057  & 0.054  &  131   \\
 0.042  & 0.041  & 0.042  & 0.055  &  262   \\
 0.051  & 0.051  &  $0.066^{\ast}$  & 0.057 &  524  \\
 0.058  & 0.054  & 0.05  & 0.052  &  1048 \\
 0.041  & 0.053  & 0.046  & 0.055  &  2097 \\
 0.051  &  $0.068^{\ast}$  & 0.049  & 0.049  &  4194 \\ \hline
\end{tabular}\label{tab:size_H_T_power}
 \end{table}

\begin{table}[H]
\caption{Power of bootstrap return predictability hypothesis test: $H_0:\; d=0$ vs. $H_a:\; d > 0$ using $\tilde{H}=\tilde{T}^{\kappa}$.  The power is calculated as the percentage of cases that reject the null hypothesis among the 1000 realizations of return process simulated under $d=0.3545$. Model parameters: $\mu =0.000001419188$, $\lambda = 128.2085$, $c=1$, and $\sigma_e=0.0007289$.}
\center
\begin{tabular}{|c|c|c|c|c|}\hline
  $\kappa=0.1$   &  $\kappa=0.3$  & $\kappa=0.5$ & $\kappa=0.7$ & $\tilde{T}$ \\ \hline
 0.124  & 0.097  & 0.074   &0.074  &  131  \\
 0.175  & 0.138  & 0.113   &0.063  &  262  \\
 0.233  & 0.185  & 0.133   &0.079  &  524  \\
 0.356  & 0.267  & 0.168   &0.072  &  1048 \\
 0.518  & 0.407  & 0.235   &0.108  &  2097 \\
 0.769  & 0.630  & 0.367   &0.166  &  4194 \\ \hline
 \end{tabular}\label{tab:test_power_power_law}
 \end{table}

 \begin{table}[H]
\caption{Rejection rate from the bootstrap hypothesis test. $\ast$ indicates rejection of the nominal size $5\%$. }
\center
\begin{tabular}{|c|c|c|c|c|}\hline
\multicolumn{5}{|c|}{$\kappa$} \\ \hline
$0.1$   &  $0.3$  & $0.5$ & $0.7$ & $\tilde{T}$ \\ \hline
0.041   & 0.053          & 0.046 &  0.055 & 2000 \\
0.051   & $0.068^{\ast}$ & 0.049 &  0.049 & 4000 \\ \hline
\end{tabular}\label{tab:bootstrap_size}
\end{table}

\begin{table}[H]
\caption{Estimated intercept $\hat{\beta}_0$ and slope $\hat{\beta}_1$ in regression of  $\log \hat{\var}(\hat{\rho}_{\tilde{H}, \tilde{T}})$ on $\log \tilde{T}$ with $\tilde{H}=\tilde{T}^{\kappa}$. }
\center
\begin{tabular}{|c|c|c|c|c|}\hline
& \multicolumn{4}{c|}{$\kappa$} \\ \hline \hline
& $0.1$ &  $0.3$  & $0.5$  & $0.7$  \\ \hline
$\hat{\beta}_0$   & 0.025  & -0.683	&-0.501	&-0.110  \\
$\hat{\beta}_1$   & -0.960 & -0.667 &-0.486	& -0.322 \\ \hline
$\kappa-1$        &-0.9	   &-0.7	&-0.5	&-0.3  \\ \hline
\end{tabular}\label{tab:fit_reg_var_rho_hat}
\end{table}

\begin{table}[H]
\caption{Shapiro-Wilk normality test on the distribution of $\hat{\rho}_{\tilde{H}, \tilde{T}}$ under the linear growth and power-law asymptotic frameworks.
The numbers shown in the table are $p$-values. The symbol $\ast$ indicates rejection of null hypothesis of normality. Model parameters: $d=0$, $\mu =0.000001419188$, $\lambda = 128.2085$, $c=1$, and $\sigma_e=0.0007289$.}
\center
\begin{tabular}{|cccc|ccccc|c|}\hline
\multicolumn{4}{|c}{$\tilde{H}=\tilde{T}^{\kappa}$} & \multicolumn{5}{|c|}{$\tilde{H}=\theta\tilde{T}$} &  \\ \hline \hline
$\kappa=0.1$ & $\kappa=0.3$ & $\kappa=0.5$ & $\kappa=0.7$ &  $\theta=0.0125$   & $\theta=0.025$  &$\theta=0.05$ &$\theta=0.125$ & $\theta=0.25$ & $\tilde{T}$ \\ \hline
    0.489  &  0.708        &  0.058         &$ {< 0.05}^\ast$    &0.489    &0.941    &$0.047^\ast$   & ${< 0.05}^\ast$    &    ${< 0.05}^\ast$ &   131  \\
    0.360  &  0.667        &  $0.008^\ast$  &${< 0.05}^\ast$     &0.344    &0.772    &$0.031^\ast$   & ${< 0.05}^\ast$    &    ${< 0.05}^\ast$ &   262  \\
    0.964  &  0.457        &  0.183         &${< 0.05}^\ast$     &0.457    &0.195    &0.209          &${< 0.05}^*$        &${< 0.05}^*$        &   525  \\
    0.852  &  0.597        &  0.537         &$0.008^*$           &0.606    &0.454    &0.943          &${< 0.05}^*$        &${< 0.05}^*$        &   1048  \\
    0.083  &  $0.032^\ast$ &  0.085         &${< 0.05}^\ast$     &0.310    &0.108    &0.296          &${< 0.05}^*$        &${< 0.05}^*$        &   2097  \\
    0.839  &  0.346        &  0.902         &$0.005^*$          &0.949     &0.624    &0.184          &${< 0.05}^*$        &${< 0.05}^*$        &   4194  \\ \hline
\end{tabular}\label{tab:k-s_normality_test}
\end{table}

\begin{figure}[H]
\begin{subfigure}[b]{0.5\textwidth}
  \centering
  \includegraphics[width=1.0\linewidth]{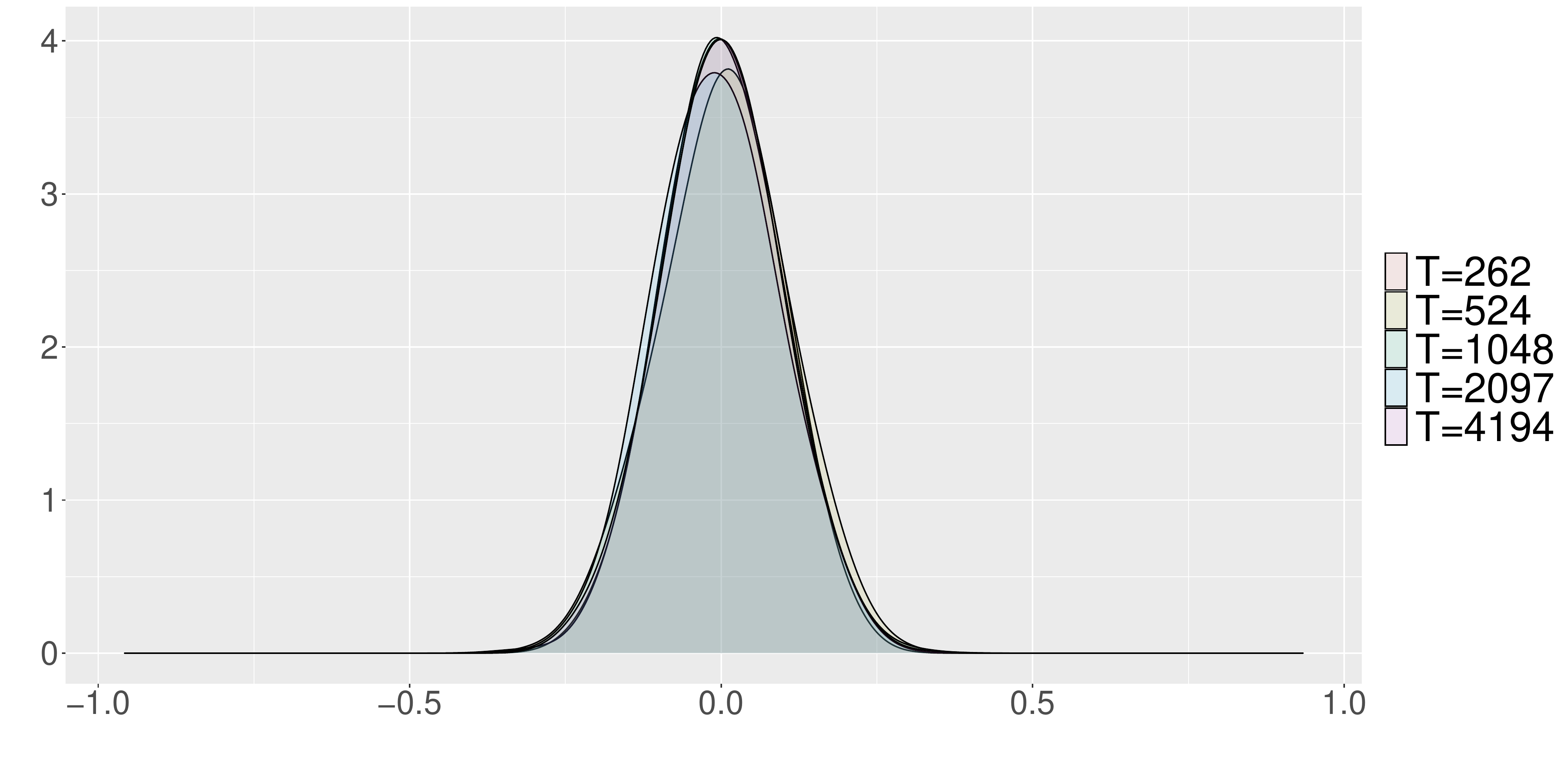}
  \caption{$\tilde{H} = 0.0125\tilde{T}$}
  \label{fig:sfig1}
\end{subfigure}
\hspace{8pt}
\begin{subfigure}[b]{0.5\textwidth}
  \centering
  \includegraphics[width=1.0\linewidth]{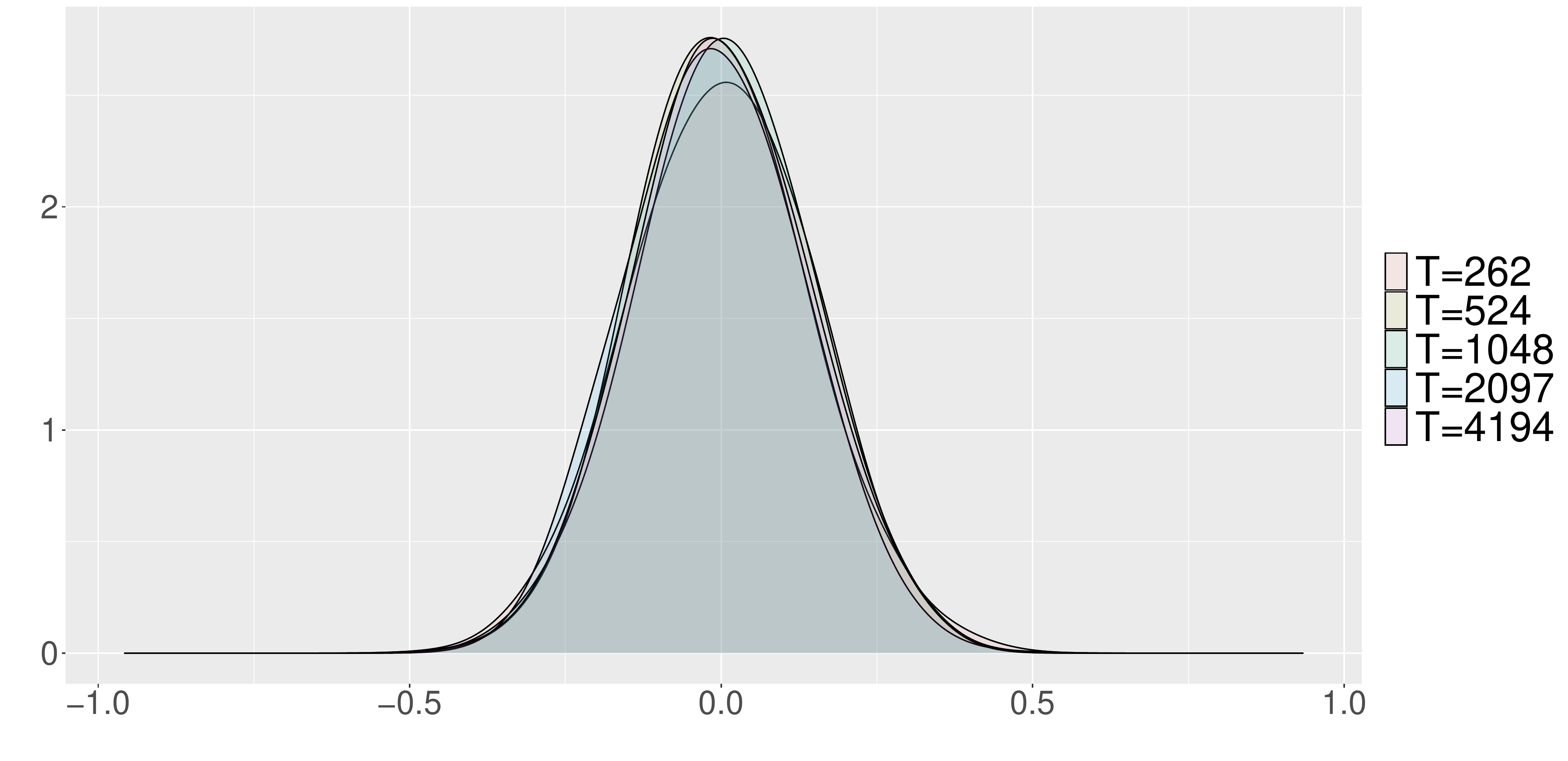}
  \caption{$\tilde{H} = 0.025\tilde{T}$}
  \label{fig:sfig1}
\end{subfigure}
\hspace{8pt}
\begin{subfigure}[b]{0.5\textwidth}
  \centering
  \includegraphics[width=1.0\linewidth]{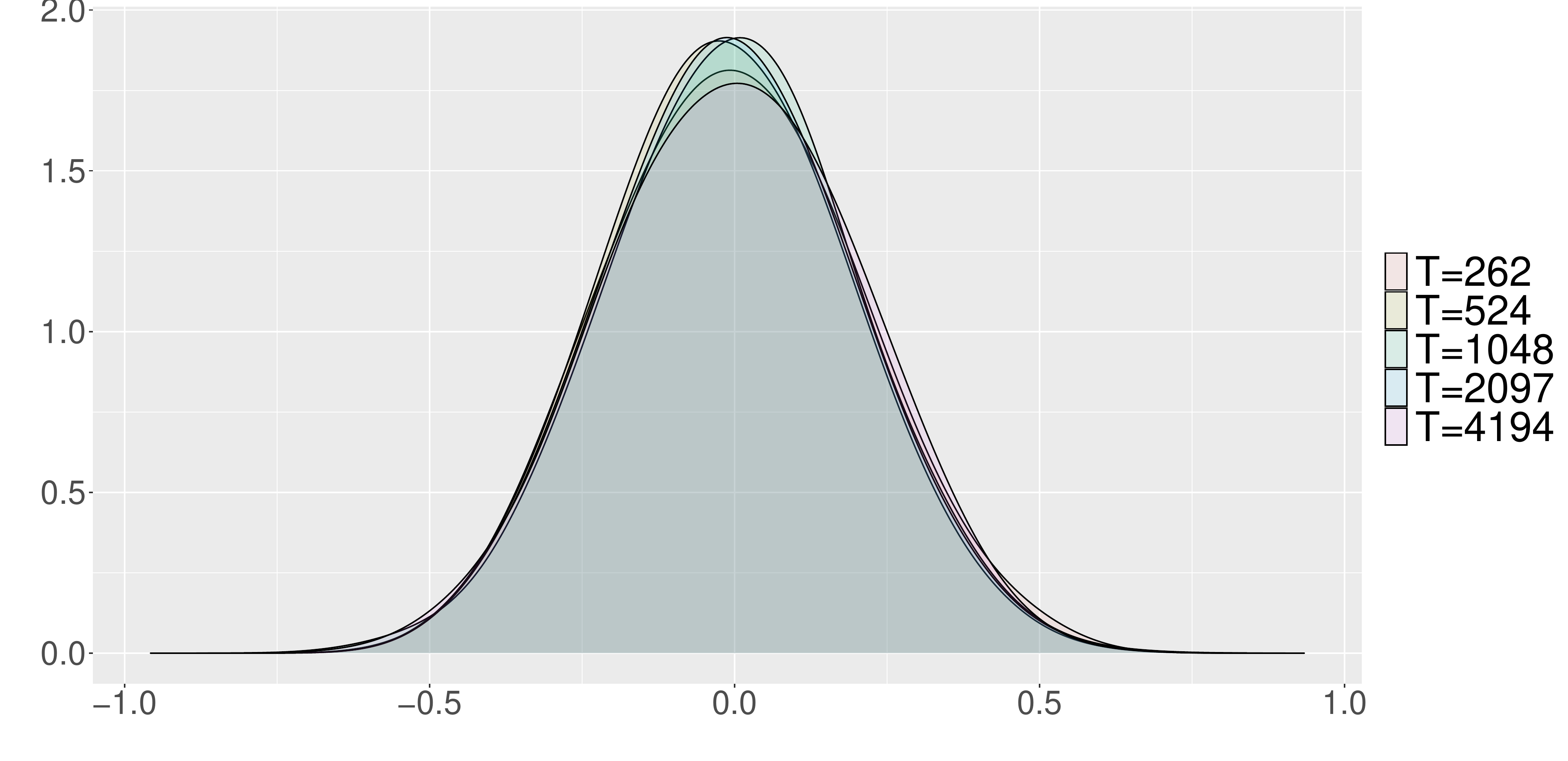}
  \caption{$\tilde{H} = 0.05\tilde{T}$}
  \label{fig:sfig2}
\end{subfigure}
\hspace{8pt}
\begin{subfigure}[b]{0.5\textwidth}
  \centering
  \includegraphics[width=1.0\linewidth]{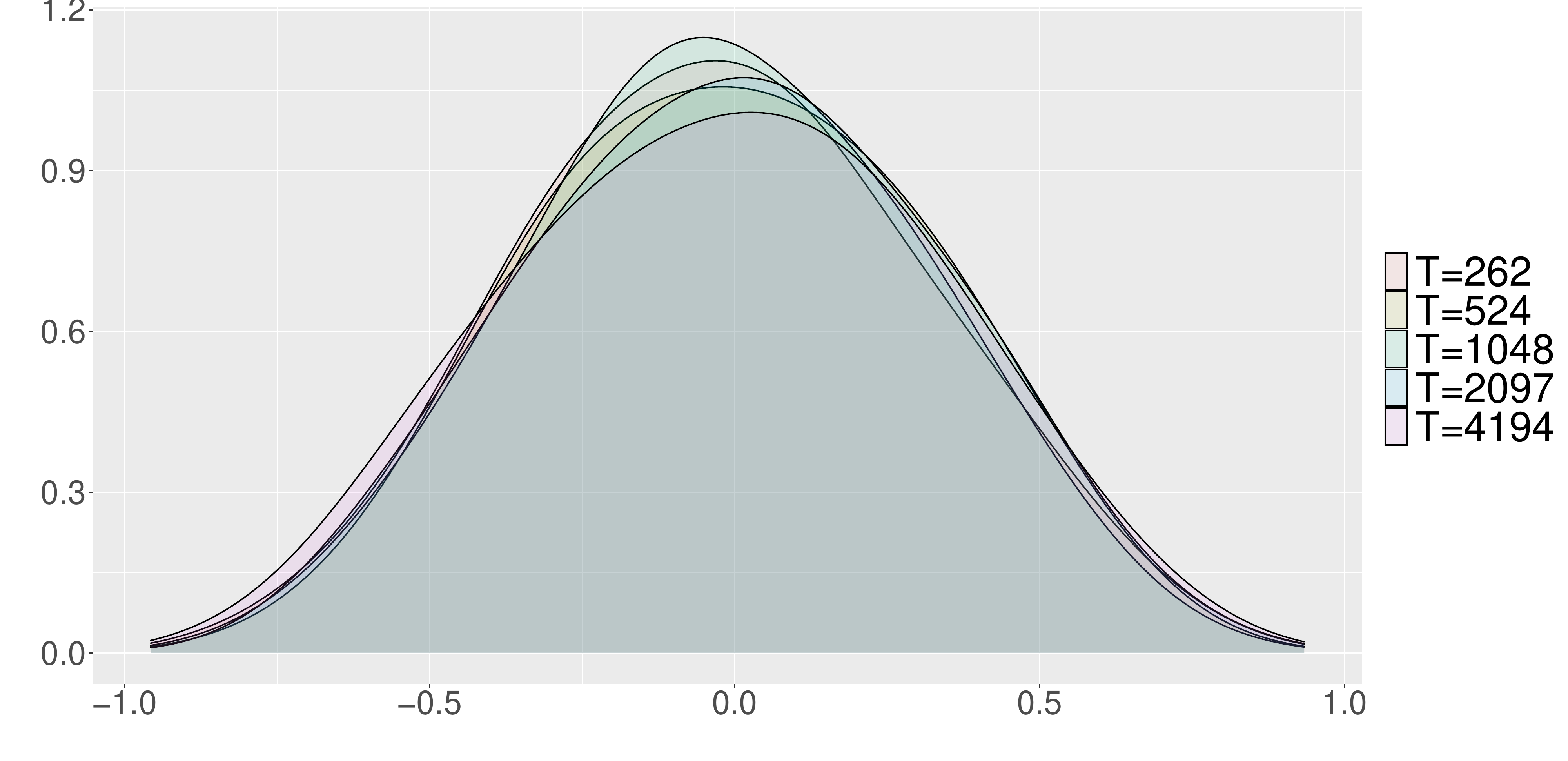}
  \caption{$\tilde{H} = 0.125\tilde{T}$}
  \label{fig:sfig2}
\end{subfigure}
\hspace{8pt}
\begin{subfigure}[b]{0.5\textwidth}
  \centering
  \includegraphics[width=1.0\linewidth]{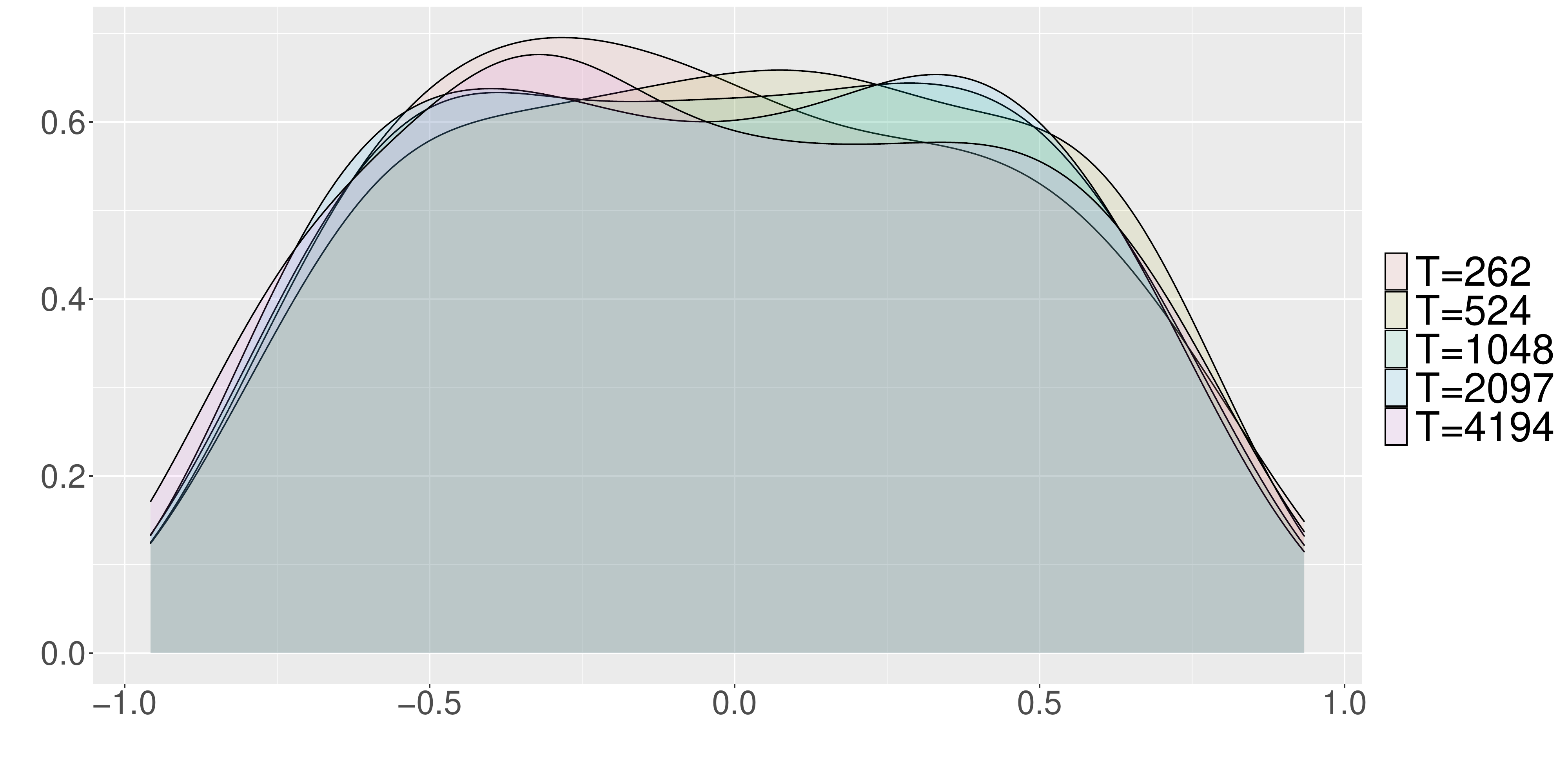}
  \caption{$\tilde{H} = 0.25\tilde{T}$}
  \label{fig:sfig2}
\end{subfigure}
\caption{Sampling distribution of 1000 realizations of $\hat{\rho}_{\tilde{H}, \tilde{T}}$ using $\tilde{H}=\theta T$ when $d=0$. Model parameters: $\mu =0.000001419188$, $\lambda = 128.2085$, $c=1$, and $\sigma_e=0.0007289$.}\label{fig:asym_dist_rho_hat_Sizova_d0}
\end{figure}

\begin{figure}[H]
\begin{subfigure}[b]{0.5\textwidth}
  \centering
  \includegraphics[width=1.0\linewidth]{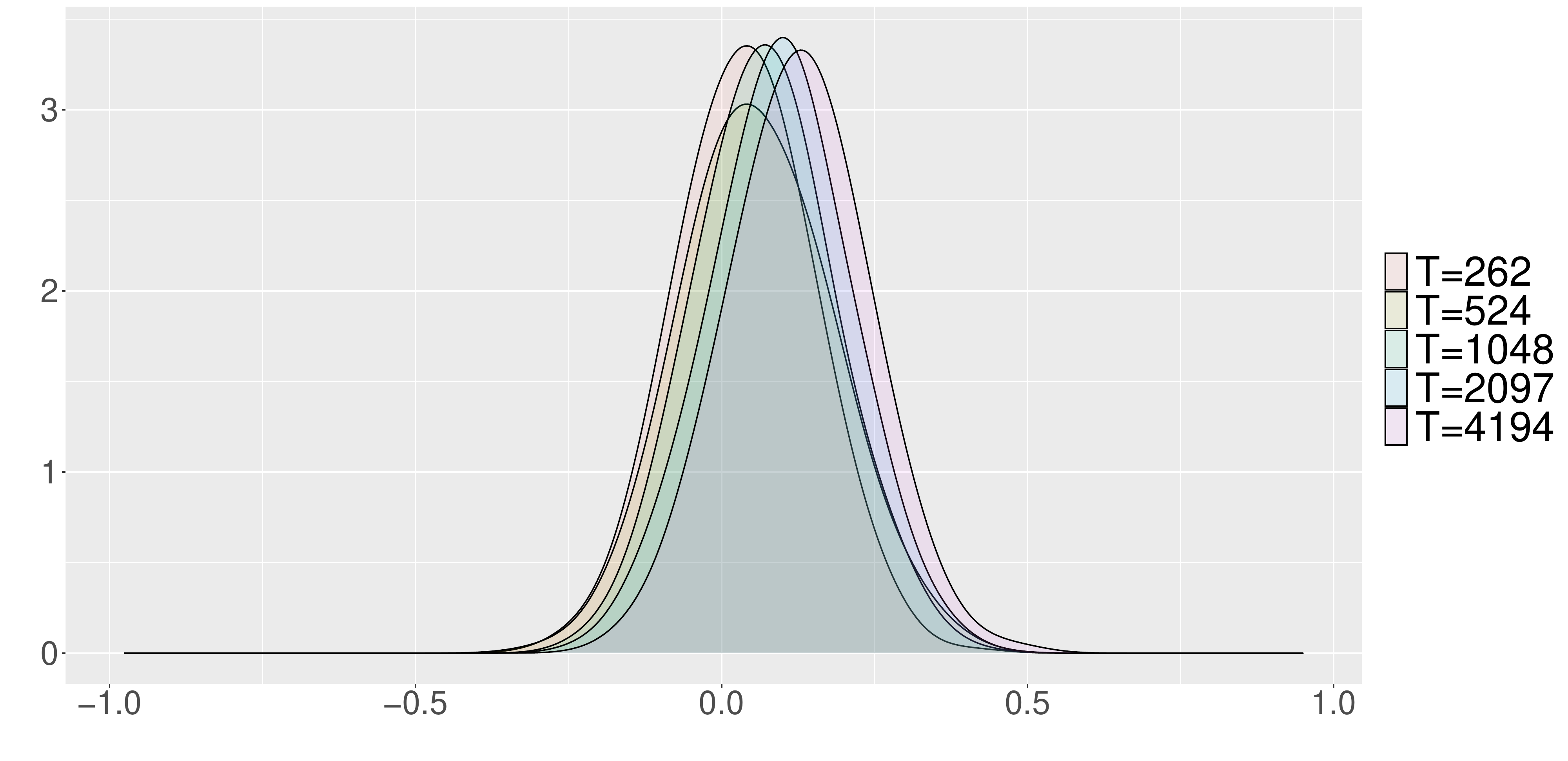}
  \caption{$\tilde{H} = 0.0125\tilde{T}$}
  \label{fig:sfig1}
\end{subfigure}
\hspace{8pt}
\begin{subfigure}[b]{0.5\textwidth}
  \centering
  \includegraphics[width=1.0\linewidth]{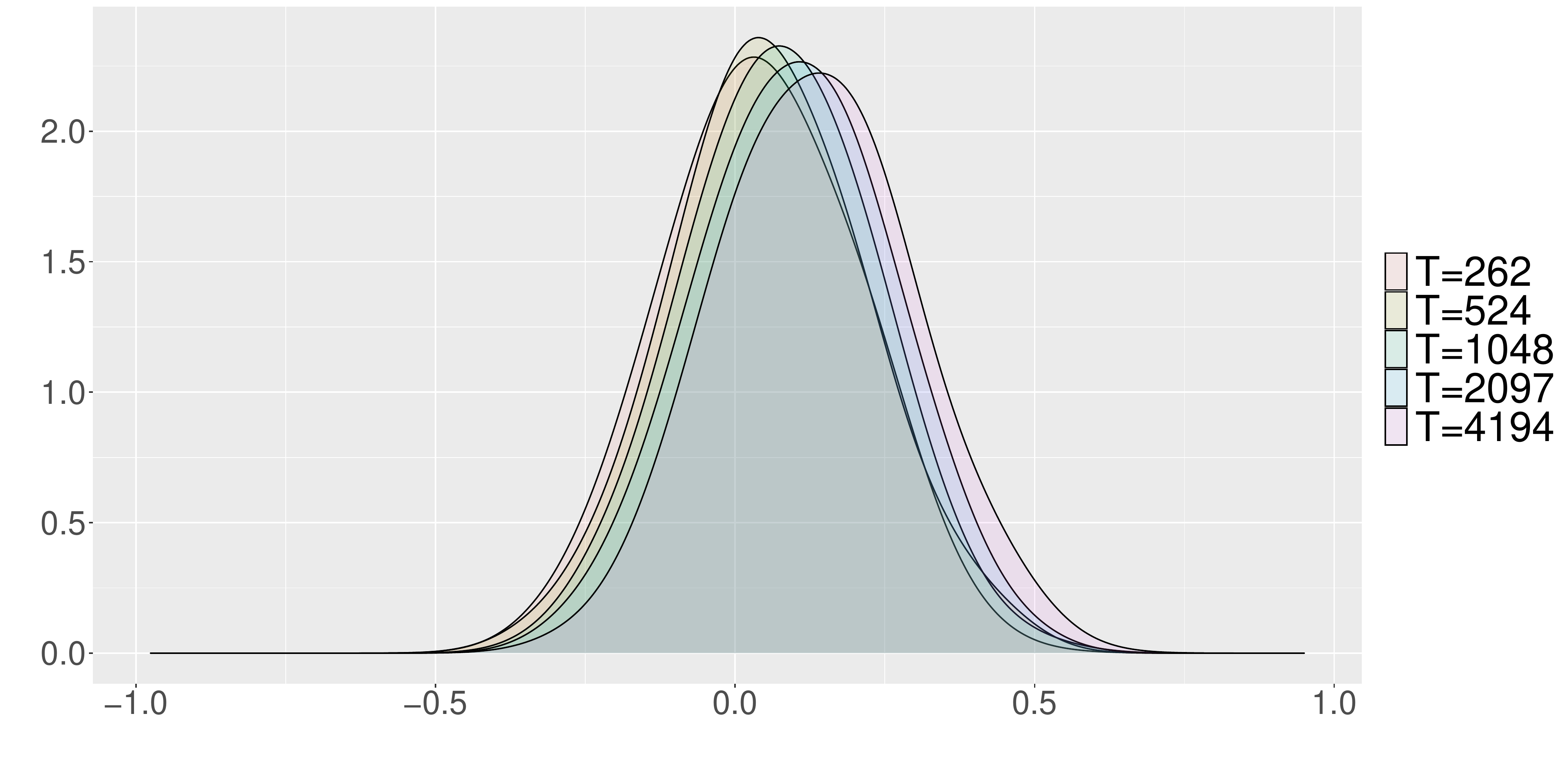}
  \caption{$\tilde{H} = 0.025\tilde{T}$}
  \label{fig:sfig1}
\end{subfigure}
\hspace{8pt}
\begin{subfigure}[b]{0.5\textwidth}
  \centering
  \includegraphics[width=1.0\linewidth]{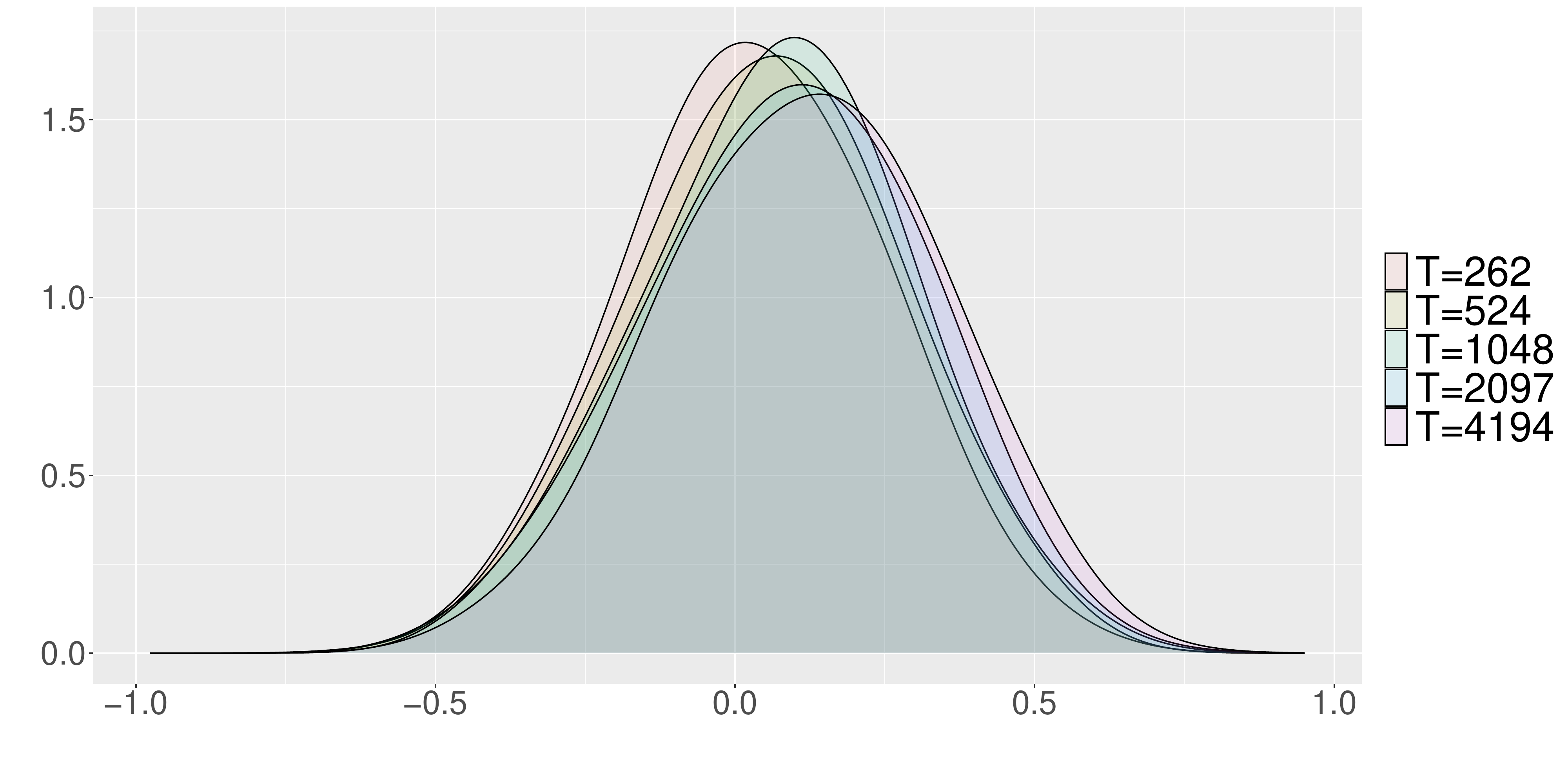}
  \caption{$\tilde{H} = 0.05\tilde{T}$}
  \label{fig:sfig2}
\end{subfigure}
\hspace{8pt}
\begin{subfigure}[b]{0.5\textwidth}
  \centering
  \includegraphics[width=1.0\linewidth]{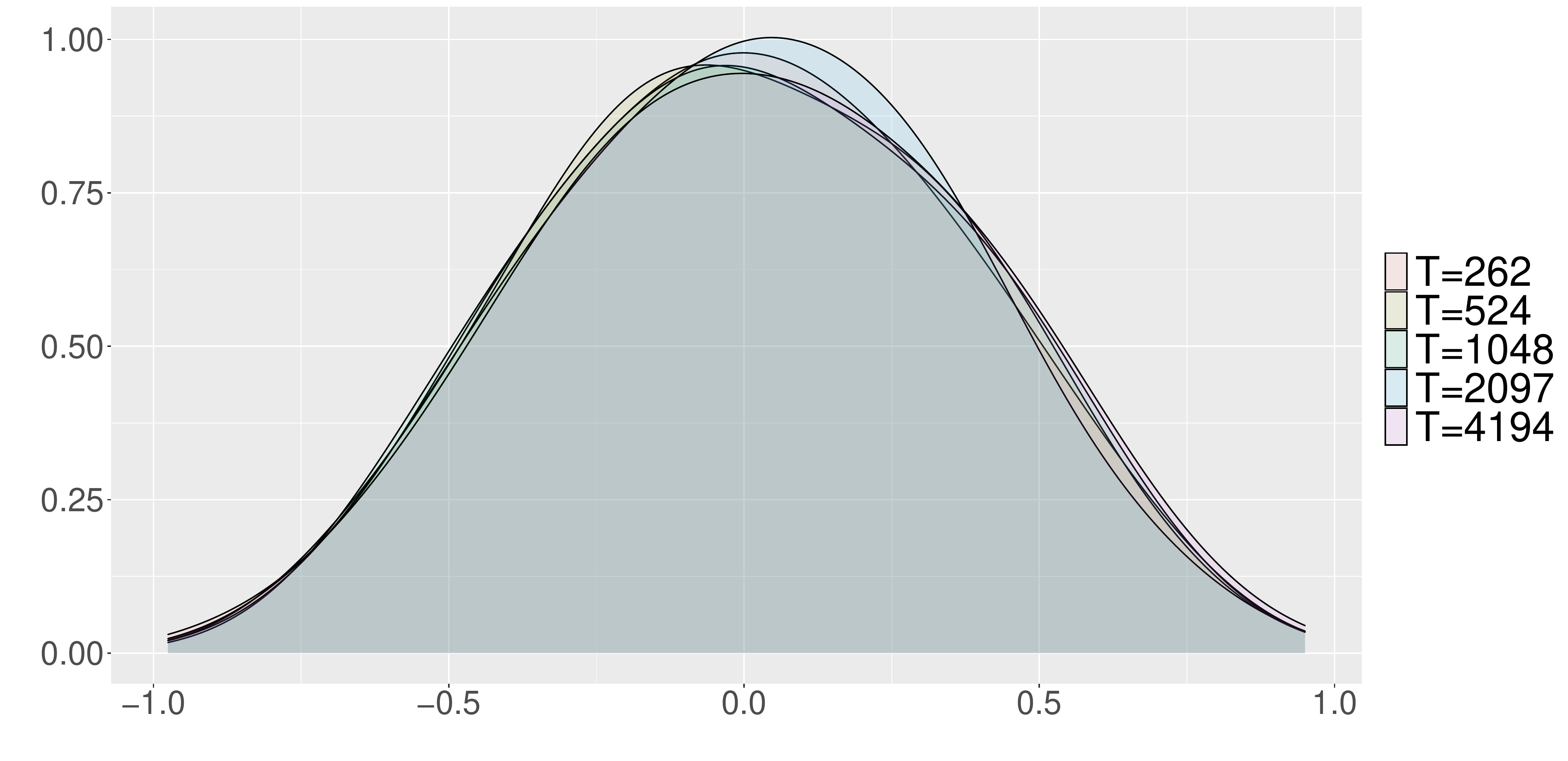}
  \caption{$\tilde{H} = 0.125\tilde{T}$}
  \label{fig:sfig2}
\end{subfigure}
\hspace{8pt}
\begin{subfigure}[b]{0.5\textwidth}
  \centering
  \includegraphics[width=1.0\linewidth]{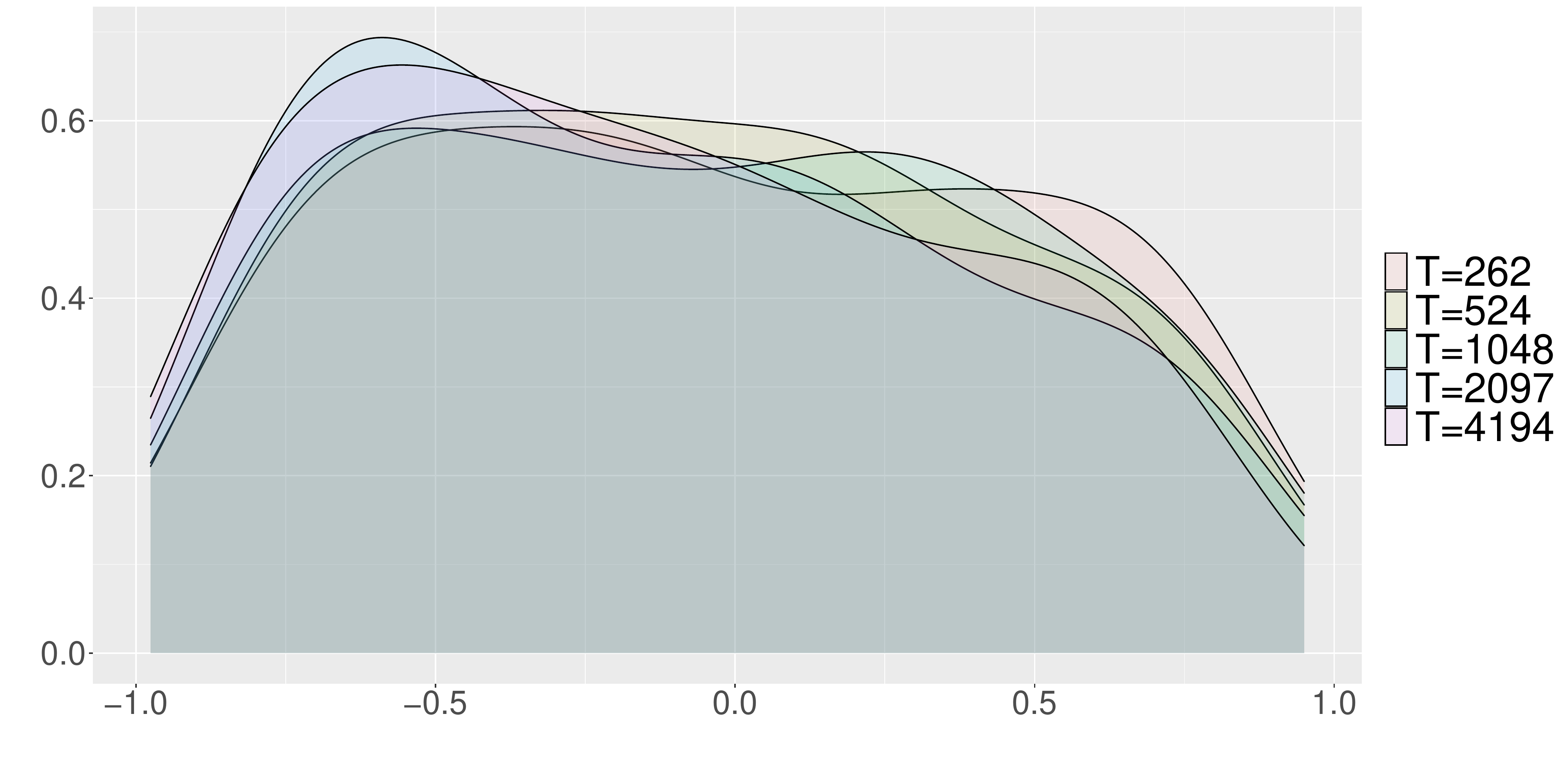}
  \caption{$\tilde{H} = 0.25\tilde{T}$}
  \label{fig:sfig2}
\end{subfigure}
\caption{Sampling distribution of 1000 realizations of $\hat{\rho}_{\tilde{H}, \tilde{T}}$ using $\tilde{H}=\theta T$ when $d=0.3545$. Model parameters: $\mu =0.000001419188$, $\lambda = 128.2085$, $c=1$, and $\sigma_e=0.0007289$.}\label{fig:asym_dist_rho_hat_Sizova_d035}
\end{figure}

\begin{figure}[H]
\begin{subfigure}[b]{0.5\textwidth}
  \centering
  \includegraphics[width=1.0\linewidth]{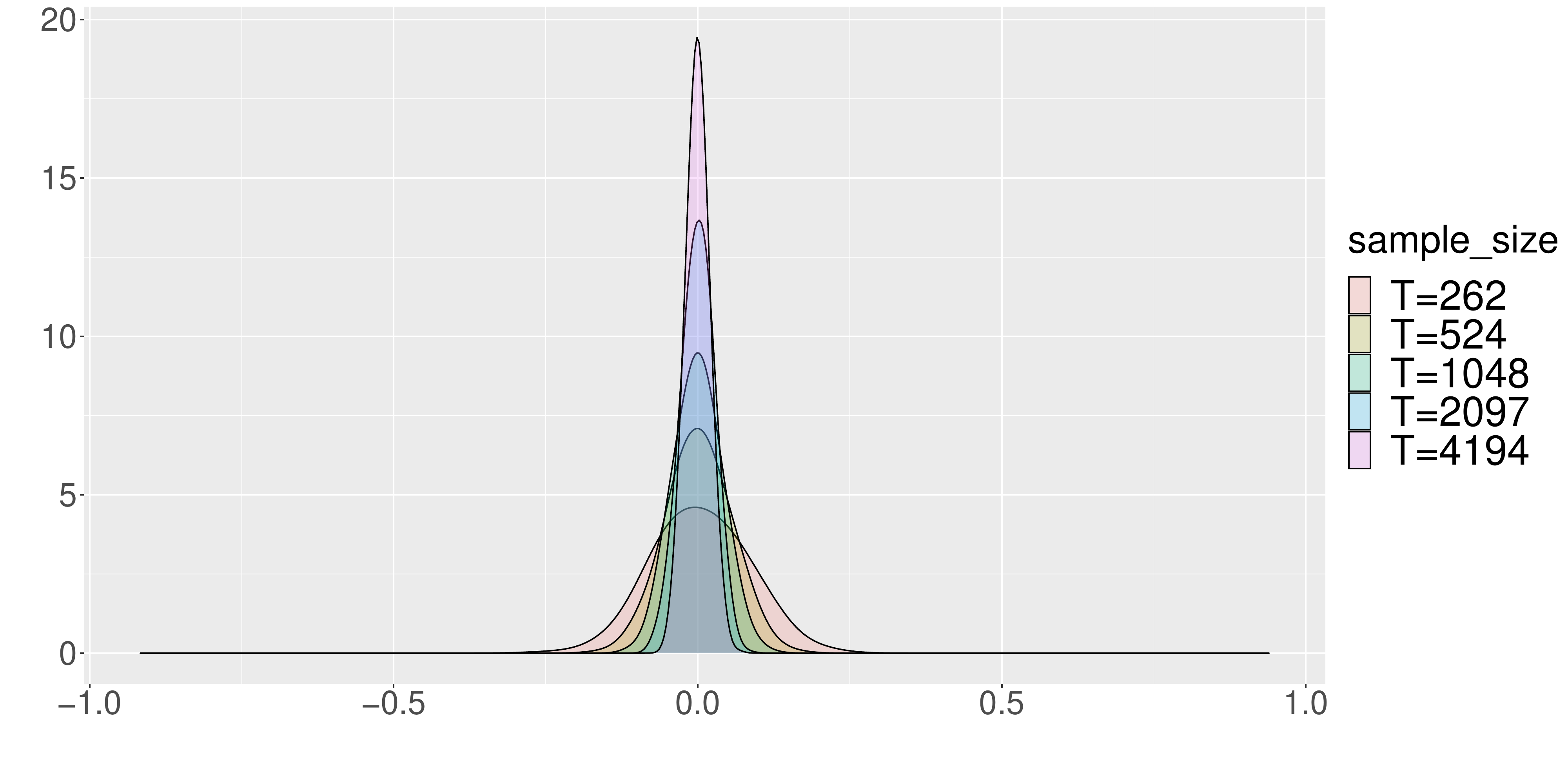}
  \caption{$\tilde{H} = \tilde{T}^{0.1}$}
  \label{fig:sfig1}
\end{subfigure}
\hspace{8pt}
\begin{subfigure}[b]{0.5\textwidth}
  \centering
  \includegraphics[width=1.0\linewidth]{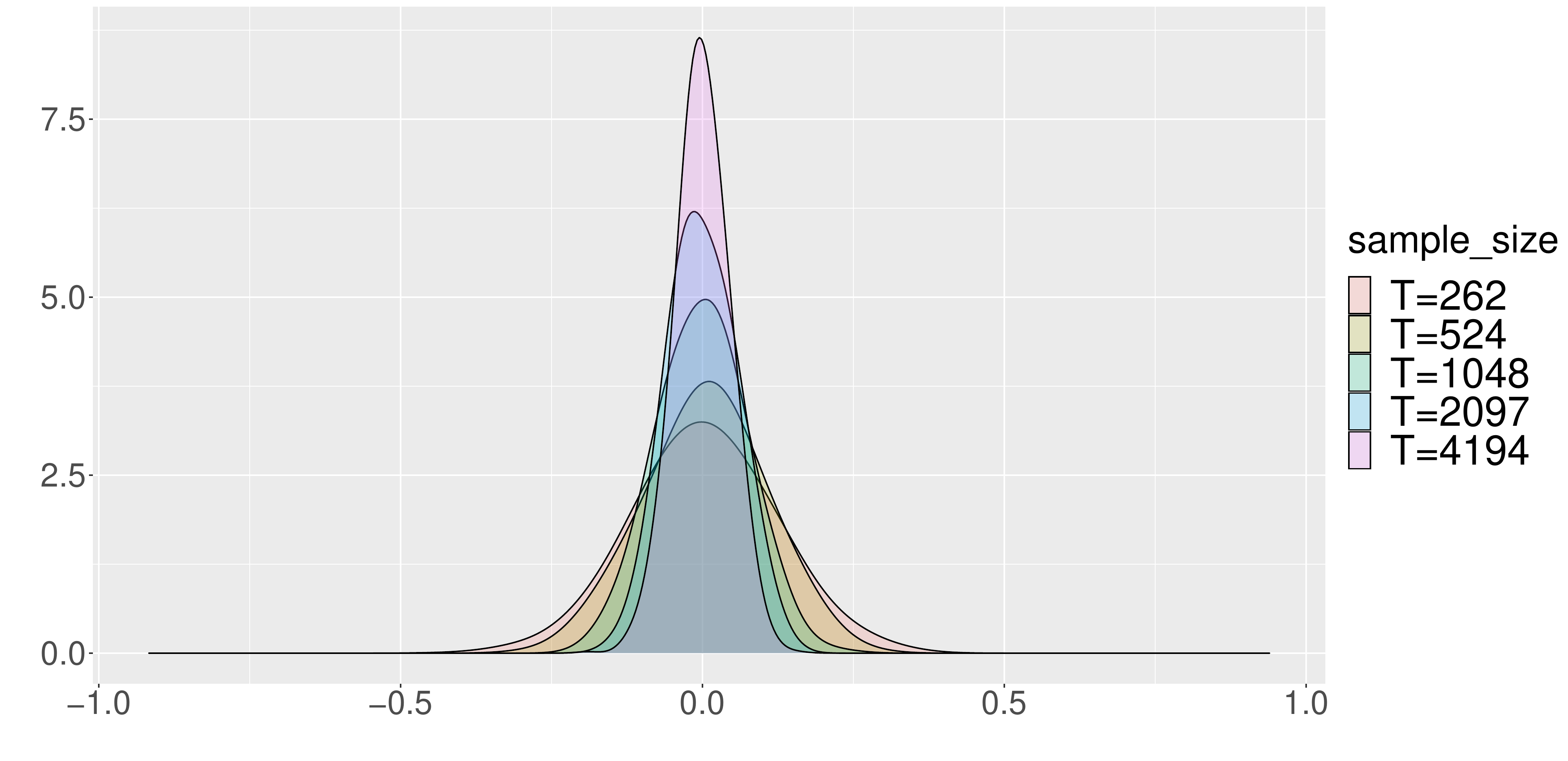}
  \caption{$\tilde{H} = \tilde{T}^{0.3}$}
  \label{fig:sfig2}
\end{subfigure}
\hspace{8pt}
\begin{subfigure}[b]{0.5\textwidth}
  \centering
  \includegraphics[width=1.0\linewidth]{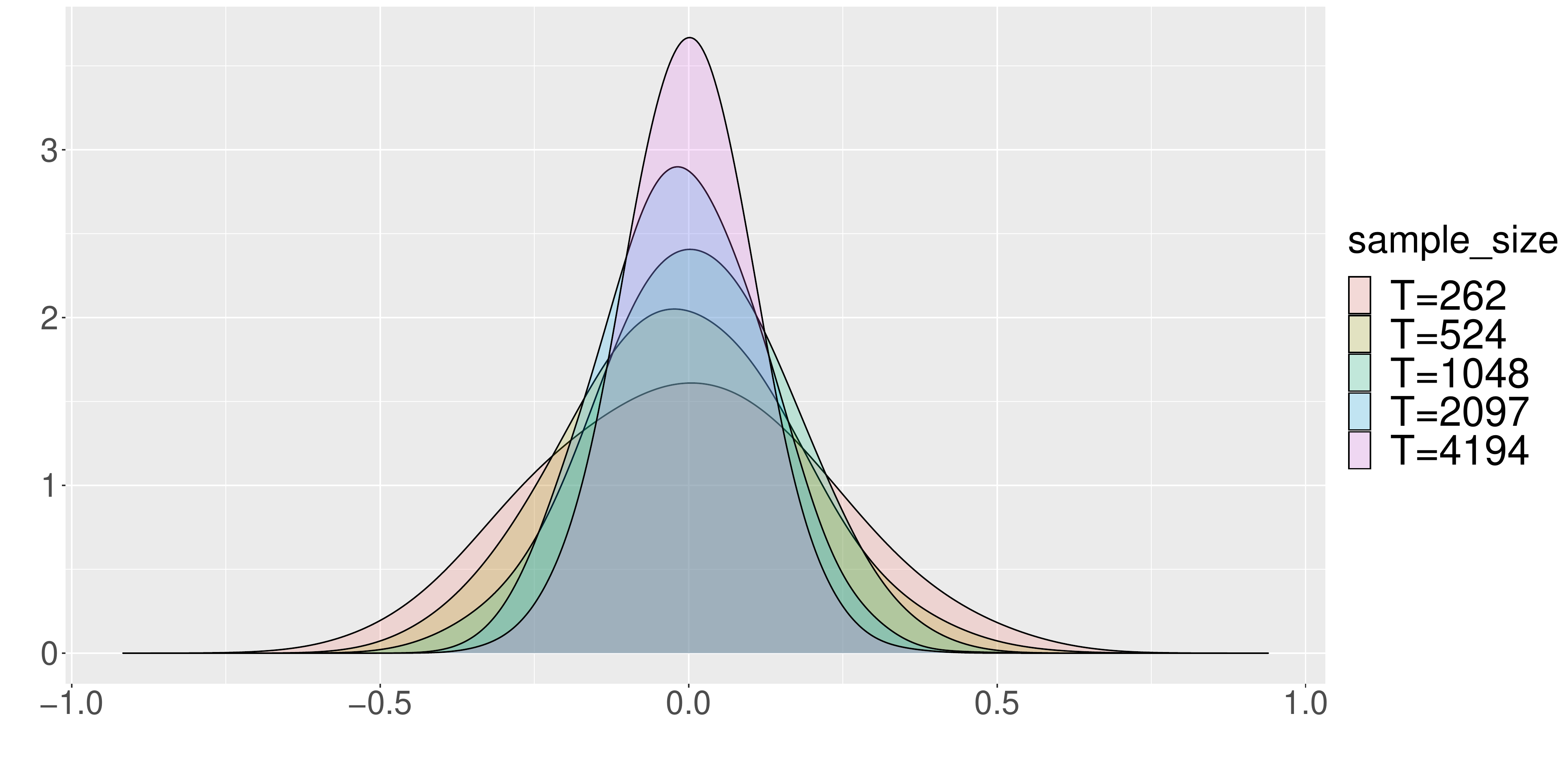}
  \caption{$\tilde{H} = \tilde{T}^{0.5}$}
  \label{fig:sfig2}
\end{subfigure}
\hspace{8pt}
\begin{subfigure}[b]{0.5\textwidth}
  \centering
  \includegraphics[width=1.0\linewidth]{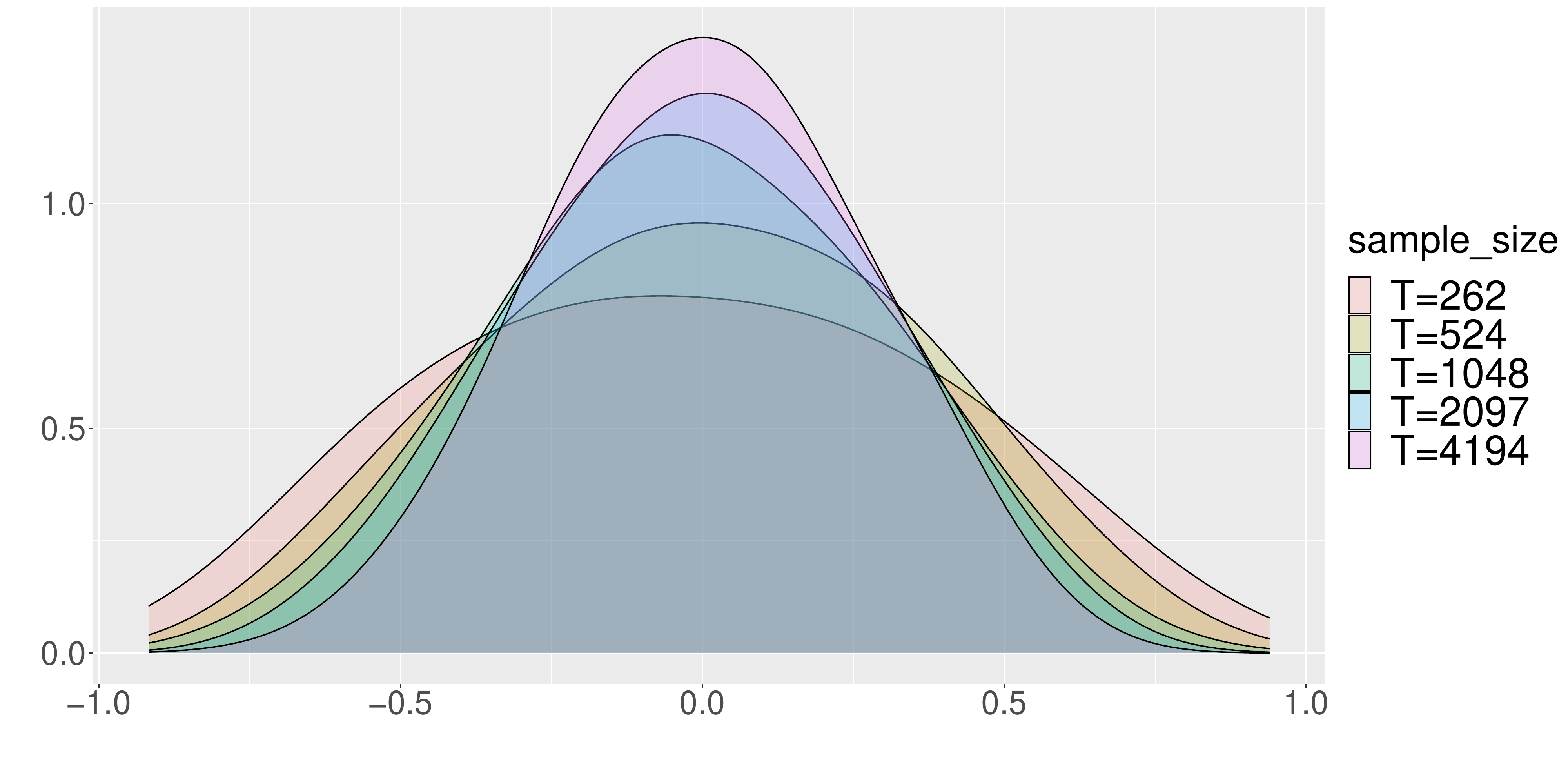}
  \caption{$\tilde{H} = \tilde{T}^{0.7}$}
  \label{fig:sfig2}
\end{subfigure}
\caption{Sampling distribution of 1000 realizations of $\hat{\rho}_{\tilde{H}, \tilde{T}}$ using $\tilde{H}=T^{\kappa}$, $\kappa=0.1, 0.3, 0.5, 0.7$ when $d=0$. Model parameters: $\mu =0.000001419188$, $\lambda = 128.2085$, $c=1$, and $\sigma_e=0.0007289$.} \label{fig:asym_dist_rho_hat_power_law_d0}
\end{figure}

\begin{figure}[H]
\begin{subfigure}[b]{0.5\textwidth}
  \centering
  \includegraphics[width=1.0\linewidth]{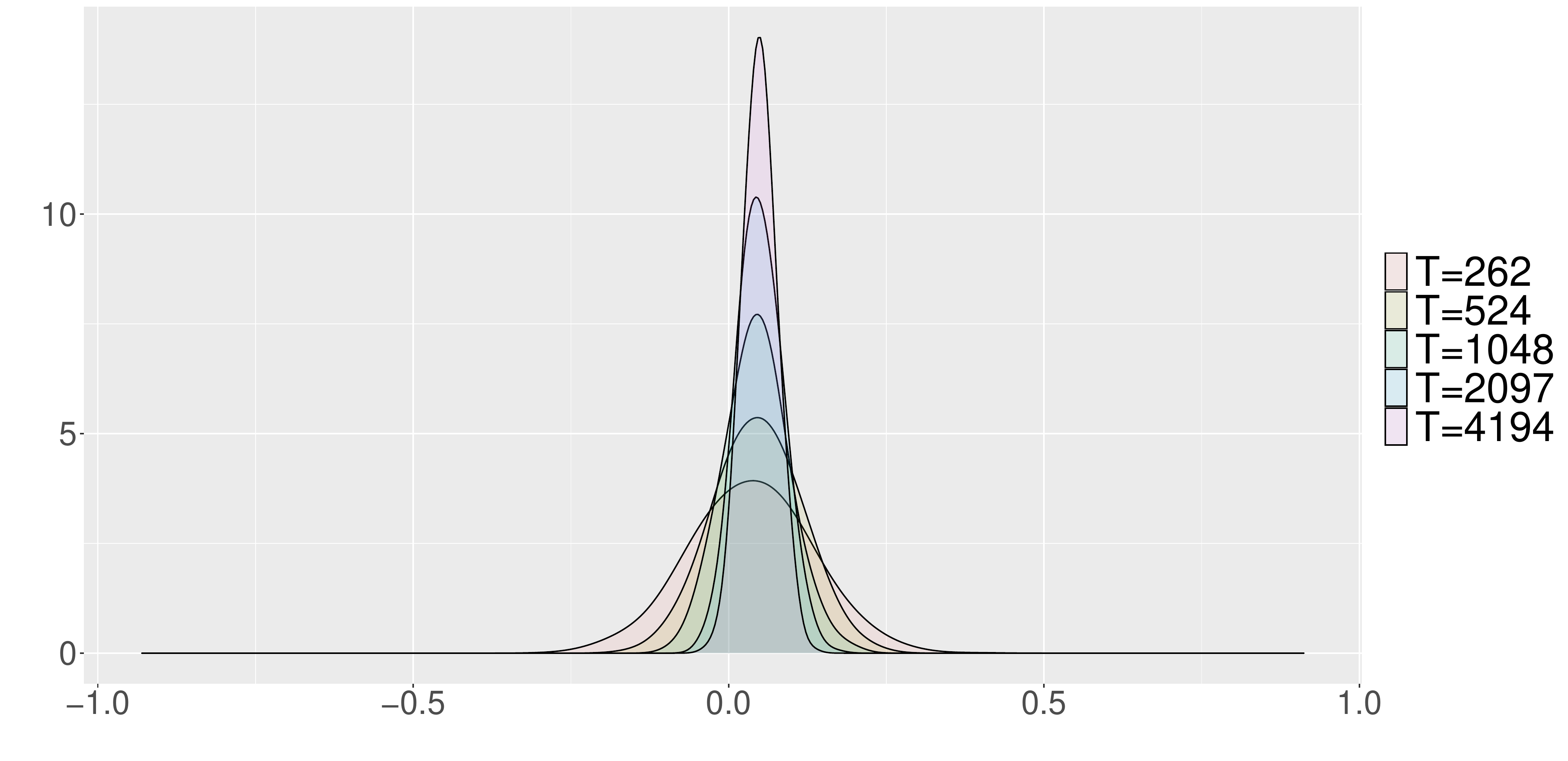}
  \caption{$\tilde{H} = \tilde{T}^{0.1}$}
  \label{fig:sfig1}
\end{subfigure}
\hspace{8pt}
\begin{subfigure}[b]{0.5\textwidth}
  \centering
  \includegraphics[width=1.0\linewidth]{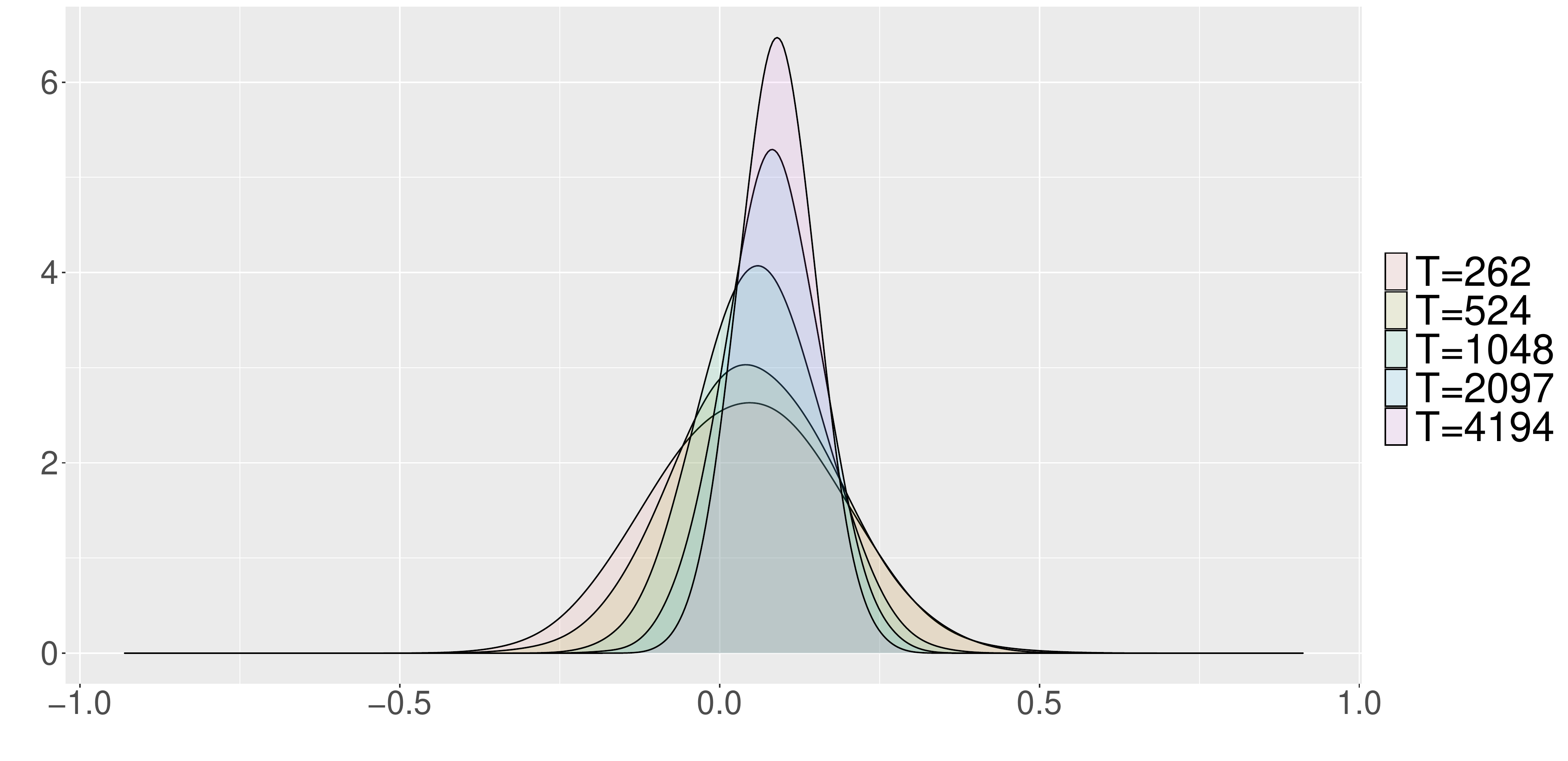}
  \caption{$\tilde{H} = \tilde{T}^{0.3}$}
  \label{fig:sfig2}
\end{subfigure}
\hspace{8pt}
\begin{subfigure}[b]{0.5\textwidth}
  \centering
  \includegraphics[width=1.0\linewidth]{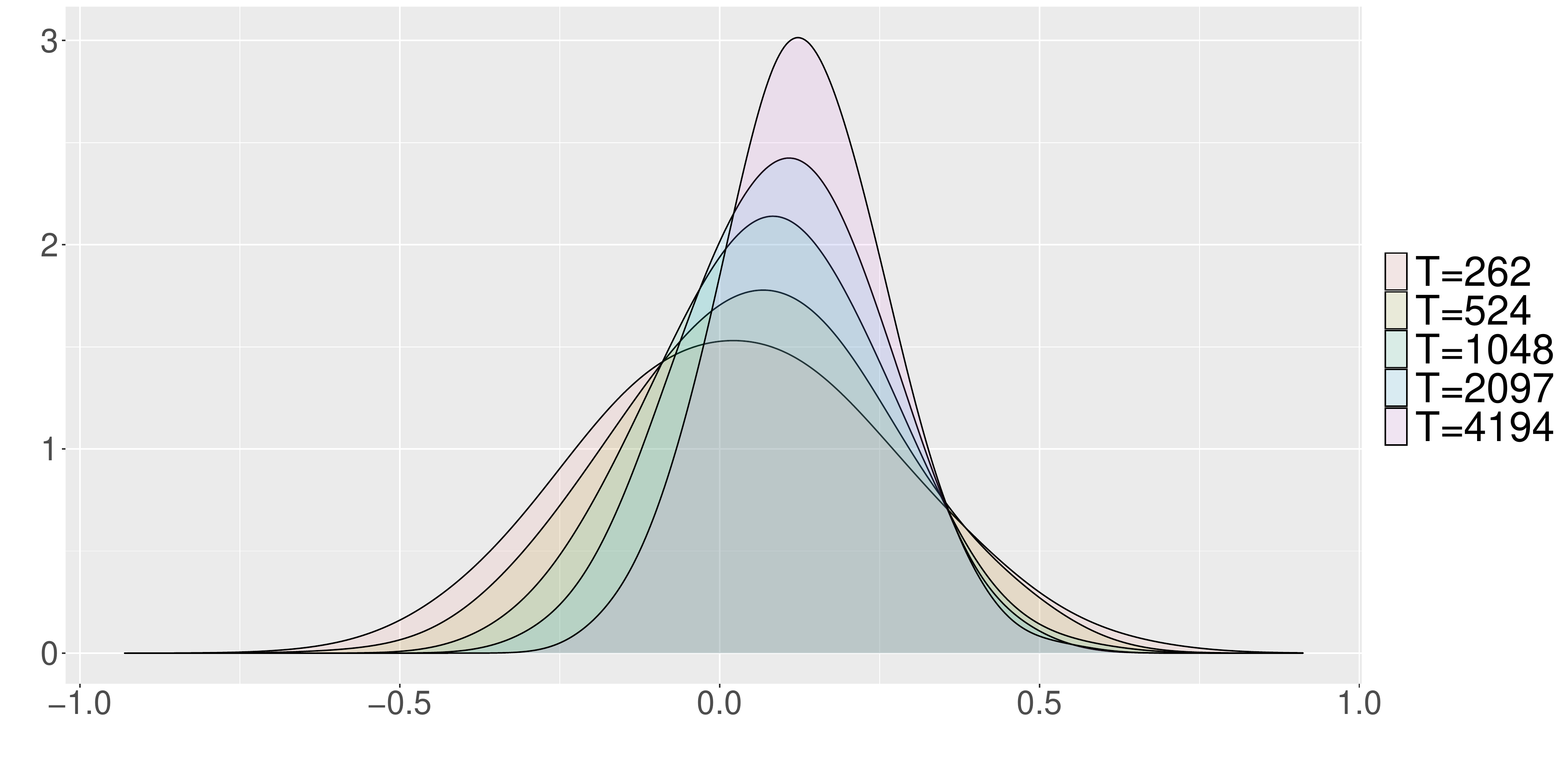}
  \caption{$\tilde{H} = \tilde{T}^{0.5}$}
  \label{fig:sfig2}
\end{subfigure}
\hspace{8pt}
\begin{subfigure}[b]{0.5\textwidth}
  \centering
  \includegraphics[width=1.0\linewidth]{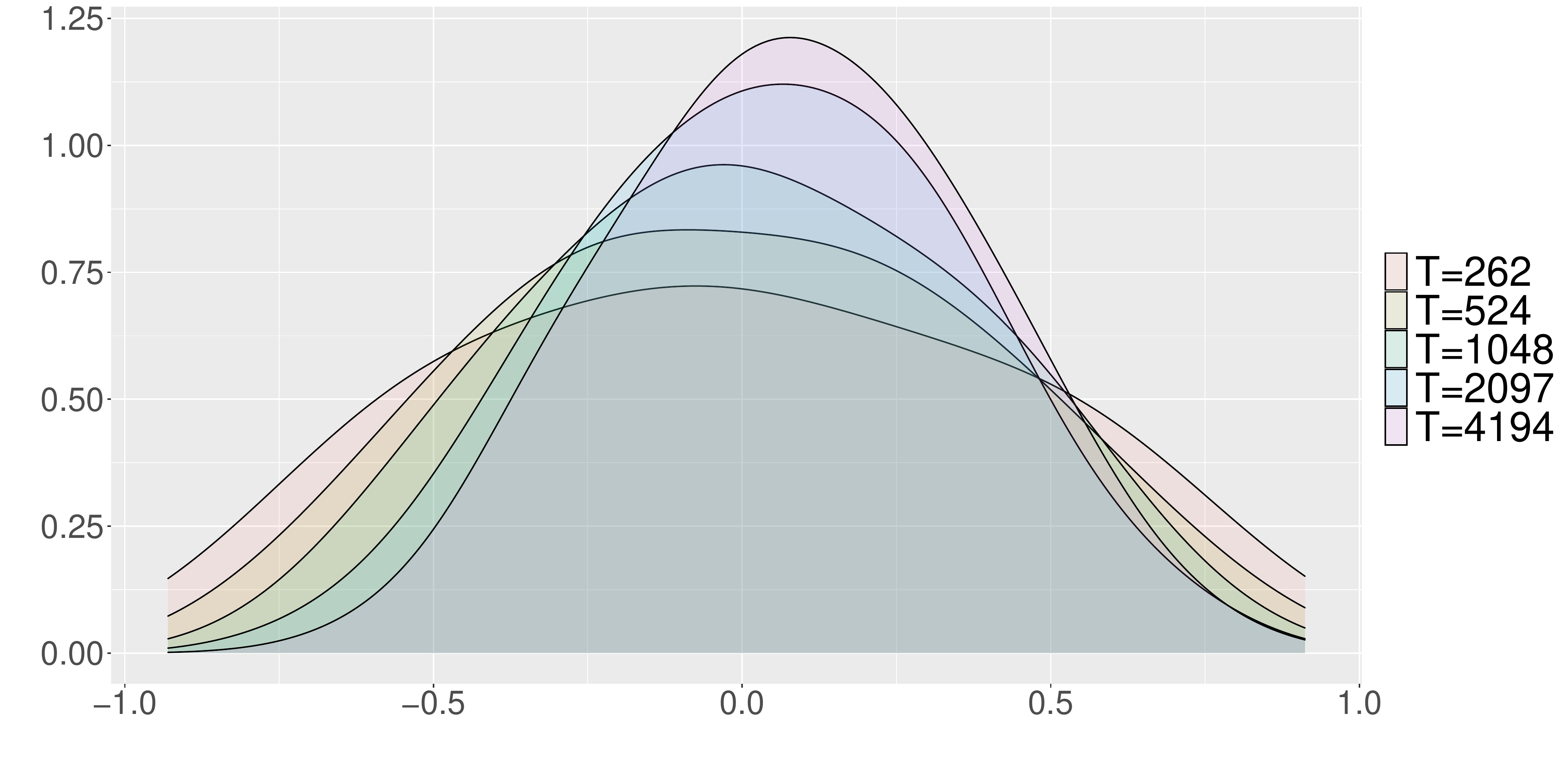}
  \caption{$\tilde{H} = \tilde{T}^{0.7}$}
  \label{fig:sfig2}
\end{subfigure}
\caption{Sampling distribution of 1000 realizations of $\hat{\rho}_{\tilde{H}, \tilde{T}}$ using $\tilde{H}=T^{\kappa}$, $\kappa=0.1, 0.3, 0.5, 0.7$ when $d=0.3545$. Model parameters: $\mu =0.000001419188$, $\lambda = 128.2085$, $c=1$, and $\sigma_e=0.0007289$.} \label{fig:sample_dist_rho_hat_power_law_d035}
\end{figure}

\begin{figure}[H]
\begin{subfigure}[b]{0.5\textwidth}
  \centering
  \includegraphics[width=1.0\linewidth]{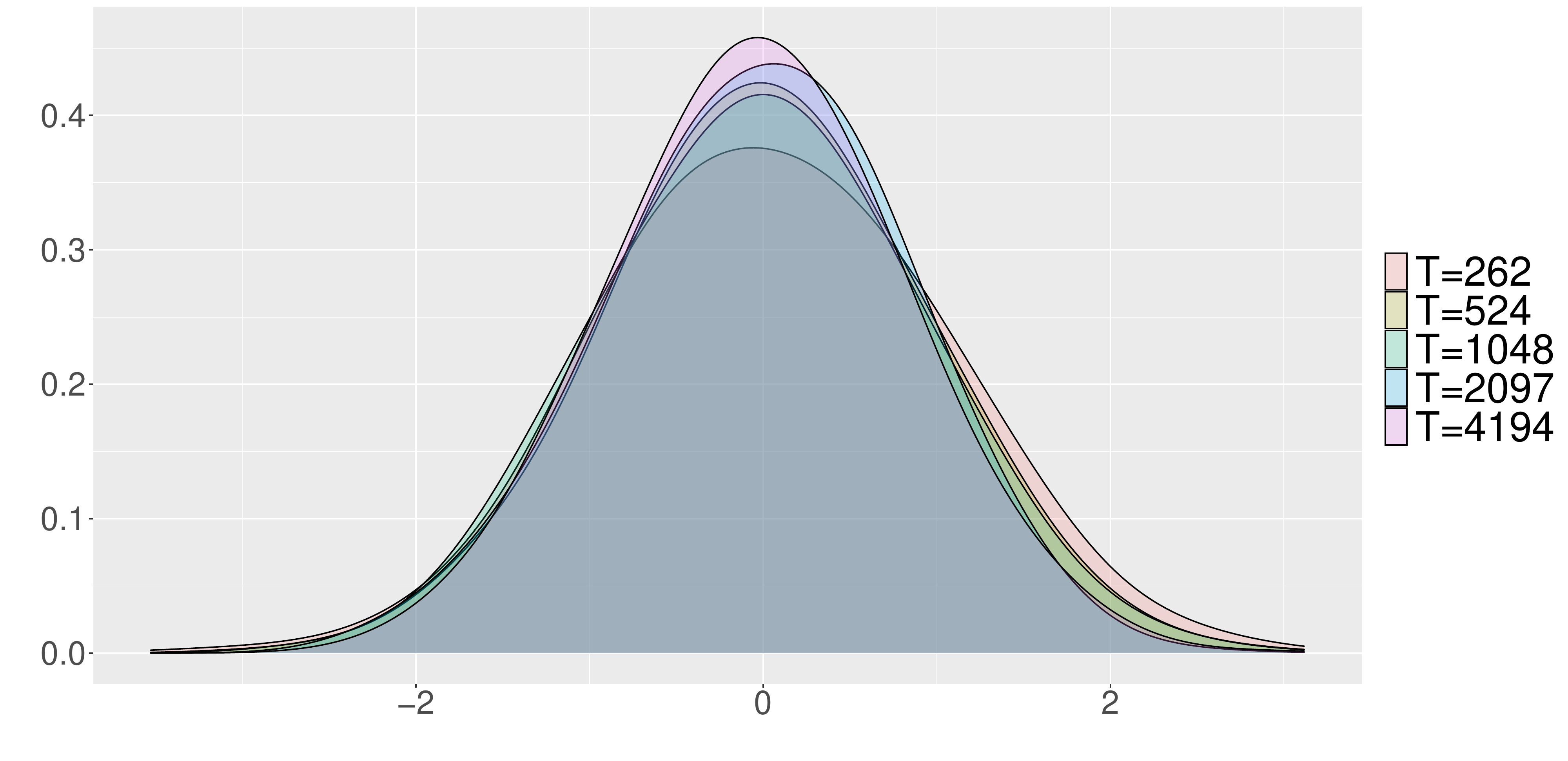}
  \caption{$\tilde{H} = \tilde{T}^{0.1}$}
  \label{fig:sfig1}
\end{subfigure}
\hspace{8pt}
\begin{subfigure}[b]{0.5\textwidth}
  \centering
  \includegraphics[width=1.0\linewidth]{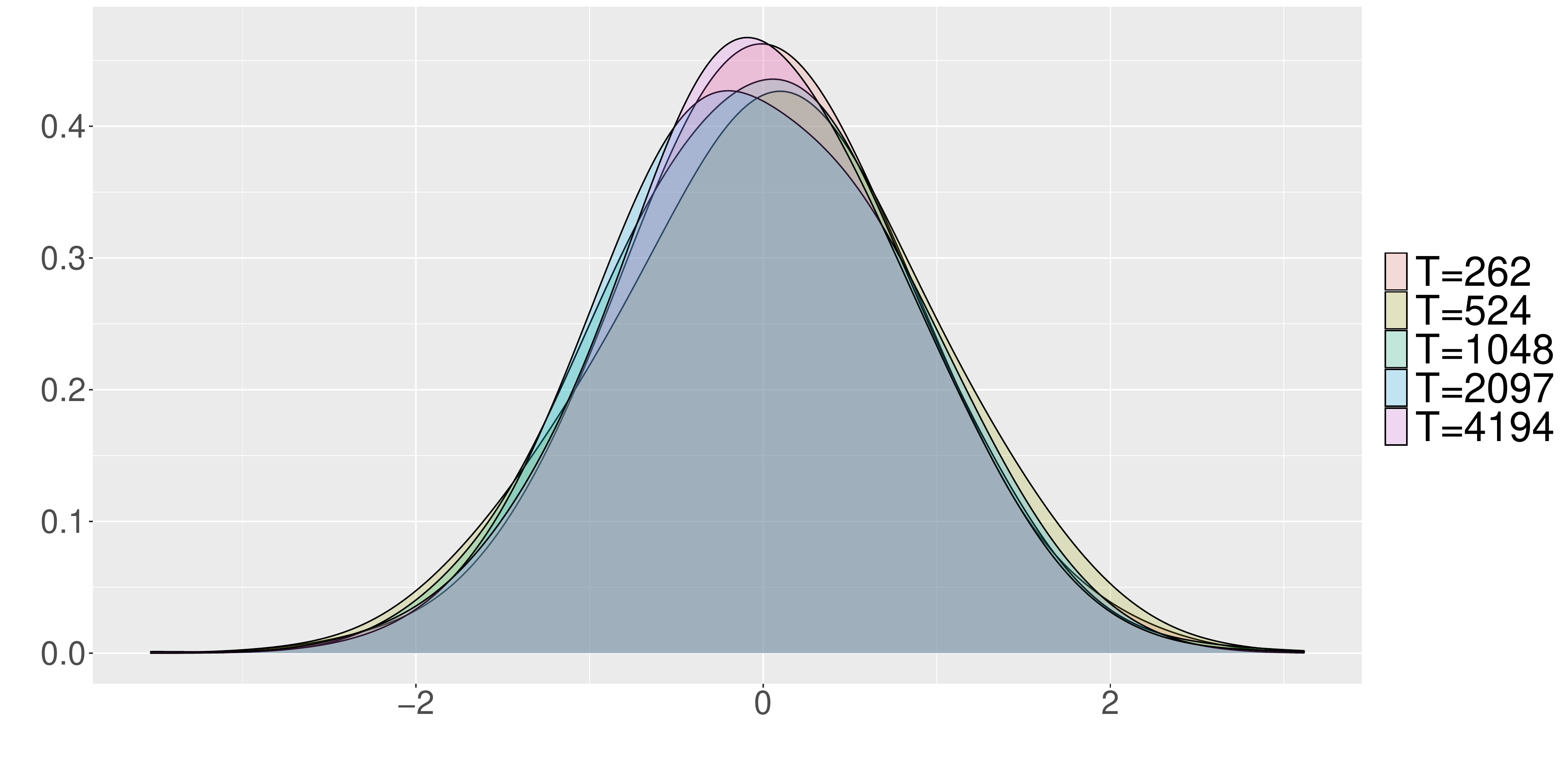}
  \caption{$\tilde{H} = \tilde{T}^{0.3}$}
  \label{fig:sfig2}
\end{subfigure}
\hspace{8pt}
\begin{subfigure}[b]{0.5\textwidth}
  \centering
  \includegraphics[width=1.0\linewidth]{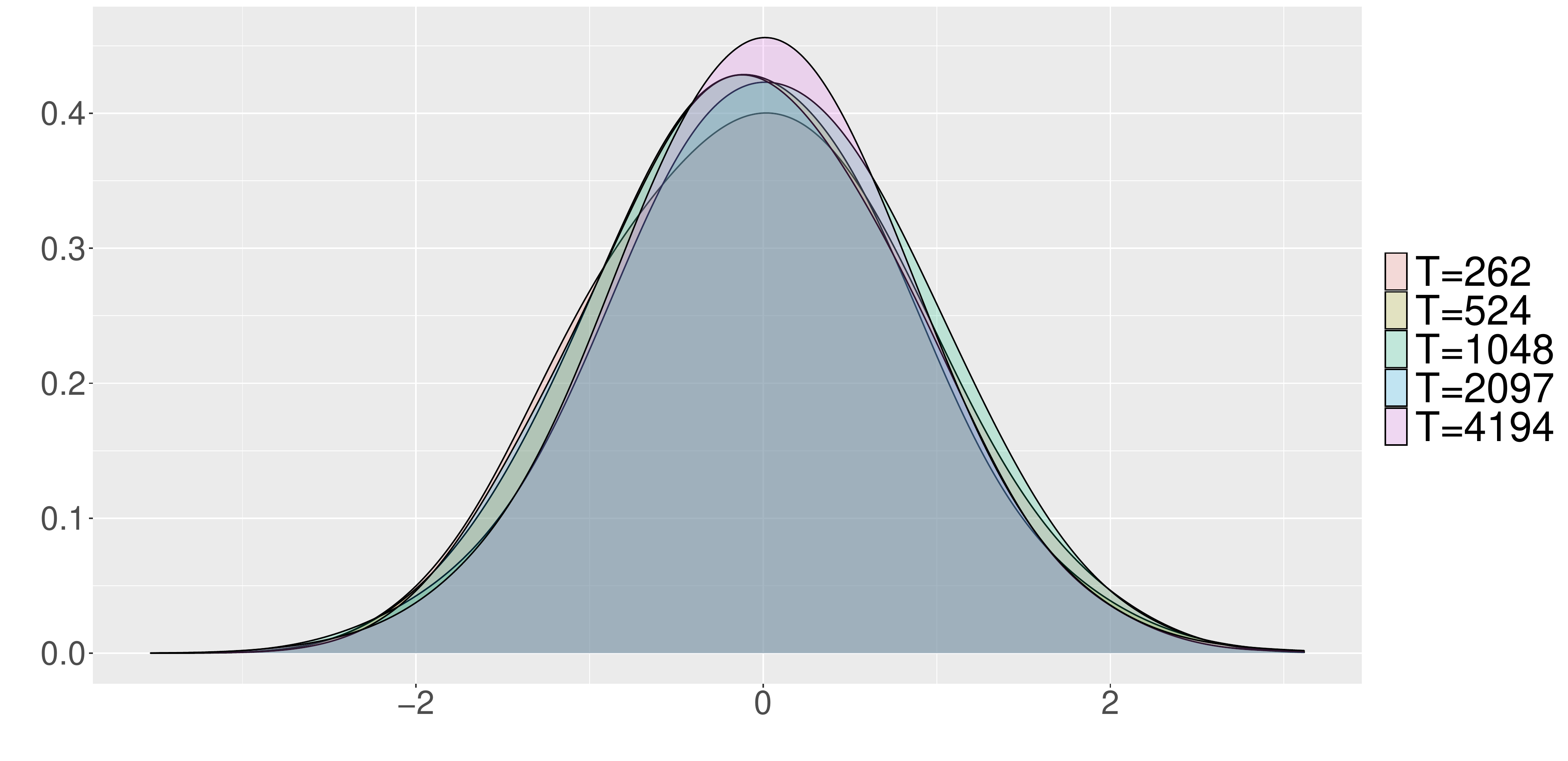}
  \caption{$\tilde{H} = \tilde{T}^{0.5}$}
  \label{fig:sfig2}
\end{subfigure}
\hspace{8pt}
\begin{subfigure}[b]{0.5\textwidth}
  \centering
  \includegraphics[width=1.0\linewidth]{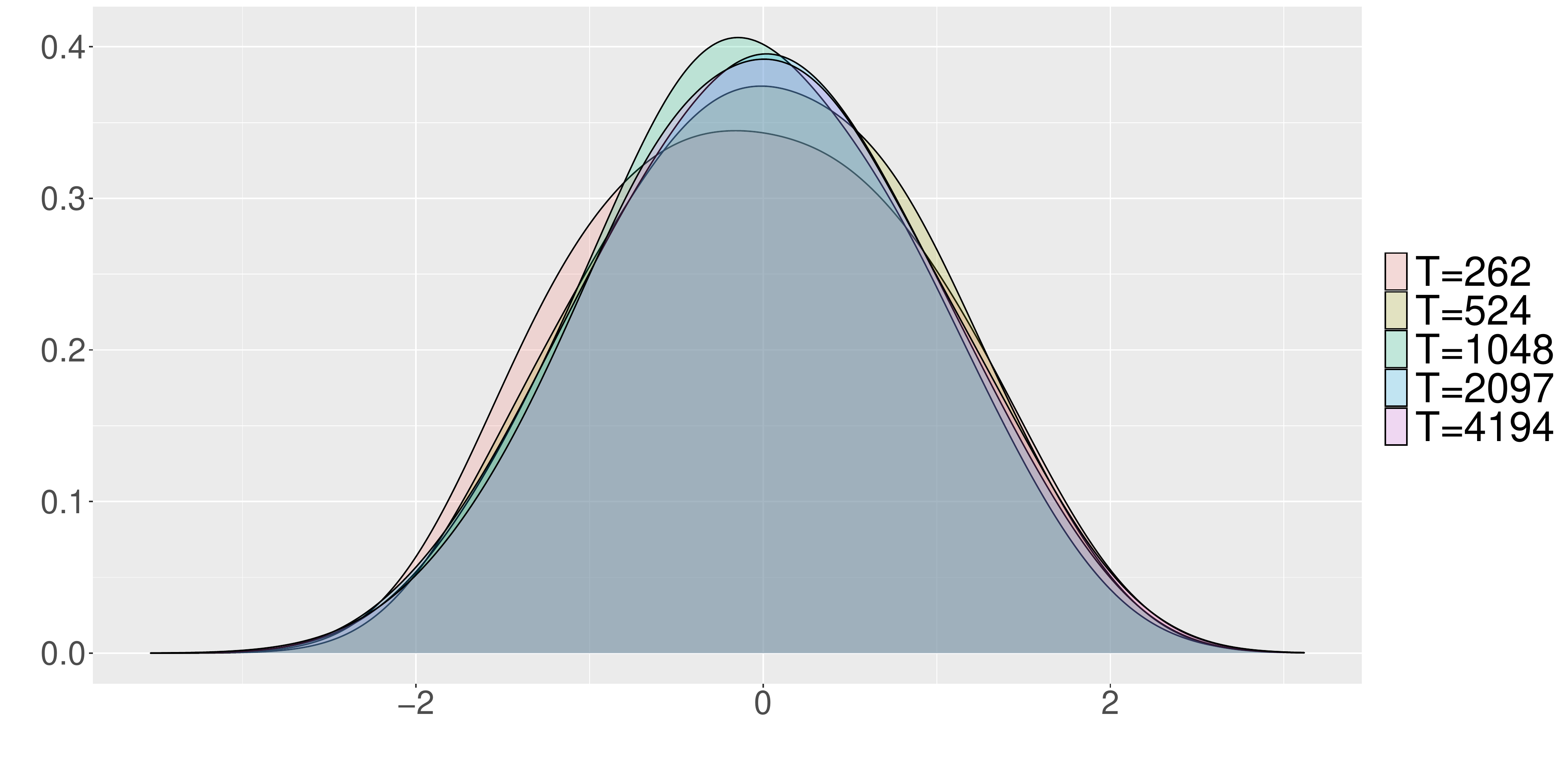}
  \caption{$\tilde{H} = \tilde{T}^{0.7}$}
  \label{fig:sfig2}
\end{subfigure}
\caption{Sampling distribution of 1000 realizations of $\sqrt{\tilde{T}^{1-\kappa}}\hat{\rho}_{\tilde{H}, \tilde{T}}$ using $\tilde{H}=T^{\kappa}$, $\kappa=0.1, 0.3, 0.5, 0.7$ when $d=0$. Model parameters: $\mu =0.000001419188$, $\lambda = 128.2085$, $c=1$, and $\sigma_e=0.0007289$.} \label{fig:asym_dist_adj_rho_hat_power_law_d0}
\end{figure}

\begin{figure}[H]
\begin{subfigure}[b]{0.5\textwidth}
  \centering
  \includegraphics[width=1.0\linewidth]{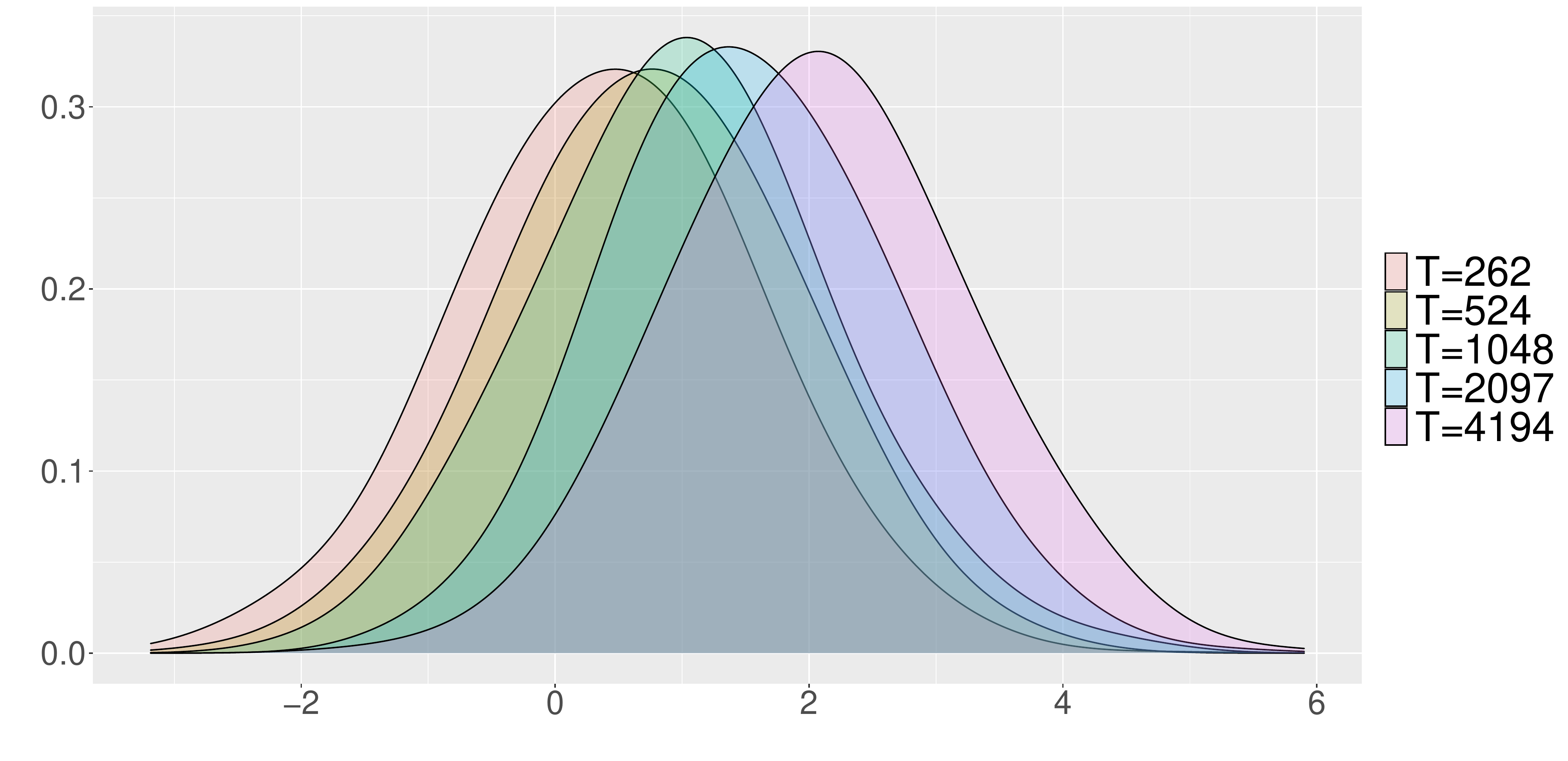}
  \caption{$\tilde{H} = \tilde{T}^{0.1}$}
  \label{fig:sfig1}
\end{subfigure}
\hspace{8pt}
\begin{subfigure}[b]{0.5\textwidth}
  \centering
  \includegraphics[width=1.0\linewidth]{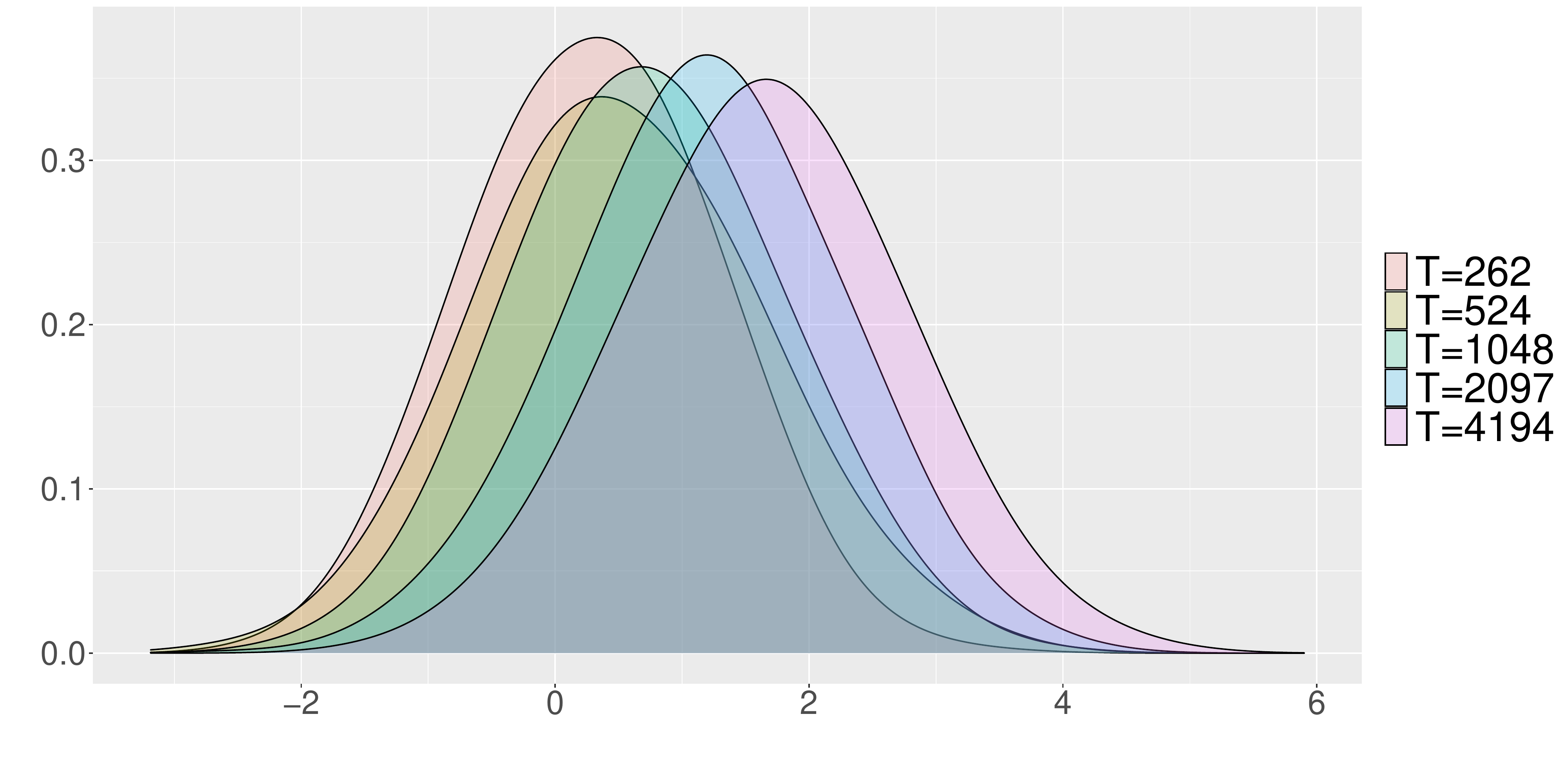}
  \caption{$\tilde{H} = \tilde{T}^{0.3}$}
  \label{fig:sfig2}
\end{subfigure}
\hspace{8pt}
\begin{subfigure}[b]{0.5\textwidth}
  \centering
  \includegraphics[width=1.0\linewidth]{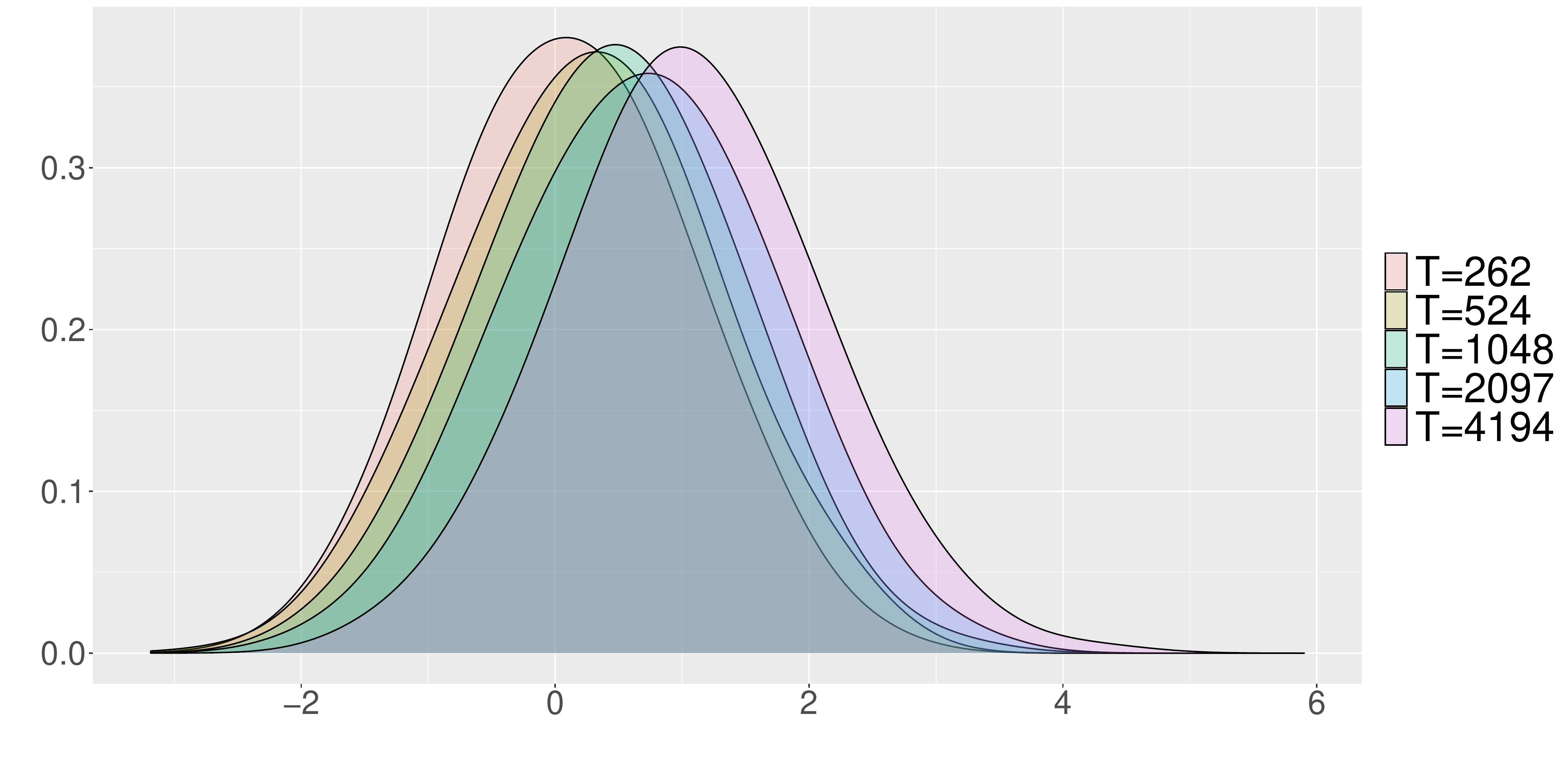}
  \caption{$\tilde{H} = \tilde{T}^{0.5}$}
  \label{fig:sfig2}
\end{subfigure}
\hspace{8pt}
\begin{subfigure}[b]{0.5\textwidth}
  \centering
  \includegraphics[width=1.0\linewidth]{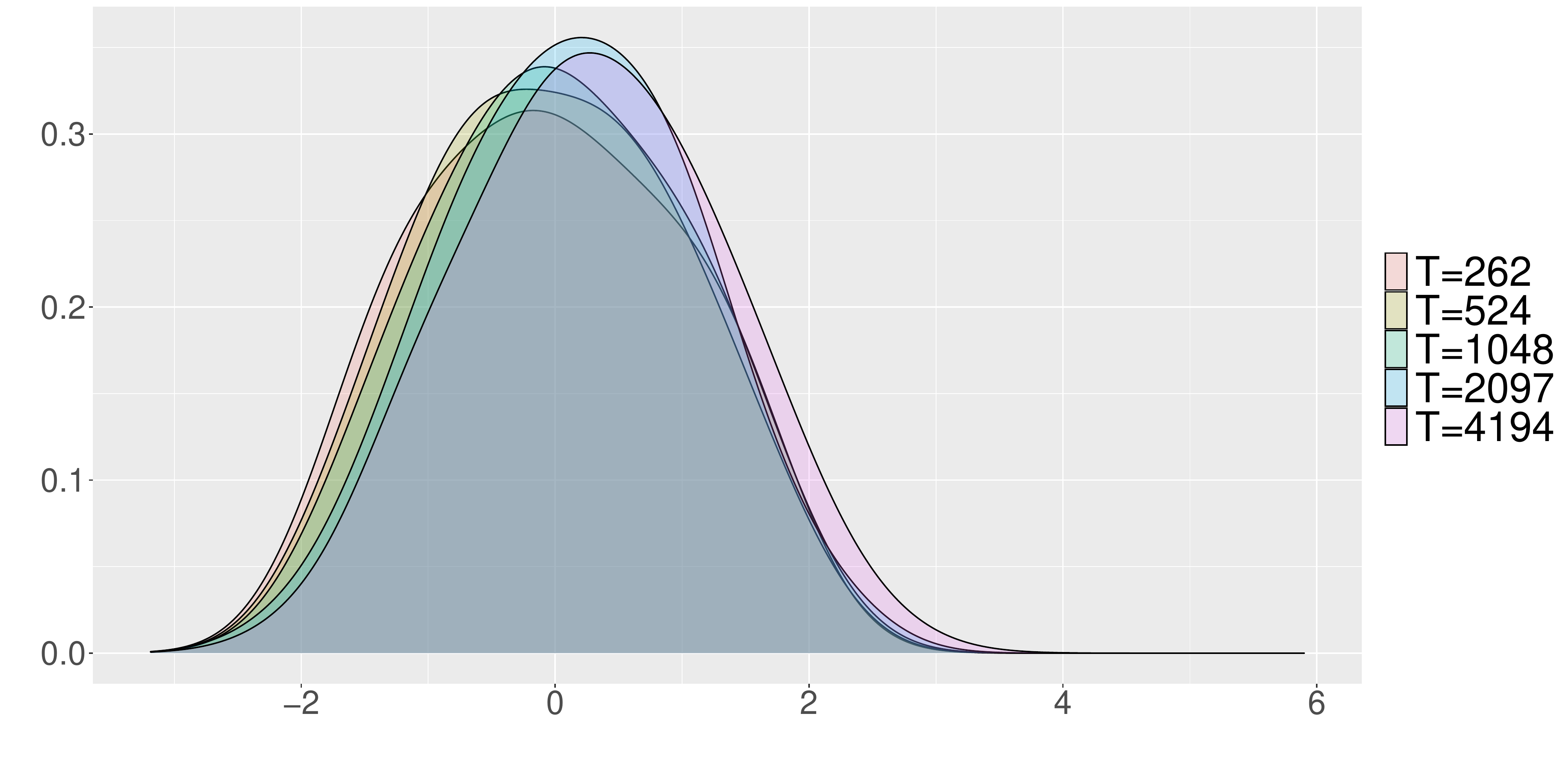}
  \caption{$\tilde{H} = \tilde{T}^{0.7}$}
  \label{fig:sfig2}
\end{subfigure}
\caption{Sampling distribution of 1000 realizations of $\sqrt{\tilde{T}^{1-\kappa}}\hat{\rho}_{\tilde{H}, \tilde{T}}$ using $\tilde{H}=T^{\kappa}$, $\kappa=0.1, 0.3, 0.5, 0.7$ when $d=0.3545$. Model parameters: $\mu =0.000001419188$, $\lambda = 128.2085$, $c=1$, and $\sigma_e=0.0007289$.} \label{fig:asym_dist_adj_rho_hat_power_law_d035}
\end{figure}

\begin{figure}[H]
\begin{subfigure}[b]{0.5\textwidth}
  \centering
  \includegraphics[width=1.0\linewidth]{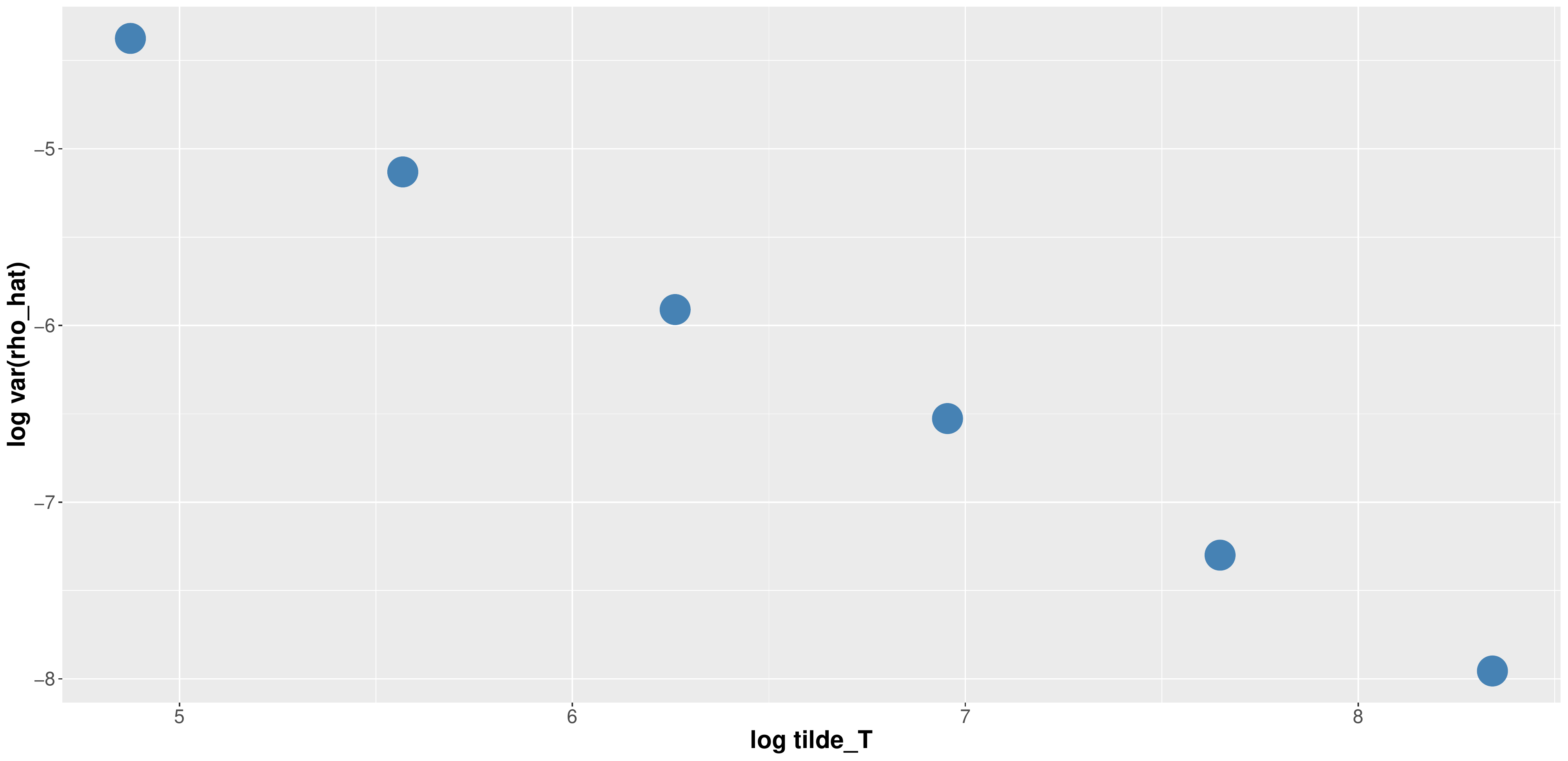}
  \caption{$\tilde{H} = \tilde{T}^{0.1}$}
  \label{fig:sfig1}
\end{subfigure}
\hspace{8pt}
\begin{subfigure}[b]{0.5\textwidth}
  \centering
  \includegraphics[width=1.0\linewidth]{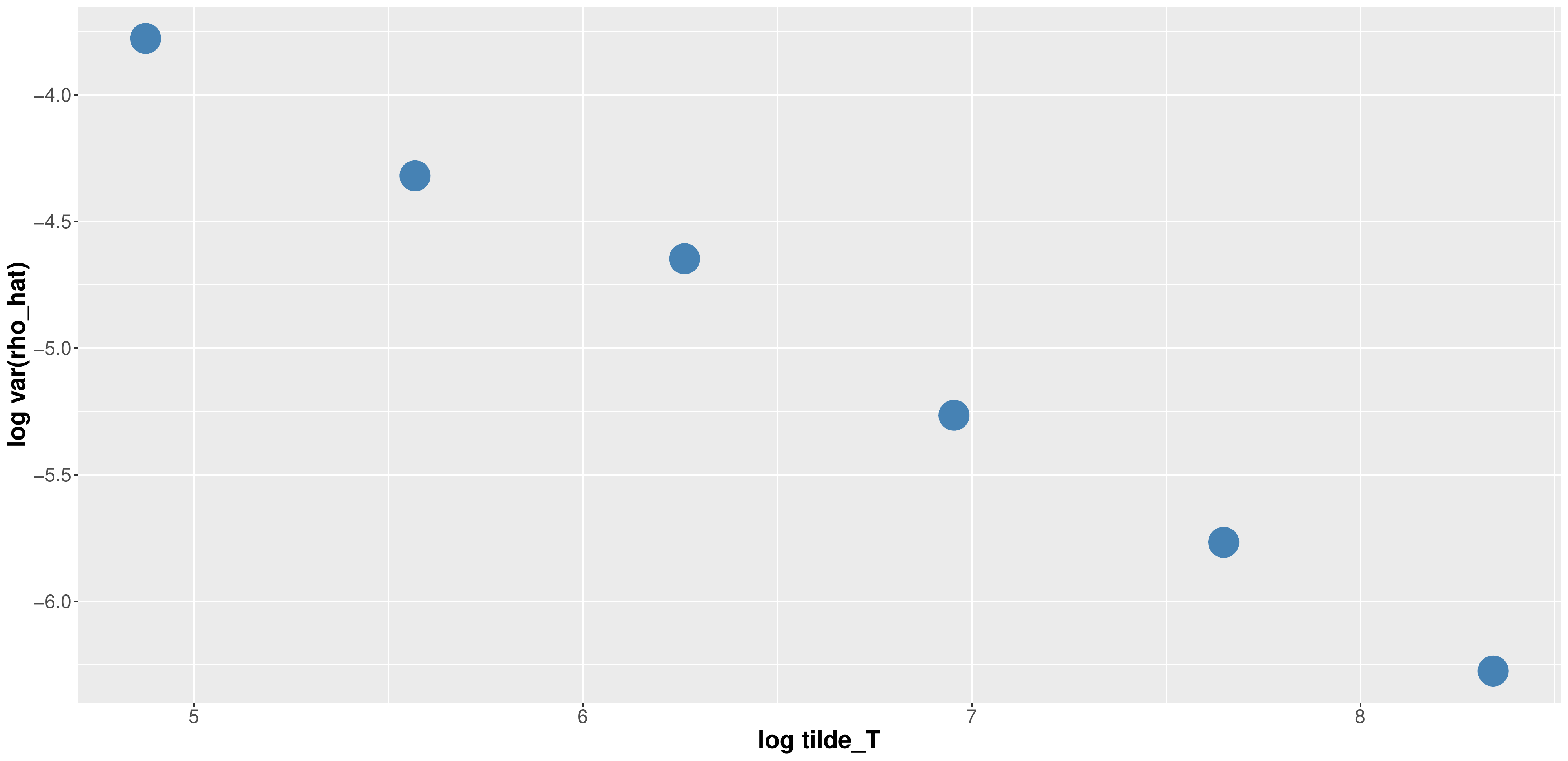}
  \caption{$\tilde{H} = \tilde{T}^{0.3}$}
  \label{fig:sfig2}
\end{subfigure}
\hspace{8pt}
\begin{subfigure}[b]{0.5\textwidth}
  \centering
  \includegraphics[width=1.0\linewidth]{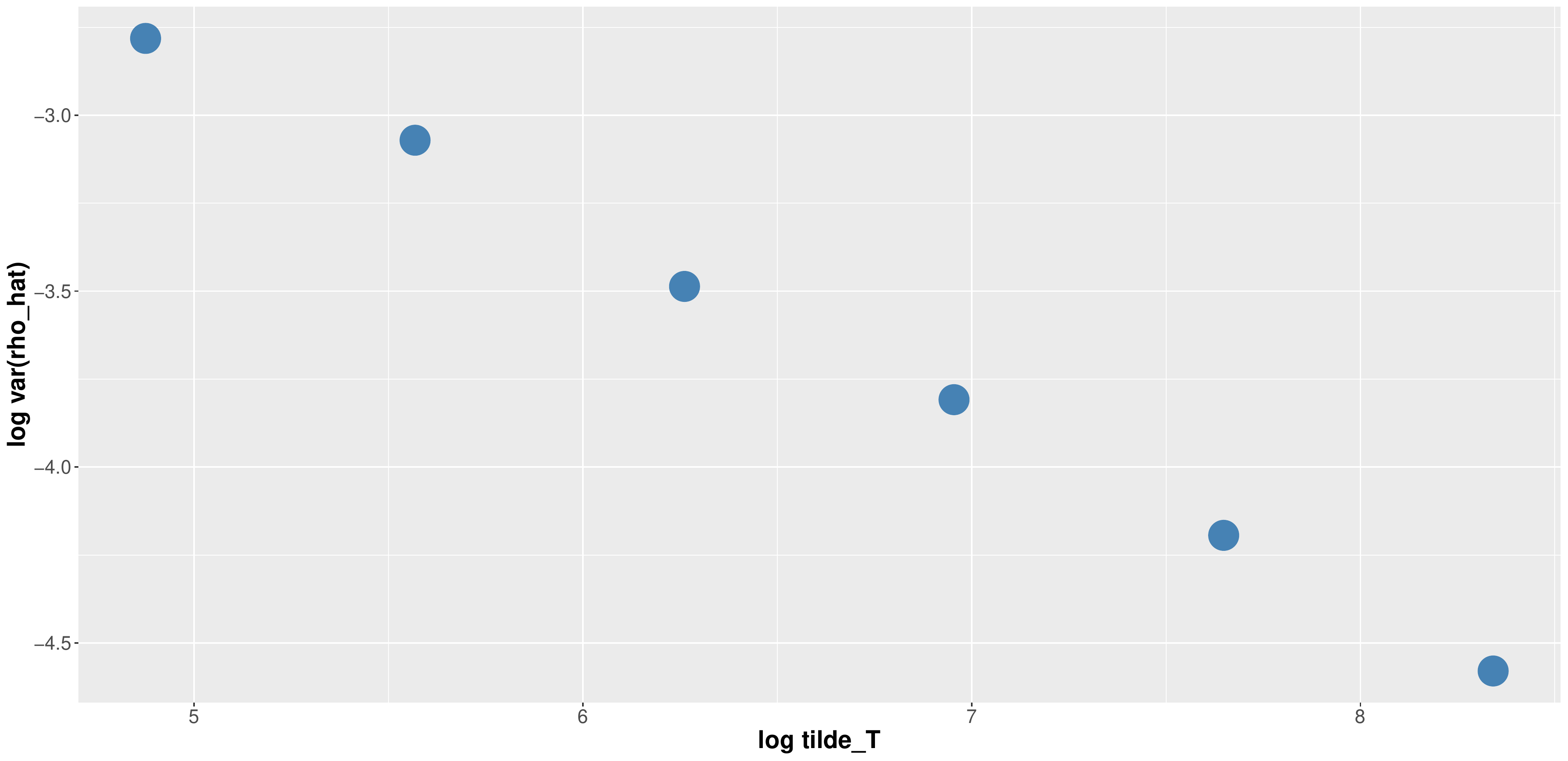}
  \caption{$\tilde{H} = \tilde{T}^{0.5}$}
  \label{fig:sfig2}
\end{subfigure}
\hspace{8pt}
\begin{subfigure}[b]{0.5\textwidth}
  \centering
  \includegraphics[width=1.0\linewidth]{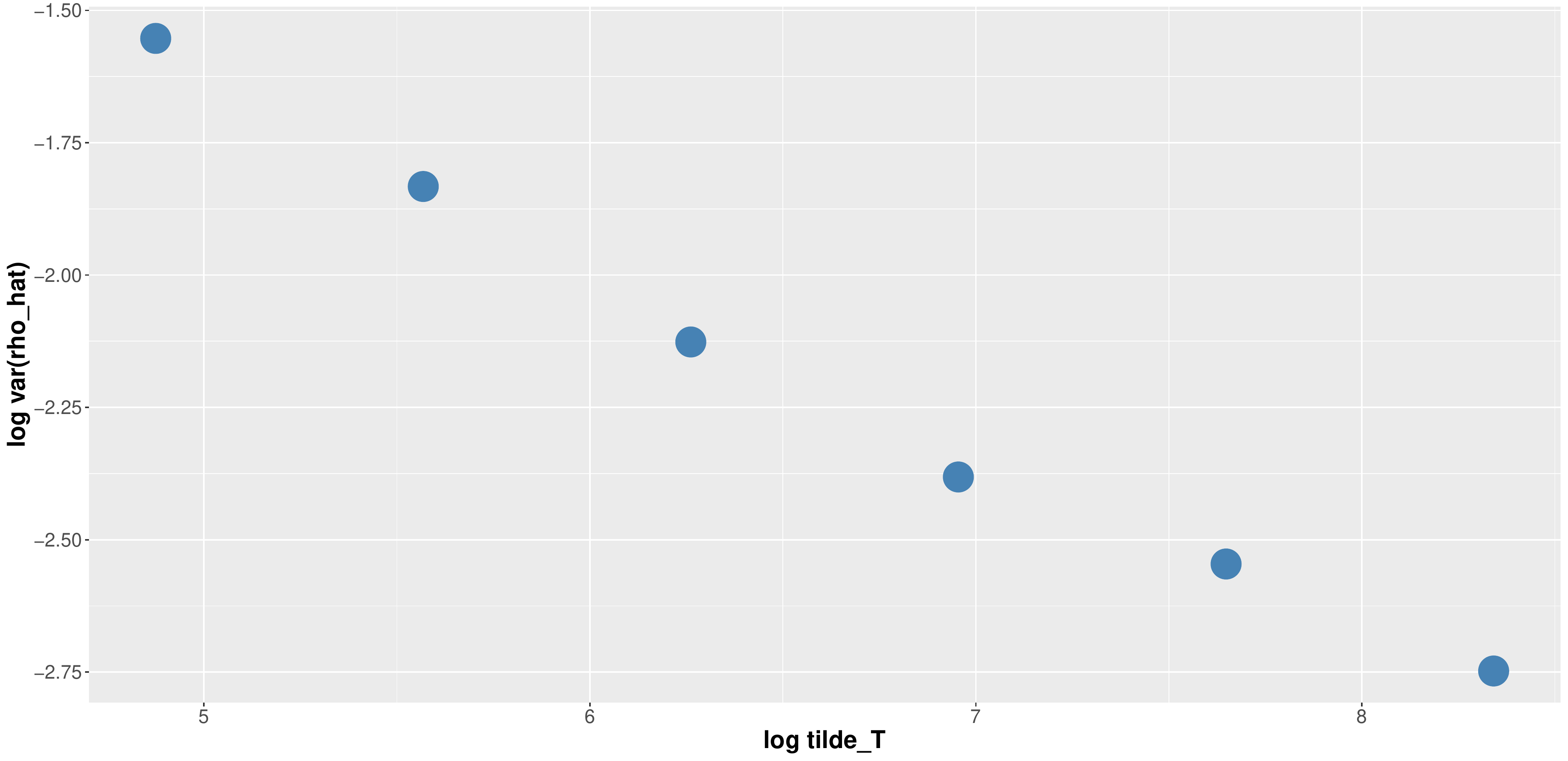}
  \caption{$\tilde{H} = \tilde{T}^{0.7}$}
  \label{fig:sfig2}
\end{subfigure}
\caption{Scatter plots of $\log (\var(\hat{\rho}_{\tilde{H}, \tilde{T}}))$ vs. $\log(\tilde{T})$, where $\tilde{T}=131, 262, 524, 1048, 2097, 4194$.} \label{fig:log_log_plot}
\end{figure}

\section{Comparison with Sizova (2013) } \label{sec:comparison}
Since our paper is not in a local-to-zero framework, we consider a test based on $\hat{\rho}_{\tilde{H}, \tilde{T}}$ with critical value obtained from Theorem 4 of Sizova (2013) under the assumption $\tilde H/ \tilde T \rightarrow \theta \in (0,1)$.
This theorem is based on a further assumption (Assumption 2) that a suitably normalized
version of the regressor converges weakly to a fractional Brownian motion.
Sizova presented several examples under which Assumptions 1 and 2 in that paper hold. The example that is closest to the context of this paper is Example 3, in which the continuous-time long-memory stochastic volatility model
of Comte and Renault (1998) is assumed for the price. In this example, the regresssor is taken to be the
integrated latent variance, $\sigma^2(t)=\int^t_{0}\sigma^2(u)du$, where $\sigma^2(u)$ is an unobserved latent variance, assumed to obey
a stationary long-memory Ornstein-Uhlenbeck model. Sizova pointed out that for this regressor, Assumption 2 would hold,
by an argument similar to the proof of the main result of Taqqu (1976). In comparison with Example 3 of Sizova (2013), we point out that
the regressor in the example could not be used in practice even if the model assumed there held since $\sigma^2(t)$ is a latent
process. Clearly, one could try using the realized variance to estimate $\sigma^2(t)$, and indeed Comte and Renault
(1998, Proposition 5.1, Page 305) provide a theoretical result on the $L^2$ convergence and $L^1$ convergence of a suitably
normalized version of realized variance to $\sigma^2(t)$. However, neither Sizova (2013) nor Comte and Renault (1998) has claimed
or proved that Assumption 2 of Sizova (2013) would hold for some function of the realized variance.
Furthermore, a practical issue
that would arise in trying to estimate $\sigma^2(t)$ in the model of Comte and Renault (1998) is that realized variance is
not a useful quantity when based on very short time increments because the squared returns are mostly zero.
This is often attributed to microstructure noise in the literature,
but in our viewpoint the problem is deeper, and is most fundamentally attributable to the step-function nature of the log
price process.

\section{Conclusion}\label{sec:conclusion}
Much of the literature on long-horizon return predictability finds that the evidence for predictability improves as the aggregation horizon increases.
Many of these studies, however, were subject to spurious findings due to the amplified correlations in the predictors and the returns caused by the overlapping aggregation.

Our work provides a theoretical framework for exploring the correlation between the long-horizon returns and the realized variance. We propose a parametric transaction-level model for the continuous-time log price based on a pure-jump process. The model determines returns and realized variance at any level of aggregation and reveals a channel through which the statistical properties of the transaction-level data are propagated to induce long-horizon return predictability. Specifically, the memory parameter of returns at the transaction level propagates unchanged to the returns and realized variance at any calendar-time frequency, which leads to a balanced predictive regression.  We show that if the long-memory parameter is zero, then there is no long-horizon return predictability. This extends a result shown previously
by Sizova (2013) based on discrete-time assumptions.
In addition, we provide consistent estimators for the model parameters.
To assess return predictability, we propose a hypothesis test based on a power-law aggregation of the returns and realized variance.  We demonstrate that
the proposed test is asymptotically correctly-sized and is consistent, whereas the related test of Sizova (2013), which used a linear aggregation framework, is inconsistent.

\appendix

\section*{Appendix}\label{sec:appendix}

\section{Model Parameter Estimation}\label{sec:model_estimates}
In this section, we derive formulas for estimating the model parameters, $\mu, \lambda, \sigma^2_e$, $c$, and $d$ based on $\log P(t)$, $t \in (0, T]$.
We use moments of the calendar-time returns $\{r_t\}^{T}_{t=1}$ and the counts $\{\Delta N(t)\}^{T}_{t=1}$ to derive formulas for estimating model parameters.
We establish the consistency of the proposed estimators in Theorem \ref{thm:consistent_estimators}.

\subsection*{Moments of returns and counts}
We first obtain certain moments of returns and counts.
The expectations of $r_t$ and $\Delta N(t)$ can be expressed as
\begin{equation}\label{eq: mean_return}
\esp[r_t] = \mu \esp[\Delta N(t)]
\end{equation}
where
\begin{equation}\label{eq:mean_count}
\esp[\Delta N(t)] = \lambda \int_{t-1}^t \esp[\rme^{Z_H(s)}]\rmd s = \lambda \rme^{\frac{1}{2}}\;.
\end{equation}
Hence
\begin{equation}\label{eq: mean_return_v2}
\esp[r_t]= \mu \lambda e^{\frac{1}{2}}\;.
\end{equation}

\noindent From Corollary \ref{cor:var_rt}, when $d =0$, Equation (\ref{eq:var_r0}) becomes
\begin{equation}\label{eq:var_rt_d0}
\var(r_t) = \mu^2 \lambda^2 \rme + \lambda e^{\frac{1}{2}}\left( \mu^2 + \sigma^2_e \right) \;.
\end{equation}

\begin{lemma}\label{lem:count_property}
The autocovariance function and the variance of the counts $\Delta N(t)$ are equal to
\begin{eqnarray}\label{eq:autocov_count}
\cov\left(\Delta N(L), \Delta N(0)\right)
&=& \lambda^2 \int^{0}_{-1}\int^{0}_{-1}\rme\left[e^{\gammaz(L+s-t)}-1 \right]\rmd s \rmd t \;, \;\; (L=1,2, \cdots, )
\end{eqnarray}
and
\begin{equation}\label{eq:var_count}
\var\left(\Delta N(t)\right) = \lambda^2 \int^{0}_{-1}\int^{0}_{-1}\rme\left[\rme^{\gammaz(s-t)}-1 \right]\rmd s \rmd t + \lambda \rme^{\frac{1}{2}}\;,
\end{equation}
where
\begin{equation}\label{eq:gammaz_L}
\gammaz(L+s-t) = \frac{1}{2}\left[|c(L+s-t)+1|^{2H}-2|c(L+s-t)|^{2H} + |c(L+s-t)-1|^{2H} \right] \;.
\end{equation}
Moreover, when $d=0$,
\begin{equation}\label{eq:var_count_d0}
\var\left(\Delta N(t)\right) = \lambda^2 \rme + \lambda \rme^{\frac{1}{2}}\;.
\end{equation}
\end{lemma}

\subsection*{Estimation of $\mu$}
Let $\bar{r}_t = \frac{1}{T}(r_1+r_2+ \cdots + r_T)$ and $\overline{\Delta N(t)} = \frac{1}{T}(\Delta N(1)+\Delta N(2) + \cdots + \Delta N(T))$.
From Equation (\ref{eq: mean_return}), the drift term can be estimated by
\begin{equation}\label{eq:est_mu}
\hat{\mu} = \frac{\bar{r}_t}{\overline{\Delta N(t)}} \;.
\end{equation}

\subsection*{Estimation of $\lambda$}
Using Equation (\ref{eq:mean_count}), we can estimate $\lambda$ by
\begin{equation}\label{eq:est_lam}
\hat{\lambda} = \frac{\overline{\Delta N(t)}}{\rme^{\frac{1}{2}}} \;.
\end{equation}

\subsection*{Estimation of $c$ and $d$ }
Let $\hat{\gamma}_{\Delta N}(1)$ and $\hat{\gamma}_{\Delta N}(2)$ be the sample lag-one and lag-two autocovariances of the counts.
Substituting $\hat{\lambda}$ into Equation (\ref{eq:autocov_count}), we solve a system of nonlinear equations to estimate $c$ and $d$:
\begin{eqnarray}
\hat{\lambda}^2 \int^{0}_{-1}\int^{0}_{-1}\rme\left[e^{\gammaz(1+s-t)}-1 \right]\rmd s \rmd t - \hat{\gamma}_{\Delta N}(1)  &=& 0 \\
\hat{\lambda}^2 \int^{0}_{-1}\int^{0}_{-1}\rme\left[e^{\gammaz(2+s-t)}-1 \right]\rmd s \rmd t - \hat{\gamma}_{\Delta N}(2)  &=& 0 \;.
\end{eqnarray}

\subsection*{Estimation of $\sigma_e^2$}
From Equation (\ref{eq:var_count}), we can express
\begin{equation}\label{eq:cov_frac_gau}
\lambda^2 \rme \int_{-1}^0 \int_{-1}^0 \left\{\rme^{\gammaz(s-t)} -1  \right\}\rmd s\rmd t  = \var(\Delta N(t)) - \lambda \rme^{\frac{1}{2}}
\end{equation}
Substituting Equation \ref{eq:cov_frac_gau} into Equation (\ref{eq:var_r0}), Equation (\ref{eq:var_r0}) becomes
\begin{equation}\label{eq:var_rt_rearrange}
\var(r_t) = \mu^2 \var(\Delta N(t)) + \lambda \rme^{\frac{1}{2}}\sigma^2_e
\end{equation}
Let $\hat{\sigma}^2_r$ be the sample variance of returns and $\hat{\sigma}^2_{\Delta N(t)}$ the sample variance of the count. Plugging in $\hat{\mu}$ and $\hat{\lambda}$ from Equations (\ref{eq:est_mu}) and (\ref{eq:est_lam}) into Equation (\ref{eq:var_rt_rearrange}), we can estimate $\sigma^2_e$ by
\begin{equation}\label{eq:est_sigma_e}
  \hat{\sigma}_e^2 = \frac{\hat{\sigma}^2_r-\hat{\mu}^2 \hat{\sigma}^2_{\Delta N(t)}}{\hat{\lambda}\rme^{\frac{1}{2}}}\;.
\end{equation}

\noindent In the following theorem, we prove that our proposed model estimators for $d=0$ are consistent.

\begin{theorem}\label{thm:consistent_estimators}
Under the null hypothesis $d=0$, $\hat{\mu}$, $\hat{\lambda}$, $\hat{\sigma}_e$, and $\hat{c}$ are consistent estimators.
\end{theorem}

\section{Simulation Methodology and Performance of Parameter Estimates}\label{sec:sim_model}
In this section, we describe a method of simulating realizations of the $\log P(t)$ process. We start with simulation of the counting process.

We approximate the intensity $\Lambda(k) := \int^k_{0}\lambda \rme^{Z_H(s)}\rmd s$ with the Riemann sum
\begin{equation}
  \label{approx_intensity}
  \tilde{\Lambda}(k) = \frac{1}{M}\sum_{i=1}^{kM}\lambda e^{Z_{H}\left(\frac{i}{M}\right)} \; ,
\end{equation}
where $i$ denotes the partition index, $i=1,2,\ldots, kM$, $k=1,2,\ldots, T$, and $M$ is the number of
partitions. Define a stationary Gaussian time series $\{X_j\}_{j=-\infty}^{\infty}$ with
mean zero and lag-$r$ autocovariance function
\begin{equation}
  \label{cont_time_frac_gauss}
  \gamma_x(r) = \frac{1}{2}\left[\left|\frac{cr}{M}+1\right|^{2H}-2\left|\frac{cr}{M}\right|^{2H}+\left|\frac{cr}{M}-1\right|^{2H}   \right] \; .
\end{equation}
Note that
\begin{equation}
  \label{intensity_equiv}
  \left\{\tilde{\Lambda}(k)\right\}_{k=1}^T  \,{\buildrel d \over =}\, \frac{1}{M}\sum_{j=1}^{kM}\lambda e^{X_j} \;, \; k=1,2,\ldots, T \; .
\end{equation}
We apply Davis and Harte's (1987) method to simulate realizations of $\{X_j\}$ and then generate
$\{\tilde{\Lambda}(k)\}_{k=1}^T$ by (\ref{intensity_equiv}).

\noindent Given $\tilde{\Lambda}$ (when $d> 0$ or when $d=0$), we simulate a sequence of independent exponential random variables $\{\tau_j\}$ with mean $1/\lambda$.
For each positive integer $k$, the simulated value of the counting process $N(k)$ is the greatest integer $j$ such that
\begin{equation}
  \label{def_counting}
  \sum_{j=1}^{N(k)}\tau_j  \le  \tilde{\Lambda}(k) \;.
\end{equation}\footnote{Conditional on $\tilde{\Lambda}(k)$, $N(k)$ is a Poisson process with mean measure $\tilde{\Lambda}(k)$, which represents the expected number of transactions within the interval $\tilde{\Lambda}(k)$. To simulate the Poisson process, we first simulate iid exponential durations $\tau_j$ with mean $1/\lambda$. The cumulative sum of $\tau_j$ represents the arrival time of transactions. Hence $N(k)$ is obtained by finding the $k^{th}$ transaction whose arrival time is closest to $\tilde{\Lambda}(k)$.}
The counts can then be generated by
\begin{equation}
  \label{eq:calendar_time count}
  \Delta N(k) = N(k) - N(k-1), \;\; k \ge 1.
\end{equation}

\noindent We consider the special case of $c=1$ for the simulations. First we simulate $2^p$ observations of $\{X_j\}$, where $p$ is a positive integer, and $40,000,000$ exponential durations with $\lambda=128.2085$ \footnote{The model parameters are calibrated based on the mean and standard deviation of the
$5$-minute Boeing returns used in Cao et al. (2017), $3.846\times 10^{-6}$ and $0.0012$. Assume the average number of transactions per 5 minutes is $2.71$.} to obtain
the counts $\{N(k)\}$.
We then simulate a sequence of $iid$ efficient shocks $\{e_i\}$ from a normal distribution with
mean zero and standard deviation $\sigma_e=0.0007289$.
We assume $\mu=1.419188\times10^{-6}$ and obtain the log calendar-time price process by the following equation
\begin{equation} \label{sim_logPrice}
  \log P(k) = \mu N(k) + \sum_{i=1}^k \rme_i \;,
\end{equation}
which is used to generate the daily returns $r_t = \log P(t) - \log P(t-1)$, for $t=1,2,\ldots, T$.
We aggregate $\{r_t\}$ to obtain the monthly returns $R_{\tilde{t}} = \sum_{k=(\tilde{t}-1)m+1}^{\tilde{t}m}r_k$, for $\tilde{t}=1,2,\ldots, \tilde{T}$. Here we assume $m=20$ trading days per month. We then aggregate the monthly returns $\{R_{\tilde{t}}\}$ over $\tilde{H}$ months to obtain the forward-aggregated returns $R_{\tilde{t}, \tilde{t}+\tilde{H}} = \sum^{\tilde{t}+\tilde{H}}_{j=\tilde{t}}R_j$.
Similarly, we aggregate the squared daily returns to obtain the monthly realized variance $RV_{\tilde{t}} = \sum_{k=(\tilde{t}-1)m+1}^{\tilde{t}m} r^2_t$ and the backward-aggregated realized variance
$RV_{\tilde{t}-\tilde{H}, \tilde{t}} = \sum_{j=\tilde{t}-\tilde{H}}^{\tilde{t}}RV_j$.
\footnote{Davis and Harte (1987) use the fast Fourier transform (FFT) algorithm for simulating $\{X_j\}$. To speed up the FFT computation, we simulate $2^{p}$ observations of $\{X_j\}$ and use $M=50$ steps to approximate $\tilde{\Lambda}(k)$. This will yield approximately $T= 2^p/M$ daily returns. By assuming $20$ trading days per month, we expect to have $\tilde{T}= 2^p/(20M)$ monthly returns. As a result, $\tilde{T}$ may not be an integer. For example, when $p=22$, we obtain $4194$ monthly returns.
In addition, since we simulate exponential durations to obtain the counting process, the number of simulated daily returns is random, but the variation is very small. For example, we simulated $2^{19}$, $2^{20}$, and $2^{21}$ observations of $\{X_j\}$; in each case, the number of resulting daily returns was identical.}

\noindent Figure \ref{fig:ret_pred} shows the estimated $\hat{\rho}_{\tilde{H}, \tilde{T}}$ for $\tilde{H}=1,2, \cdots, 120$ from two realizations of $\{r_t\}$ of our model.
The upward trend in $\hat{\rho}_{\tilde{H}, \tilde{T}}$ as $\tilde{H}$ increases is similar to the one shown in Figure 1 of Sizova (2013), where she measured return predictability by estimating (1) the coefficient of determination $R^2$ from regressing the forward-aggregated returns on the backward-aggregated realized variance and (2)
the sample correlations between the forward-aggregated returns and the backward-aggregated  dividend yields using the empirical 1952-2011 AMEX/NYSE data.
The similar pattern shown in our simulations and Sizova's (2013) Figure 1 indicates that  the return process generated from our model matches the return-predictability pattern observed in empirical financial data.

\noindent We estimate the parameters of our model using the simulated data and evaluate their biases with respect to the model parameters.
We present our estimation results in Table \ref{tab:est_model_parm} and Figure \ref{fig:boxplot_est_model}. These results indicate a larger estimation bias for a larger memory parameter $d$.

\begin{figure}[H]
\begin{subfigure}[b]{0.5\textwidth}
  \centering
  \includegraphics[width=1.0\linewidth]{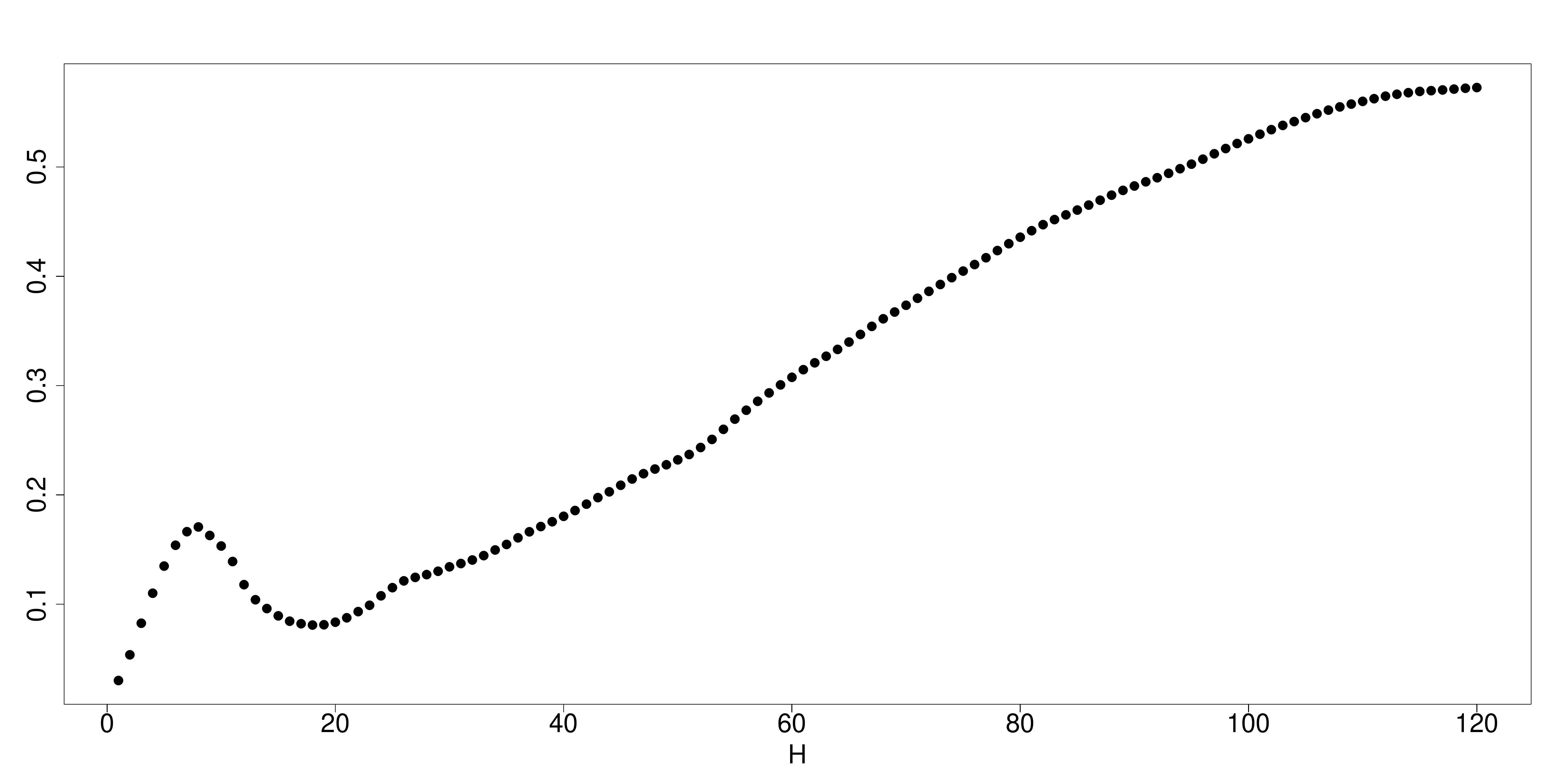}
  \caption{realization 1}
  \end{subfigure}
\qquad
\begin{subfigure}[b]{0.5\textwidth}
  \centering
  \includegraphics[width=1.0\linewidth]{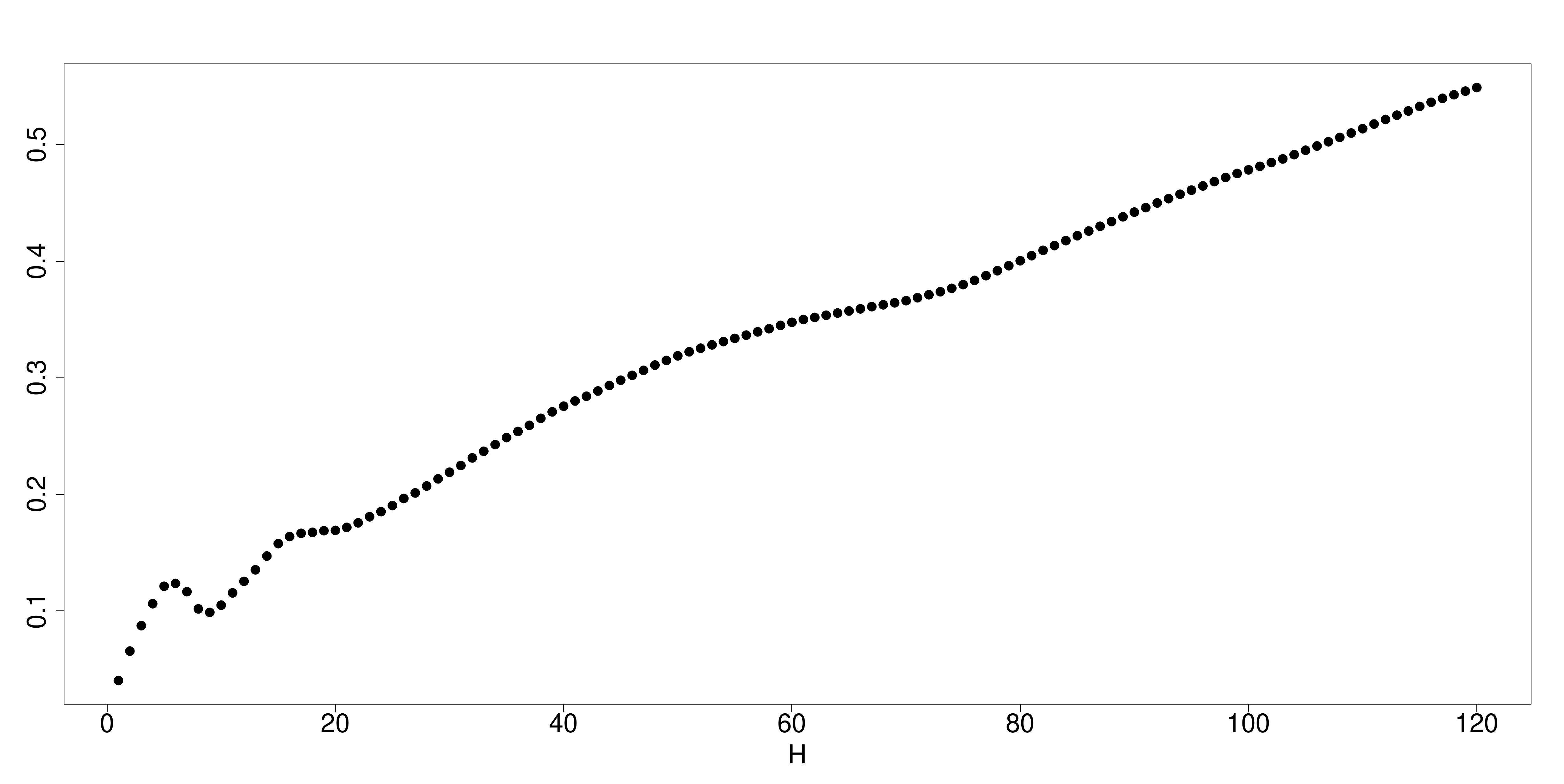}
  \caption{realization 2}
  \end{subfigure}
\caption{Return predictability $\hat{\rho}_{\tilde{H}, \tilde{T}}$ from two realizations of $\{r_t\}$. }\label{fig:ret_pred}
\end{figure}

\begin{table}[H]
\caption{Estimated model parameters for $T=83,886$, with $d=0,0.15, 0.25, 0.35$, $\mu =0.000001419188$, $\lambda = 128.2085$, $c=1$, and $\sigma_e=0.0007289$. The reported estimates are the average of the estimated model parameters over 500 realizations. The symbol $*$ indicates a statistically significant difference between the average and the true model parameter at the significance level $0.05$.}
\center
\begin{tabular}{|c|c|c|c|c|c|}\hline
 \emph{\textbf{d}}    & $\hat{\mu}$    & $\hat{\lambda}$ & $\hat{\sigma_e}$ & $\hat{d}$  & $\hat{c}$ \\ \hline
0.0   & 1.406521e-06   &128.1886        & 0.0007289484     & 0.00013  & 1.0011 \\
0.15  & 1.406393e-06   &128.0719        & 0.0007287938     & 0.1497   & 1.0027* \\
0.25  & 1.406578e-06   &127.7369        & 0.0007288298     & 0.2479*  & 1.0039* \\
0.35  & 1.408259e-06   &126.5984        & 0.0007288756     & 0.3411*  & 1.0741* \\ \hline
\end{tabular}\label{tab:est_model_parm}
\end{table}

\begin{figure}[H]
  \centering
  \includegraphics[scale=0.36]{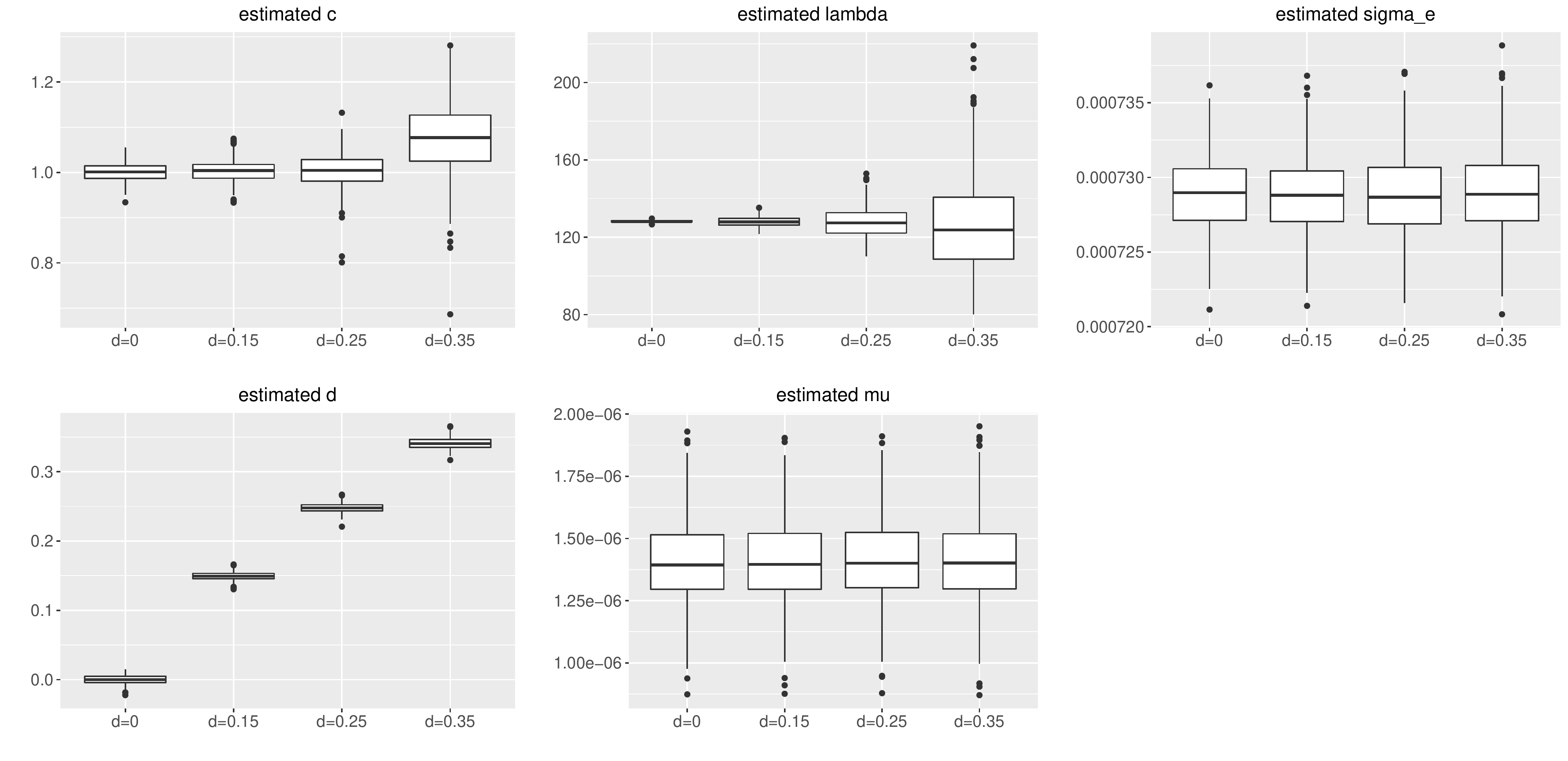}
  \caption{Boxplots of $\hat{\mu}$, $\hat{\lambda}$, $\hat{\sigma}_e$, $\hat{c}$, and $\hat{d}$ estimated from 500 realizations of the simulated counts and return processes with size $T=83,886$ when the population long memory parameter $d=0, 0.15, 0.25, 0.35$.}\label{fig:boxplot_est_model}
\end{figure}

\section{Proofs}\label{sec:proof}

\subsubsection*{Proof of Lemma \ref{lem:autocov_ret}}
\begin{proof}
The lag-L autocovariance of $r_j$ is given by
\begin{eqnarray}
\cov(r_0, r_{L}) &=& \cov\left(\esp[r_0|\Lambda],\; \esp[r_L|\Lambda]\right) = \mu^2 \cov \left(\mathbb{E}[\Delta N(0)|\Lambda]\;, \mathbb{E}[(\Delta N(L))|\Lambda] \right) \notag \\
                  &=& \mu^2 \lambda^2 \rme \int_{s=-1}^0 \int_{t=-1}^0 \left\{\rme^{\gammaz(L+t-s)}-1  \right\}\rmd t\rmd s \;.
\end{eqnarray}
From (\ref{def:cov_A0AL}) and Lemma \ref{lem:slowvary_cov_A0AL_B0BL_A0BL}, as $L \to \infty$,
\[
\rme \int_{s=-1}^0 \int_{t=-1}^0 \left\{\rme^{\gammaz(L+t-s)}-1  \right\}\rmd t\rmd s \sim \rme \; \gammaz(L)\;.
\]
It follows that $\cov(r_0, r_{L}) \sim \mu^2 \lambda^2 \rme \; \gammaz(L)$.
\end{proof}

\subsubsection*{Proof of Corollary \ref{cor:var_rt}}
\begin{proof}
The variance of $r_j$ can be represented as
\begin{eqnarray}
\var(r_j) &=& \var\left(\esp[r_j|\Lambda]\right) + \esp\left[\var(r_j|\Lambda)\right] \notag \\
          &=& \var\left(\mu \lambda \int_{t-1}^t e^{Z_H(s)}\rmd s  \right) + \esp\left[\var(r_j|\Lambda)\right] \notag \\
          &=& \mu^2\lambda^2\rme^1\int_{s=-1}^0\int_{t=-1}^0\left\{\rme^{\gammaz(s-t)}-1\right\}\rmd t \rmd s +\lambda \rme^{1/2}(\mu^2+\sigma_e^2)\;.
\end{eqnarray}
\end{proof}

\begin{lemma}\label{lem:cov_ret_pred}
$\cov\left(R_{\tilde{t}, \tilde{t}+\tilde{H}}, RV_{\tilde{t}-\tilde{H}, \tilde{t}}\right)$ is given by Equation (\ref{eq:cov_ret_pred}) and is positive.
\end{lemma}
\begin{proof}
By (\ref{eq:long term ret}) and (\ref{eq:def_real_var}),
\begin{equation}\label{eq:H_1_ret}
   R_{\tilde{t}, \tilde{t}+\tilde{H}} = \sum_{j=\tilde{t}m+1}^{(\tilde{t}+\tilde{H})m} r_j  = \sum_{j=\tilde{t}m+1}^{(\tilde{t}+\tilde{H})m}\left[\mu \Delta N(j)+\sum_{k=N(j-1)+1}^{N(j)}e_k \right]\;,
\end{equation}
\begin{equation}\label{eq:H_1_RV}
   RV_{\tilde{t}-\tilde{H}, \tilde{t}} = \sum_{j=(\tilde{t}-\tilde{H})m+1}^{\tilde{t}m}r_j^2 = \sum_{j=(\tilde{t}-\tilde{H})m+1}^{\tilde{t}m} \left[\mu \Delta N(j)+\sum_{k=N(j-1)+1}^{N(j)}e_k \right]^2 \;.
\end{equation}
Hence
\begin{eqnarray}
\cov\left(R_{\tilde{t}, \tilde{t}+\tilde{H}}, RV_{\tilde{t}-\tilde{H}, \tilde{t}}  \right) &=& \cov \left(\sum_{j=\tilde{t}m+1}^{(\tilde{t}+\tilde{H})m} r_j\; , \sum_{j=(\tilde{t}-\tilde{H})m+1}^{\tilde{t}m}r_j^2  \right) \notag \\
                                             &=& \mathbb{\esp}\left[\sum_{j=\tilde{t}m+1}^{(\tilde{t}+\tilde{H})m }\sum_{i=\tilde{t}m+1}^{(\tilde{t}+\tilde{H})m} r_j\; r^2_{i-\tilde{H}m} \right]
         -\mathbb{\esp}\left[\sum_{j=\tilde{t}m+1}^{(\tilde{t}+\tilde{H})m}r_j\right]\mathbb{\esp}\left[\sum_{j=\tilde{t}m+1}^{(\tilde{t}+\tilde{H})m}r^2_{j-\tilde{H}m}\right] \notag \\
                                             &=& \sum_{j=\tilde{H}m+1}^{2\tilde{H}m}\sum_{i=\tilde{H}m+1}^{2\tilde{H}m}\cov(r_j\;, r^2_{i-\tilde{H}m}) \notag \\
                                             &=& \sum_{L=1}^{\tilde{H}m}\cov(r_L, r^2_0)L +\sum_{L=\tilde{H}m+1}^{2\tilde{H}m-1}\left(2\tilde{H}m-L\right)\cov(r_L, r^2_0)
\end{eqnarray}
where $L = j-(i-\tilde{H}m)$.
\\
Conditional on $\Lambda$, $N(t)$ is a poisson process, so its increments are independent of each other.  Since $L > 0$,  $\cov(\Delta N(L)\;, (\Delta N(0))^2 |\Lambda )=0$.
The efficient shock $e_k$ is independent of the counting process. These properties imply
\[\esp\left[\cov(\Delta N(L)\;,\left(\sum_{k=N(-1)+1}^{N(0)}e_k\right)^2  |\Lambda )\right] =0 \] and
\[\esp[\cov(r_L, r_0^2|\Lambda]=0 \;. \]
Hence
\begin{eqnarray}
\cov\left(r_L\;, r^2_{0}  \right) &=& \cov\left(\esp[r_L|\Lambda],\; \esp[r^2_0|\Lambda]\right) \notag\\
                                  &=& \mu^3 \cov \left(\mathbb{E}[\Delta N(L)|\Lambda]\;, \mathbb{E}[(\Delta N(0))^2|\Lambda] \right) \notag \\
                                  &+&  \mu \cov \left(\esp[\Delta N(L)|\Lambda],\;  \esp\left[\left(\sum_{k=N(-1)+1}^{N(0)}e_k\right)^2|\Lambda\right]  \right)
\end{eqnarray}
First, we consider
\begin{eqnarray}
&&\cov \left(\mathbb{E}[\Delta N(L)|\Lambda]\;, \mathbb{E}[(\Delta N(0))^2|\Lambda] \right) \notag \\
                                                    &=& \cov \left(\lambda \int_{t=L-1}^{L} \rme^{Z_H(t)}\rmd t,\;
                                                    \lambda  \int_{t=-1}^{0}\rme^{Z_H(t)}\rmd t
                                                    +\lambda^2 \int_{t=-1}^{0}\int_{s=-1}^{0}\rme^{Z_H(s)}\rme^{Z_H(t)}\rmd s\rmd t    \right) \notag \\
                                                    &=& \lambda^2 \cov \left(\int_{s=L-1}^L \rme^{Z_H(s)}\rmd s \;,  \int_{t=-1}^{0} \rme^{Z_H(t)}\rmd t \right)\notag \\
                                                    &+& \lambda^3 \cov \left(\int_{s=L-1}^{L}\rme^{Z_H(s)}\rmd s \;, \int_{t=-1}^{0}\int_{s=-1}^{0}\rme^{Z_H(s)}\rme^{Z_H(t)}\rmd s\rmd t \right)\label{eq:cov_count}
\end{eqnarray}
For all non-negative integer $L$, define
\begin{eqnarray}
A_L &=& \int_{s=-1}^0 \rme^{Z_H(s+L)}\rmd s = \int_{s=L-1}^L \rme^{Z_H(s)\rmd s} \label{var:A_L}\\
B_L &=& \int_{s=-1}^0\int_{t=-1}^0\rme^{Z_H(s+L)}\rme^{Z_H(t+L)}\rmd s \rmd t = \int_{s=L-1}^L\int_{t=L-1}^L\rme^{Z_H(s)}\rme^{Z_H(t)}\rmd s \rmd t \label{var:B_L}
\end{eqnarray}
Hence Equation (\ref{eq:cov_count}) can also be represented as
\begin{equation}\label{eq:1st_cov}
\cov \left(\mathbb{E}[\Delta N(L)|\Lambda]\;, \mathbb{E}[(\Delta N(0))^2|\Lambda] \right) = \lambda^2 \cov(A_L, A_0) + \lambda^3 \cov(A_L, B_0) \;.
\end{equation}
We will use the following property to evaluate (\ref{eq:1st_cov}). Let $(X,Y)$ be a Gaussian vector such that
$\esp[X]=\esp[Y]=0$, $\var(X)=\sigma^2_X$, $\var(Y)=\sigma^2_Y$ and $\cov(X,Y)=\rho$. Note that
\begin{align}\label{formula:cov_log_normal_1}
  \cov(\rme^{X},\rme^Y)  = \rme^{\sigma^2_X/2+\sigma_Y^2/2} (\rme^{\rho}-1) \; .
\end{align}
This yields
\begin{eqnarray}\label{formula:cov_log_normal_2}
\cov\left(\rme^{Z_H(s)},\; \rme^{Z_H(t)}\rme^{Z_H(u)}\right) &=&\rme^{1/2+(1/2+1/2+\gammaz(t-u))}\left(\rme^{\gammaz(s-t)+\gammaz(s-u)}-1\right) \notag\\
&=& \rme^{\frac{3}{2}}\rme^{\gammaz(t-u)}\left[\rme^{\gammaz(s-t)+\gammaz(s-u)}-1\right]
\end{eqnarray}
By Fubini's Theorem as well as (\ref{formula:cov_log_normal_1})and (\ref{formula:cov_log_normal_2}), Equation (\ref{eq:cov_count}) becomes
\begin{eqnarray}
\cov \left(\mathbb{E}[\Delta N(L)|\Lambda]\;, \mathbb{E}[(\Delta N(0))^2|\Lambda] \right)
&=& \lambda^2 \int_{s=L-1}^L \int_{t=-1}^{0} \cov(\rme^{Z_H(s)},\; \rme^{Z_H(t)})\rmd t \rmd s \notag \\
&+& \lambda^3 \rme^{\frac{3}{2}}\int_{s=L-1}^{L}\int_{t=-1}^{0}\int_{u=-1}^{0}\rme^{\gammaz(t-u)}\left[\rme^{\gammaz(s-t)+\gammaz(s-u)}-1\right]\rmd u\rmd t\rmd s \notag \\
&=& \lambda^2 \rme \int_{s=-1}^0 \int_{t=-1}^0 \left\{\rme^{\gammaz(L+s-t)}-1  \right\} \rmd t \rmd s \notag \\
&+& \lambda^3 \rme^{\frac{3}{2}}\int_{s=-1}^{0}\int_{t=-1}^{0}\int_{u=-1}^{0}\rme^{\gammaz(t-u)}\left[\rme^{\gammaz(s-t+L)+\gammaz(s-u+L)}-1\right]\rmd u\rmd t\rmd s\;. \notag
\end{eqnarray}
It follows that
\begin{eqnarray}
\cov(A_L, A_0) 
               &=& \rme \int_{s=-1}^0 \int_{t=-1}^0 \left\{\rme^{\gammaz(L+s-t)}-1  \right\} \rmd t \rmd s \label{def:cov_A0AL} \\
\cov(A_L, B_0) 
               &=& \rme^{\frac{3}{2}}\int_{s=-1}^{0}\int_{t=-1}^{0}\int_{u=-1}^{0}\rme^{\gammaz(t-u)}\left[\rme^{\gammaz(s-t+L)+\gammaz(s-u+L)}-1\right]\rmd u\rmd t\rmd s \notag\\
               &=& \cov(A_0, B_L) \;. \label{def:cov_A0BL}
\end{eqnarray}
Next, we consider
\begin{eqnarray}
 \cov \left(\mathbb{E}[\Delta N(L)|\Lambda]\;, \mathbb{E}\left[ \left(\sum_{k=N(-1)+1}^{N(0)}e_k\right)^2 |\Lambda\right] \right)
 &=& \lambda^2 \cov\left(\int_{s=L-1}^{L}\rme^{Z_H(s)}\rmd s \;, \sigma^2_e \int_{s=-1}^{0}\rme^{Z_H(s)}\rmd s \right) \notag \\
 &=& \lambda^2 \sigma^2_e \int_{s=L-1}^{L}\int_{t=-1}^{0} \cov\left(\rme^{Z_H(s)}, \rme^{Z_H(t)}  \right) \rmd s \rmd t \notag \\
 &=& \sigma^2_e  \lambda^2  \rme \int_{s=-1}^0 \int_{t=-1}^0 \left\{\rme^{\gammaz(L+s-t)}-1  \right\} \rmd t \rmd s \notag \\
 &=& \sigma^2_e  \lambda^2 \cov(A_L, A_0) \label{eq:2nd_cov}
\end{eqnarray}
Combining Equations (\ref{eq:1st_cov}) and (\ref{eq:2nd_cov}), $\forall L > 0$, we obtain
\begin{eqnarray}\label{eq:cov_r0_sq_ret}
\cov\left(r_L\;, r^2_{0}  \right) &=& \mu^3 \left(\lambda^3\cov(A_0,B_L) +\lambda^2 \cov(A_L, A_0)\right) + \mu \sigma_e^2 \lambda^2\cov(A_L, A_0)\;.
\end{eqnarray}
Since $\gammaz(r) > 0$ $\forall r\in \mathbb{R}$ (see Lemma \ref{lem:gammaz_fun}), it follows that $\cov(A_0,B_L) > 0$ and $\cov(A_L, A_0) > 0$
and thus $\cov\left(r_L\;, r^2_{0}\right) > 0$.
From Equation (\ref{eq:cov_ret_pred}), it follows that $\cov\left(R_{t, t+\tilde{H}}, RV_{t-\tilde{H}, t}\right)> 0$.
\end{proof}

\begin{corollary}\label{cor:return_predicatability_extended_model}
For model in Equation (\ref{eq: price_model_two_shocks}),
$\cov\left(R_{\tilde{t}, \tilde{t}+\tilde{H}}, RV_{\tilde{t}-\tilde{H}, \tilde{t}}  \right) > 0. $
\end{corollary}

\begin{proof}
The return process is then given by
\begin{equation}\label{eq:ret_two_shocks}
 r_t = \mu_1 \Delta N_1(t) + \mu_2 \Delta N_2(t) + \sum_{k=N_1(t-1)+1}^{N_1(t)}e_{1,k} + \sum_{k=N_2(t-1)+1}^{N_2(t)}e_{2,k}
\end{equation}
Under the new model, Equations (\ref{eq:H_1_ret}) and (\ref{eq:H_1_RV}) become
\begin{equation}\label{eq:H_1_ret_two_shocks}
   R_{\tilde{t}, \tilde{t}+1} = \sum_{j=\tilde{t}m+1}^{(\tilde{t}+1)m}\left[\mu_1 \Delta N_1(t) + \mu_2 \Delta N_2(t) + \sum_{k=N_1(t-1)+1}^{N_1(t)}e_{1,k} + \sum_{k=N_2(t-1)+1}^{N_2(t)}e_{2,k} \right]
\end{equation}
\begin{equation}\label{eq:H_1_RV_two_shocks}
   RV_{\tilde{t}-1, \tilde{t}} = \sum_{j=(\tilde{t}-1)m+1}^{\tilde{t}m} \left[\mu_1 \Delta N_1(t) + \mu_2 \Delta N_2(t) + \sum_{k=N_1(t-1)+1}^{N_1(t)}e_{1,k} + \sum_{k=N_2(t-1)+1}^{N_2(t)}e_{2,k}\right]^2
\end{equation}
Using the expression for the $r_t$ in Equation (\ref{eq:ret_two_shocks}) and the assumptions we make for the efficient shocks and the counting processes , we obtain
\begin{eqnarray}\label{eq:cov_ret_pred_two_shocks}
 \cov\left(r_0\;, r^2_{L}  \right)     &=& \mu_1^3 \cov\left(\Delta N_1(0)\;, (\Delta N_1(L))^2  \right) \notag \\
                                       &+& \mu_1 \cov\left(\Delta N_1(0)\;, \left(\sum_{k=N_1(L-1)+1}^{N_1(L)}e_{1,k}\right)^2  \right) \notag \\
                                       &+& \mu_2^3 \cov\left(\Delta N_2(0)\;, (\Delta N_2(L))^2  \right) \notag \\
                                       &+& \mu_2 \cov\left(\Delta N_2(0)\;, \left(\sum_{k=N_2(L-1)+1}^{N_2(L)}e_{2,k}\right)^2  \right) \notag \\
                                       &+& 2\mu_1\mu_2^2 \cov \left(\Delta N_1(0)\;, \Delta N_1(L)\Delta N_2(L)  \right) \notag \\
                                       &+& 2\mu_2\mu_1^2 \cov \left(\Delta N_2(0)\;, \Delta N_1(L)\Delta N_2(L)  \right) \notag \\
                                       &+& 2\mu_1 \cov\left(\Delta N_1(0)\;, \left(\sum_{k=N_1(L-1)+1}^{N_1(L)}e_{1,k}\right)\times \left(\sum_{k=N_2(L-1)+1}^{N_2(L)}e_{2,k}\right) \notag  \right) \notag \\
                                       &+& 2\mu_2 \cov\left(\Delta N_2(0)\;, \left(\sum_{k=N_1(L-1)+1}^{N_1(L)}e_{1,k}\right)\times \left(\sum_{k=N_2(L-1)+1}^{N_2(L)}e_{2,k}\right)  \right) \; .
\end{eqnarray}
Note
\begin{eqnarray*}
&&\cov\left(\Delta N_1(0)\;, \left(\sum_{k=N_1(L-1)+1}^{N_1(L)}e_{1,k}\right)\times \left(\sum_{k=N_2(L-1)+1}^{N_2(L)}e_{2,k}\right)\right) \\
&=& \cov \left(\mathbb{E}[ \Delta N_1(0)|\Lambda ], \; \mathbb{E}\left[\sum_{k=N_1(L-1)+1}^{N_1(L)}\sum_{l=N_2(L-1)+1}^{N_2(L)}e_{1,k}e_{2,l}|\Lambda\right] \right) \\
&+& \mathbb{E}\left[\cov\left(\Delta N_1(0) |\Lambda\right),\; \cov\left(\sum_{k=N_1(L-1)+1}^{N_1(L)}\sum_{l=N_2(L-1)+1}^{N_2(L)}e_{1,k}e_{2,l}|\Lambda\right)   \right] \\
&=& 0 \; .
\end{eqnarray*}
The same conclusion for
\[
\cov\left(\Delta N_2(0)\;, \left(\sum_{k=N_1(L-1)+1}^{N_1(L)}e_{1,k}\right)\times \left(\sum_{k=N_2(L-1)+1}^{N_2(L)}e_{2,k}\right)\right)=0 \; .
\]
\begin{eqnarray*}
\cov \left(\Delta N_1(0)\;, \Delta N_1(L)\Delta N_2(L)  \right) &=& \cov \left(\mathbb{E}[\Delta N_1(0)|\Lambda]\;, \mathbb{E}[\Delta N_1(L)\Delta N_2(L)|\Lambda] \right) \notag \\
                                                    &+& \mathbb{E}\left[\cov(\Delta N_1(0)\;, \Delta N_1(L)\Delta N_2(L))|\Lambda )  \right] \notag \\
                                                    &=& \lambda^3\cov \left(\int_{s=0-1}^{0}\rme^{Z_H(s)}\rmd s \;, \int_{t=L-1}^{L}\int_{s=L-1}^{L}\rme^{Z_H(s)}\rme^{Z_H(t)}\rmd s\rmd t \right) \notag \\
                                                    &=& \lambda^3 \int_{s=-1}^{0}\int_{t=L-1}^{L}\int_{u=L-1}^{L}\cov\left(\rme^{Z_H(s)},\; \rme^{Z_H(t)}\rme^{Z_H(u)}\right)\rmd u\rmd t\rmd s \\
                                                    &=& \cov(A_0,B_L)
\end{eqnarray*}
Same for $\cov \left(\Delta N_2(0)\;, \Delta N_1(L)\Delta N_2(L) \right) = \cov(A_0, B_L)$.
\\
\\
Therefore, Equation (\ref{eq:cov_ret_pred_two_shocks}) becomes
\begin{eqnarray}\label{eq:cov_ret_pred_two_shocks_final}
\cov\left(r_0\;, r^2_{L}  \right) &=& (\mu^3_1+\mu^3_2)\left(\lambda^3\cov(A_0,B_L)+\lambda^2\cov(A_0,A_L)\right)+ 2 \mu_1\mu_2(\mu_1+\mu_2)\lambda^3\cov(A_0,B_L) \notag \\
                                    &+& 2(\mu_1+\mu_2)\lambda^2\sigma_e^2 \cov(A_L, A_0) \notag \\
                                    &=& (\mu_1+\mu_2)(\mu_1^2+\mu_2^2+\mu_1\mu_2)\lambda^3 \cov(A_0,B_L) + (\mu^3_1+\mu^3_2)\lambda^2\cov(A_0,A_L) \notag \\
                                    &+&  2(\mu_1+\mu_2)\lambda^2 \sigma_e^2 \cov(A_L, A_0)   \;. \notag
\end{eqnarray}
Note
$(\mu_1+\mu_2)(\mu_1^2+\mu_2^2+\mu_1\mu_2)=(\mu_1+\mu_2)^3 - \mu_1\mu_2(\mu_1+\mu_2) > 0$,\;\;
$\mu^3_1+\mu^3_2 =(\mu_1+\mu_2)(\mu_1^2+\mu_2^2-\mu_1\mu_2) > 0$. This proves $\cov\left(r_0\;, r^2_{L}  \right) > 0$.
\end{proof}

\begin{lemma}\label{lem:slowvary_cov_A0AL_B0BL_A0BL}
Let $A_L$  and $B_L$ be defined in (\ref{var:A_L}) and (\ref{var:B_L}) respectively.
As $L \rightarrow \infty$,
$\cov(A_L, A_0) \sim \rme\gammaz(L)$, $\cov(A_0, B_L) \sim 2\rme^{\frac{3}{2}}\gammaz(L) \int_{s=-1}^0 \int_{u=-1}^0 \rme^{\gammaz(s-u)}\rmd u \rmd s $, and
\begin{equation}
\cov(B_0, B_L) \sim 4\rme^2\gammaz(L)\int_{s=-1}^0 \int_{t=-1}^0 \int_{u=-1}^0 \int_{v=-1}^0 \rme^{\gammaz(s-t)}\rme^{\gammaz(u-v)}\rmd v\rmd u \rmd t \rmd s.
\end{equation}
\end{lemma}

\begin{proof}
It can be shown that
\begin{eqnarray}\label{eq:cov_counts}
\cov(A_0, A_L) &=& \rme \int_{s=-1}^0 \int_{t=-1}^0 \cov\left(\rme^{Z_H(s+L)},\; \rme^{Z_H(t)}\right)\rmd t  \notag \\
               &=& \rme \int_{s=-1}^0 \int_{t=-1}^0 \left\{\rme^{\gammaz(L+t-s)}-1  \right\}\rmd t\rmd s \notag \\
               &=& \rme \gammaz(L) \int_{u=-1}^1 (1-|u|) \left\{\frac{\rme^{\gammaz(L+u)}-1}{\gammaz(L)}\right\}\rmd u
\end{eqnarray}
The ratio $\frac{\rme^{\gammaz(L+u)}-1}{\gammaz(L)}$ is uniformly bounded for all $u \in (-1,1)$ and by condition (\ref{eq:Z-lrd}), the ratio converges to one as $L$ goes to $\infty$. By the dominated convergence theorem,
the integral converges to~$\int_{-1}^1(1-|t|)\rmd t=1$.  Therefore, as $L \rightarrow \infty$,
\begin{equation}\label{eq:comp1}
\cov(A_0, A_L) \sim \rme \gammaz(L) \;.
\end{equation}
Next,
\begin{eqnarray*}
\cov(A_L, B_0) &=& \rme^{\frac{3}{2}}\int_{t=-1}^0\int_{s=-1}^0\int_{u=-1}^0 \cov\left( \rme^{Z_H(t+L)} \;,\rme^{Z_H(s)}\rme^{Z_H(u)} \right) \rmd u \rmd s \rmd t \\
               &=& \rme^{\frac{3}{2}}\int_{t=-1}^0\int_{s=-1}^0\int_{u=-1}^0 \rme^{\gammaz(s-u)}\left[\rme^{\gammaz(L+t-s)}\rme^{\gammaz(L+t-u)}-1 \right] \rmd u \rmd s \rmd t \\
               &=& \rme^{\frac{3}{2}}\gammaz(L)\int_{t=-1}^0\int_{s=-1}^0\int_{u=-1}^0 \rme^{\gammaz(s-u)}\left\{\frac{\rme^{\gammaz(L+t-s)}\rme^{\gammaz(L+t-u)}-1}{\gammaz(L)}\right\} \rmd u \rmd s \rmd t
\end{eqnarray*}
Again condition (\ref{eq:Z-lrd}) implies that the ratio $\frac{\rme^{\gammaz(L+t-s)}\rme^{\gammaz(L+t-u)}-1}{\gammaz(L)}$ is uniformly bounded for $-1 \le u \le 0$, $-1 \le v \le 0$, and $-1 \le t \le 0$ and converges to $2$ as $L$ goes to $\infty$.
The function $\rme^{\gammaz(s-u)}\left\{\frac{\rme^{\gammaz(L+t-s)}\rme^{\gammaz(L+t-u)}-1}{\gammaz(L)}\right\}$ is positive and is integrable with respect to the Lebesque product measure. By the Fubini's theorem and the dominated convergence theorem, as $L \rightarrow \infty$,
\begin{align}\label{eq:comp2}
\cov(A_L, B_0) & \sim 2 \rme^{\frac{3}{2}}\gammaz(L) \int_{s=-1}^0 \int_{u=-1}^0 \rme^{\gammaz(s-u)}\rmd u \rmd s
\end{align}
Finally, we can express
\begin{eqnarray}
&& \cov(B_L, B_0) \\
&=&  \int_{s=-1}^0 \int_{t=-1}^0 \int_{u=-1}^0 \int_{v=-1}^0 \cov\left(\rme^{Z_H(s+L)}\rme^{Z_H(t+L)},\; \rme^{Z_H(v)}\rme^{Z_H(u)} \right) \rmd v\rmd u \rmd t \rmd s \notag \\
&=& \rme^2 \int_{s=-1}^0 \int_{t=-1}^0 \int_{u=-1}^0 \int_{v=-1}^0 \rme^{\gammaz(s-t)}\rme^{\gammaz(u-v)}\left[\rme^{\gammaz(L+s-v)}\rme^{\gammaz(L+s-u)}\rme^{\gammaz(L+t-v)}\rme^{\gammaz(L+t-u)} -1 \right]\rmd v\rmd u \rmd t \rmd s \notag\\
\label{def:cov_B0BL}
\end{eqnarray}
Applying a similar argument as for $\cov(A_L, B_0)$ and the dominated convergence theorem, we can show that as $L \rightarrow \infty$,
\begin{equation}\label{eq:comp4}
\cov(B_L, B_0) \sim 4 \rme^2 \gammaz(L)\int_{s=-1}^0 \int_{t=-1}^0 \int_{u=-1}^0 \int_{v=-1}^0 \rme^{\gammaz(s-t)}\rme^{\gammaz(u-v)}\rmd v\rmd u \rmd t \rmd s \;.
\end{equation}
\end{proof}

\noindent We will use Lemma \ref{lem:cov_ret_pred} and Lemma \ref{lem:slowvary_cov_A0AL_B0BL_A0BL} for the proof of Theorem \ref{thm:long memory squared returns}.

\subsection*{Proof of Theorem \ref{thm:long memory squared returns}}
\begin{proof}
For $L > 0$, conditional on $\Lambda$, $\Delta N(0)$ and $\Delta N(L)$ are independent. Therefore,
\begin{eqnarray}
\cov\left(r_0^2\;, r_L^2\right) &=& \cov\left(\esp[r_0^2|\Lambda]\;, \esp[r_L^2|\Lambda] \right) 
\notag \\
 &=& \mu^4 \cov\left(\esp[(\Delta N(0))^2|\Lambda]\;, \esp[(\Delta N(L))^2)|\Lambda]\right) + \mu^2\sigma_e^2 \cov\left(\esp[\Delta N(0)|\Lambda]\;, \esp[(\Delta N(L))^2|\Lambda]\right) \notag \\
 &+& \mu^2\sigma_e^2 \cov\left(\esp[\Delta N(L)|\Lambda]\;, \esp[(\Delta N(0))^2|\Lambda]\right) + \sigma_e^4 \cov\left(\esp[\Delta N(0)|\Lambda]\;, \esp[\Delta N(L)|\Lambda] \right), \label{eq:cond_cov_1}
\end{eqnarray}
where
\begin{eqnarray}\label{eq:cond_cov_sq_ct}
&&\cov\left(\esp[(\Delta N(0))^2|\Lambda],\; \esp[(\Delta N(L))^2|\Lambda] \right)\notag \\
&=& \lambda^2 \cov(A_0, A_L) + \lambda^3 \cov(A_L, B_0) + \lambda^3 \cov(A_0, B_L) + \lambda^4 \cov(B_L, B_0) \notag \\
&=&  \lambda^2 \cov(A_0, A_L) + 2\lambda^3 \cov(A_L, B_0) + \lambda^4 \cov(B_L, B_0)
\end{eqnarray}
\begin{eqnarray}\label{eq:cond_cov_ct_sq_ct}
 &&\cov\left(\esp[\Delta N(0)|\Lambda]\;, \esp[\Delta N(L)^2|\Lambda] \right) \notag \\
  &=& \cov\left(\esp[\Delta N(L)|\Lambda]\;, \esp[\Delta N(0)^2|\Lambda] \right)  \;,
\end{eqnarray}
and
\begin{eqnarray}\label{eq:cond_cov_ct}
 \cov\left(\esp[\Delta N(0)|\Lambda]\;, \esp[\Delta N(L)|\Lambda] \right) = \lambda^2 \cov(A_0, A_L) \;.
\end{eqnarray}
Using (\ref{eq:1st_cov}), (\ref{eq:cond_cov_sq_ct}), (\ref{eq:cond_cov_ct_sq_ct}), and (\ref{eq:cond_cov_ct}), $\cov(r_0^2, r_L^2)$ can be expressed as
\begin{eqnarray}
\cov(r_0^2, r_L^2)
                  &=& (\mu^2+\sigma^2_e)^2\lambda^2\cov(A_0, A_L) + 2\mu^2\lambda^3(\mu^2+\sigma^2_e)\cov(A_0, B_L) + \mu^4\lambda^4\cov(B_0, B_L)\;. \notag \\
                  \label{eq: cov_sq_ret_L}
\end{eqnarray}
Applying Lemma \ref{lem:slowvary_cov_A0AL_B0BL_A0BL}  to (\ref{eq: cov_sq_ret_L}), it follows that as $L \rightarrow \infty$,
\[
\lim_{L \rightarrow \infty}\frac{ \tilde{\gamma}(L)}{\gammaz(L)} = C \;.
\]
where
\begin{eqnarray}
 C &=& \mu^4 C_1 + 2\mu^2\sigma_e^2 C_2 + \sigma^4_e \lambda^2\rme \label{eq:def_C} \;, \\
 C_1 & = & \lambda^2\rme\left(1+4\lambda\rme^{\frac{1}{2}} \int_{s=-1}^0 \int_{u=-1}^0 \rme^{\gammaz(s-u)}\rmd u \rmd s    \right. \notag \\
    & + & \left. 4\lambda^2\rme \int_{s=-1}^0 \int_{t=-1}^0 \int_{u=-1}^0 \int_{v=-1}^0 \rme^{\gammaz(s-t)}\rme^{\gammaz(u-v)}\rmd v\rmd u \rmd t \rmd s \;, \right)\label{eq:def_C1}\;, \\
C_2 & = & \lambda^2\rme\left(2\lambda\rme^{\frac{1}{2}}  \int_{s=-1}^0 \int_{u=-1}^0 \rme^{\gammaz(s-u)}\rmd u \rmd s    +1 \right)\label{eq:def_C2} \;.
\end{eqnarray}
\end{proof}

\begin{lemma}\label{lem:convergence_real_var}
\[
\var\left(RV_{\tilde{t}-\tilde{H}, \tilde{t}}\right) \sim \tilde{C}_{\mu, \lambda,d,\sigma_e} \frac{ \ell(\tilde{H}m)(\tilde{H}m)^{2d+1}}{d(2d+1)}
\]
where $\tilde{C}_{\mu, \lambda,d,\sigma_e}> 0 $ is a function of $\mu$, $\lambda$, $d$, and $\sigma_e$ defined in (\ref{eq:def_C_var_realV}).
\end{lemma}

\begin{proof}
Let $\tilde{\gamma}(k)$ be the lag-k autocovariance function of squared return.
In Theorem \ref{thm:long memory squared returns}, we prove that as $k \to \infty$
\begin{align*}
\cov(r_0^2, r_k^2) & \sim  \left( \mu^4 C_1 + 2\mu^2\sigma_e^2 C_2 + \sigma^4_e \lambda^2\rme \right)\gammaz(k)\;,
\end{align*}
Therefore, $\tilde{\gamma}(k)$ can be represented as
\[
\tilde{\gamma}(k) = \tilde{C}_{\mu, \lambda,d,\sigma_e}\ell(k)k^{2d-1}\;,
\]
where $\ell(k)$ is a slowly varying function of fractional Gaussian which satisfies $\lim_{k \to \infty}\ell(k) = \frac{2d+1}{4d}$
and
\begin{align}\label{eq:def_C_var_realV}
 \tilde{C}_{\mu, \lambda,d,\sigma_e} &= \mu^4 C_1 + 2\mu^2\sigma_e^2 C_2 + \sigma^4_e \lambda^2\rme \;.
\end{align}
\noindent
Note that
\begin{eqnarray}\label{eq:var_agg_RV_expression}
\var\left(RV_{\tilde{t}-\tilde{H}, \tilde{t}}\right) &=& \cov \left(\sum_{j=(\tilde{t}-\tilde{H})m+1}^{\tilde{t}m}r_j^2 \;, \sum_{j=(\tilde{t}-\tilde{H})m+1}^{\tilde{t}m}r_j^2  \right)\notag \\
                            &=& \sum_{j=(\tilde{t}-\tilde{H})m+1}^{\tilde{t}m}\sum_{i=(\tilde{t}-\tilde{H})m+1}^{\tilde{t}m} \cov \left(r_j^2, r_i^2 \right) \notag\\
                            &=& \sum_{k=-(\tilde{H}m-1)}^{\tilde{H}m-1}\left(\tilde{H}m-|k|\right)\tilde{\gamma}(k) \notag \\
                            &=& \sum_{k=-(\tilde{H}m-1)}^{\tilde{H}m-1} (\tilde{H}m)\tilde{\gamma}(k) - \sum_{k=-(\tilde{H}m-1)}^{\tilde{H}m-1}|k|\tilde{\gamma}(k)\notag \\
                            &=& (\tilde{H}m)\var(r_0^2) + 2(\tilde{H}m)\sum_{k=1}^{\tilde{H}m-1}\tilde{\gamma}(k) - 2\sum_{k=1}^{\tilde{H}m-1}k \tilde{\gamma}(k)
\end{eqnarray}
Since $0 < d < \frac{1}{2}$, $(\tilde{H}m)\var(r_0^2) = o(\tilde{H}^{2d+1}) $.
Applying Proposition 2.2.1 again from Pipiras and Taqqu (2017) for $\sum_{k=1}^n \tilde{\gamma}(k)$ as $n \to \infty$,
$\sum_{k=1}^n \tilde{\gamma}(k) \sim \frac{\ell(n)n^{2d}}{2d} $.
We obtain
\begin{align}
\var\left(RV_{\tilde{t}-\tilde{H}, \tilde{t}}\right)& \sim  \tilde{C}_{\mu, \lambda,d,\sigma_e}\left\{\frac{2 \ell(\tilde{H}m)(\tilde{H}m)^{2d+1}}{2d} - \frac {2 \ell(\tilde{H}m)(\tilde{H}m)^{2d+1}}{2d+1} \right\} \notag \\
                            & =  \tilde{C}_{\mu, \lambda,d,\sigma_e} \frac{ \ell(\tilde{H}m)(\tilde{H}m)^{2d+1}}{d(2d+1)} \;.  \label{eq:denm_realV}
\end{align}

\end{proof}

\begin{lemma}\label{lem:convergence_agg_ret}
\[
\var \left(R_{\tilde{t}, \tilde{t}+\tilde{H}}\right) \sim  C_{\mu, \lambda} \frac{ \ell(\tilde{H}m)(\tilde{H}m)^{2d+1}}{d(2d+1)} \;,
\]
where $C_{\mu,\lambda} = \mu^2\lambda^2 \rme $.
\end{lemma}

\begin{proof}

Let $\gamma_r(k)=\cov(r_0, r_k)$ be the lag-k autocovariance function of returns.
By Equation(\ref{eq:cov_ro_ret_L}) and the convergence result in (\ref{eq:cov_counts}), as $k \to \infty$
\[
 \cov(r_0, r_k) \sim \mu^2\lambda^2 \rme \gammaz(k) \;.
\]
Therefore, $\gamma_r(k)$ has a representation of a slowly varying function as
\[
\gamma_r(k) = C_{\mu, \lambda}\ell(k)k^{2d-1} \;,
\]
where $C_{\mu,\lambda} = \mu^2\lambda^2 \rme $ and $\ell(k)$ is a slowly varying function of fractional Gaussian which satisfies $\lim_{k \to \infty}\ell(k) = \frac{2d+1}{4d}$
\\
\\
\noindent
Note
\begin{align}\label{eq:var_agg_Rt_expression}
\var\left(R_{\tilde{t}, \tilde{t}+\tilde{H}} \right)& = \cov\left(R_{\tilde{t}, \tilde{t}+\tilde{H}}\;, R_{\tilde{t}, \tilde{t}+\tilde{H}} \right) \notag \\
                                &= \cov \left(\sum_{j=\tilde{\tilde{t}}m+1}^{(\tilde{t}+\tilde{H})m} r_j\;,  \sum_{i=\tilde{\tilde{t}}m+1}^{(\tilde{t}+\tilde{H})m}r_i  \right) \notag \\
                                &= (\tilde{\tilde{H}}m)\var(r_0) + 2\sum_{k=1}^{\tilde{\tilde{H}}m-1}(\tilde{H}m-k)\gamma_r(k) \notag \\
                                & = (\tilde{H}m)\var(r_0)+ 2\sum_{k=1}^{\tilde{H}m-1}(\tilde{H}m)\gamma_r(k) - 2\sum_{k=1}^{\tilde{H}m-1}k\; \gamma_r(k)
\end{align}
Since $0 < d < \frac{1}{2}$, $(\tilde{H}m)\var(r_0) = o(\tilde{H}^{2d+1}) $.
Applying Proposition 2.2.1 again from Pipiras and Taqqu (2017) for $\sum_{k=1}^n \gamma_r(k)$ when $n \to \infty$,
we obtain
\begin{align}
\var \left(R_{\tilde{t}, \tilde{t}+\tilde{H}} \right) & \sim C_{\mu, \lambda}\left\{\frac{2 \ell(\tilde{H}m)(\tilde{H}m)^{2d+1}}{2d} - \frac {2 \ell(\tilde{H}m)(\tilde{H}m)^{2d+1}}{2d+1} \right\} \notag \\
                                & =  C_{\mu, \lambda} \frac{ \ell(\tilde{H}m)(\tilde{H}m)^{2d+1}}{d(2d+1)} \;.  \label{eq:denm_ret}
\end{align}
\end{proof}

\noindent We will use Lemmas \ref{lem:convergence_real_var} and \ref{lem:convergence_agg_ret} for the proof of Theorem \ref{thm:convergence_rho}.

\subsection*{Proof of Theorem \ref{thm:convergence_rho}}

\begin{proof}

Equation(\ref{eq:cov_r0_sq_ret}) and the convergence results of (\ref{eq:comp1}) and (\ref{eq:comp2}) imply that as $k \to \infty$,
\[
\cov(r_k, r^2_0) \sim C_{\mu, \lambda, d, \sigma_e} \gammaz(k)
\]
where
\begin{align}\label{eq:def_C_cov_r0_sq_rL}
C_{\mu,\lambda, d, \sigma_e} & = \mu^3 C_2 + \mu\sigma^2_e \lambda^2 \rme
\end{align}
Define $\gamma_c(k)= \cov(r_k, r^2_0)$. It can be represented as
\[
 \gamma_c(k)= C_{\mu, \lambda, d, \sigma_e}\ell(k)k^{2d-1} \;,
\]
where $\ell(k)$ is a slowly varying function of fractional Gaussian noise which satisfies $\lim_{k \to \infty}\ell(k) = \frac{2d+1}{4d}$.
\\
\\
From Equation(\ref{eq:cov_ret_pred}), $\cov\left(R_{t, t+\tilde{H}}, RV_{t-\tilde{H}, t}  \right)$ can be expressed as
\begin{eqnarray}
\cov\left(R_{\tilde{t}, \tilde{t}+\tilde{H}}, RV_{\tilde{t}-\tilde{H}, \tilde{t}}  \right) &=& \sum_{k=1}^{\tilde{H}m}k \gamma_c(k) +\sum_{k=\tilde{H}m+1}^{2\tilde{H}m-1}\left(2\tilde{H}m-k\right)\gamma_c(k) \notag \\
                                             &=& \sum_{k=1}^{\tilde{H}m}k \gamma_c(k) + \sum_{k=1}^{2\tilde{H}m-1}(2\tilde{H}m-k)\gamma_c(k) - \sum_{k=1}^{\tilde{H}m}(2\tilde{H}m-k)\gamma_c(k) \notag \\
                                             &=& \sum_{k=1}^{\tilde{H}m}k \gamma_c(k) + 2(\tilde{H}m)\sum_{k=1}^{2\tilde{H}m-1}\gamma_c(k) - 2(\tilde{H}m)\sum_{k=1}^{\tilde{H}m}\gamma_c(k) \notag \\
                                             &-& \sum_{k=1}^{2\tilde{H}m-1}k\gamma_c(k)+\sum_{k=1}^{\tilde{H}m}k\gamma_c(k) \label{eq:convergence_cov_ret_real_var}
\end{eqnarray}
Applying Proposition 2.2.1 from Pipiras and Taqqu (2017) in which they showed that as $n \to \infty$,
\[
\sum_{k=1}^n \gamma_c(k) \sim C_{\mu, \lambda, d, \sigma_e}\frac{\ell(n)n^{2d}}{2d} \;,
\]
we obtain
\begin{align}
  \cov\left(R_{\tilde{t}, \tilde{t}+\tilde{H}}, RV_{\tilde{t}-\tilde{H}, \tilde{t}}  \right) & \sim  \left\{\frac{\ell(\tilde{H}m)(\tilde{H}m)^{2d+1}}{2d+1} + \frac{(2\tilde{H}m)\ell(2\tilde{H}m-1)(2\tilde{H}m-1)^{2d}}{2d} \right. \notag \\
                                               & -  \left.\frac{(2\tilde{H}m)\ell(\tilde{H}m)(\tilde{H}m)^{2d}}{2d} \right. \notag \\
                                               &- \left. \frac{\ell(2\tilde{H}m-1)(2\tilde{H}m-1)^{2d+1}}{2d+1} + \frac{\ell(\tilde{H}m)(\tilde{H}m)^{2d+1}}{2d+1}\right\} C_{\mu, \lambda, d, \sigma_e} \notag \\
                                               & \sim  C_{\mu, \lambda, d, \sigma_e} \frac{\ell(\tilde{H}m)(\tilde{H}m)^{2d+1}}{d(2d+1)}\left[2^{2d}-1 \right] \;. \label{eq:num_cov}
\end{align}

\noindent
Combining Lemmas \ref{lem:convergence_real_var}, \ref{lem:convergence_agg_ret}, and Equation (\ref{eq:num_cov}), we show that as $\tilde{H} \to \infty$
\begin{align}
\corr\left(R_{\tilde{t},\tilde{t}+\tilde{H}}, \; RV_{\tilde{t}-\tilde{H}, \tilde{t}}\right) & \to \frac{C_{\mu,\lambda, d, \sigma_e}}{\sqrt{\tilde{C}_{\mu, \lambda,d,\sigma_e}C_{\mu,\lambda}}} \left(2^{2d}-1\right) \;, \label{eq:rho}
\end{align}
Following the definitions of $C_1$, $C_2$, $\tilde{C}_{\mu, \lambda,d,\sigma_e}$, and $C_{\mu,\lambda, d, \sigma_e}$  in (\ref{eq:def_C1}), (\ref{eq:def_C2}), (\ref{eq:def_C_var_realV}), (\ref{eq:def_C_cov_r0_sq_rL}) and Lemma \ref{lem:convergence_agg_ret}, it can be shown that
\begin{align*}
\left( C_{\mu,\lambda, d, \sigma_e}\right)^2 = \tilde{C}_{\mu, \lambda,d,\sigma_e}C_{\mu, \lambda} \;,
\end{align*}
Therefore,
\[
\corr\left(R_{\tilde{t},\tilde{t}+\tilde{H}},\; RV_{\tilde{t}-\tilde{H}, \tilde{t}}\right)  \to  \left(2^{2d}-1\right) \;.
\]

\end{proof}

\subsection*{Proof of Lemma \ref{lem:count_property}}
\begin{proof}
For $L > 0$, conditional on $\Lambda$, $\Delta N(0)$ and $\Delta N(L)$ are independent. Therefore, the lag-L autocovariance of $\Delta N(t)$ is equal to

\begin{eqnarray*}
\cov \left(\mathbb{E}[\Delta N(L)|\Lambda]\;, \mathbb{E}[(\Delta N(0))|\Lambda] \right) &=&
 \lambda^2\cov \left(\int_{s=L-1}^{L}\rme^{Z_H(s)}\rmd s \;, \int_{t=-1}^{0}\rme^{Z_H(s)}\rmd s \right) \notag \\
 &=& \lambda^2 \int_{s=L-1}^{L}\int_{t=-1}^{0}\cov(\rme^{Z_H(s)}, \rme^{Z_H(t)})\rmd t\rmd s \notag \\
 &=& \lambda^2  \rme \int_{s=-1}^0 \int_{t=-1}^0 \left\{\rme^{\gammaz(L+s-t)}-1  \right\} \rmd t \rmd s \label{lagL_autocov_count}\\
\end{eqnarray*}

We can express
\begin{eqnarray*}
\var\left(\Delta N(t) \right) &=& \var \left( \mathbb{E} [\Delta N(t)|\Lambda] \right) + \mathbb{E}[\var(\Delta N(t)|\Lambda])] \\
     &=& \lambda \rme^{\frac{1}{2}} + \lambda^2  \rme \int_{s=-1}^0 \int_{t=-1}^0 \left\{\rme^{\gammaz(s-t)}-1  \right\} \rmd t \rmd s
\end{eqnarray*}
\end{proof}

\begin{lemma}\label{lem:gammaz_fun}
Let $Z_H(t) = B_H(ct)-B_H(ct-1)$, where $B_H(t)$ is the fractional Brownian motion with Hurst index $\frac{1}{2} < H < 1$, and
$c > 0$. The lag-$r$ autocovariance function $\gammaz(r)$ of $Z_H$ is equal to
  \begin{equation}
   \gammaz(r)= \frac{1}{2}\left[|cr+1|^{2H}-2|cr|^{2H} + |cr-1|^{2H} \right]\;.
  \end{equation}
Moreover, for all $ r \in \mathbb{R}$ , $\gammaz(r) > 0$.
\end{lemma}
\begin{proof}
For any real lag $r \ge 0$, the autocovariance of fraction Brownian motion can be expressed as
\begin{equation}\label{eq:fBM_autocov}
\cov\left(B_H(t+r)\;, B_H(t) \right)=
\frac{1}{2}\left[|t+r|^{2H}-|r|^{2H} + |t|^{2H}\right]\;.
\end{equation}
Using Equation (\ref{eq:fBM_autocov}), the lag-$r$ autocovariance function of $Z_H(t)$ is equal to
\begin{eqnarray}
\gammaz(r) &=& \cov(Z_H(t+r), Z_H(t)) \notag\\
                   &=& \cov\left(B_H(c(t+r))- B_H(c(t+r)-1)\;, B_H(ct)-B_H(ct-1)\right) \notag\\
                   &=& \cov\left(B_H(c(t+r)), B_Hct)\right)+\cov\left(B_H(c(t+r)-1)), B_H(ct-1)\right) \notag\\
                   &-& \cov\left(B_H(c(t+r)-1), B_H(ct)\right) - \cov\left(B_H(c(t+r)), B_H(ct-1)\right) \notag\\
                   &=& \frac{1}{2}\left[|cr+1|^{2H}-2|cr|^{2H} + |cr-1|^{2H} \right] \label{eq:cont_time_frac_gauss_v2} \;.
\end{eqnarray}
Define the function $f(x)=x^{2H}$. We can express
\begin{equation}\label{eq:autoCV_ZH_rep}
|cr+1|^{2H}-2|cr|^{2H} + |cr-1|^{2H} = (cr)^{2H-2} \frac{f(1+\frac{1}{cr})-2f(1)+f(1-\frac{1}{cr})}{\frac{1}{(cr)^2}}
\end{equation}
Consider $x > 0$ and observe that if $\frac{1}{2} < H < 1$, $f'(x)=2H x^{2H-1} > 0 $ and $f''(x)=2H (2H-1)x^{2H-2} > 0$.
Hence $f(x)$ is convex. By Jensen's inequality, $\frac{1}{2}f(1+\frac{1}{cr}) + \frac{1}{2}f(1-\frac{1}{cr}) > f(1)$.
Therefore, $f(1+\frac{1}{cr})-2f(1)+f(1-\frac{1}{cr}) > 0$ and $\gammaz(r) > 0$.
If $r < 0$, it can be shown that $\gammaz(r) = \gammaz(-r) > 0$.
Finally, if $r=0$, then by Equation (\ref{eq:autoCV_ZH_rep}), $\gammaz(r) = 1$.
\end{proof}

\noindent We will use Lemma \ref{lem:gammaz_fun} for the proof of Lemma \ref{lem:exp_gammaz_d0}.

\begin{lemma}\label{lem:exp_gammaz_d0}
If $d=0$,
%
$ \forall L \in \Nset$, $-1 < s < 0$, and $-1 < t < 0$,
\\
\\
(i) if $0 < c < \frac{1}{2}$
\\
\[
\rme^{\gammaz(L+s-t)}=
\left\{
\begin{array}
{r@{\quad:\quad}l}
 1                     & \mbox{if}\;\;  L  > \lfloor \frac{1}{c}+1 \rfloor  \;, \\
 \rme^{1-c|L+ s-t|}    & \mbox{if}\;\;  1 \le L \le \lfloor \frac{1}{c} +1\rfloor \;.
\end{array} \right.
\]
\\
\\
(ii) if $\frac{1}{2} \le c < 1$
\\
\[
\rme^{\gammaz(L+s-t)}=
\left\{
\begin{array}
{r@{\quad:\quad}l}
 1                     & \mbox{if}\;\;  L > 2 \;, \\
 \rme^{1-c|L+ s-t|}  & \mbox{if}\;\;  1 \le L \le 2 \;.
\end{array} \right.
\]
\\
\\
(iii) if $c \ge 1$
\\
\[
\rme^{\gammaz(L+s-t)}=
\left\{
\begin{array}
{r@{\quad:\quad}l}
 1                     & \mbox{if}\;\;  L \ge 2 \;, \\
 \rme^{1-c|L+ s-t|}  & \mbox{if}\;\;  L =1 \;.
\end{array} \right.
\]
\\

\end{lemma}
\begin{proof}
If $d=0$, i.e. $H=\frac{1}{2}$, $\gammaz(r)$ can be expressed as
\begin{equation}\label{eq:fGN_autocov_d0}
  \gammaz(r) = \frac{1}{2}c\left[\left|r+\frac{1}{c}\right|-2|r| + \left|r-\frac{1}{c}\right|\right] \;.
\end{equation}
If $|r| \ge \frac{1}{c}$, $\gammaz(r) = 0$.  If $|r| < \frac{1}{c}$, $\gammaz(r) = 1-c|r|$. Therefore,
\[
\rme^{\gammaz(r)}=
\left\{
\begin{array}
{r@{\quad:\quad}l}
 1                  & \mbox{if}\;\;  |r| \ge \frac{1}{c} \;, \\

 \rme^{1-c|r|}      & \mbox{if}\;\;  |r| < \frac{1}{c} \;.
\end{array} \right.
\]
\\
Since $-1 \le s \le 0$, $-1 \le t \le 0$, $|s-t| \le 1$. $ \forall L \in \Nset$, $0 \le L+s-t \le L+1$.
Moreover, if $|L+s-t| \ge \frac{1}{c}$, $\rme^{\gammaz(L+s-t)} = 1$. If $|L+s-t| < \frac{1}{c}$, $\rme^{\gammaz(L+s-t)} = \rme^{1-c(L+s-t)}$.
Therefore, depending on the range of $c$, $\rme^{\gammaz(L+s-t)}$ has different expressions:
\\
\\
(i) if $0 < c < \frac{1}{2}$
\\
\[
\rme^{\gammaz(L+s-t)}=
\left\{
\begin{array}
{r@{\quad:\quad}l}
 1                     & \mbox{if}\;\;  L  > \lfloor \frac{1}{c}+1 \rfloor  \;, \\
 \rme^{1-c|L+ s-t|}    & \mbox{if}\;\;  1 \le L \le \lfloor \frac{1}{c} +1\rfloor \;.
\end{array} \right.
\]
\\
\\
(ii) if $\frac{1}{2} \le c < 1$
\\
\[
\rme^{\gammaz(L+s-t)}=
\left\{
\begin{array}
{r@{\quad:\quad}l}
 1                     & \mbox{if}\;\;  L > 2 \;, \\
 \rme^{1-c|L+ s-t|}  & \mbox{if}\;\;  1 \le L \le 2 \;.
\end{array} \right.
\]
\\
\\
(iii) if $c \ge 1$
\\
\[
\rme^{\gammaz(L+s-t)}=
\left\{
\begin{array}
{r@{\quad:\quad}l}
 1                     & \mbox{if}\;\;  L \ge 2 \;, \\
 \rme^{1-c|L+ s-t|}  & \mbox{if}\;\;  L =1 \;.
\end{array} \right.
\]


\end{proof}

\noindent We will use Lemmas \ref{lem:gammaz_fun}, \ref{lem:exp_gammaz_d0}, \ref{lem:autocov_ret}, and Corollary \ref{cor:var_rt} for the proof of Theorem \ref{thm:convergence_rho_shortM}.

\subsection*{Proof of Theorem \ref{thm:convergence_rho_shortM}}
\begin{proof}
Assume $d=0$. Let
\begin{equation}\label{eq:def_h}
h =
\left\{
\begin{array}
{l@{\quad:\quad}l}
 \lfloor \frac{1}{c} \rfloor +1                                   & \mbox{if}\;\;  0 < c < 1 \;, \\
 \max(\lfloor \frac{1}{c}-1 \rfloor + 1\;,\; 1)                   & \mbox{if}\;\;   c  \ge 1 \;,
\end{array} \right.
\end{equation}
\noindent Then $\forall L \in \Nset$ and $L > h$, $\gammaz(L)=0$ and thus $\cov\left(A_0, B_L \right)=\cov\left(A_L, A_0 \right)=\cov\left(B_0, B_L \right)=0$.
From (\ref{eq:cov_r0_sq_ret}), if $L > h$, $\cov(r_L, r_0^2)=0$, .
Therefore, for all $\tilde{H}m > h$, Equation (\ref{eq:cov_ret_pred}) becomes
\begin{equation}\label{eq:rho_d0_numerator}
 \cov \left(R_{t, t+\tilde{H}},\; RV_{t-\tilde{H},t} \right) = \sum_{L=1}^h\cov(r_L, r^2_0)L \;,
\end{equation}
and (\ref{eq:var_agg_RV_expression}) and (\ref{eq:var_agg_Rt_expression}) can be expressed as
\begin{equation}\label{eq:rho_d0_denm_ret}
\var\left(R_{t, t+\tilde{H}}\right)=(\tilde{H}m)\var(r_0) + 2\sum_{L=1}^h(\tilde{H}m-L)\cov(r_0, r_L)\;,
\end{equation}
and
\begin{equation}\label{eq:rho_d0_denm_real_var}
\var(RV_{t-\tilde{H},t})=(\tilde{H}m)\var(r_0^2) + 2\sum_{L=1}^h(\tilde{H}m-L)\cov(r^2_0, r^2_L) \;.
\end{equation}
By Equations (\ref{eq:rho_d0_numerator}), (\ref{eq:rho_d0_denm_ret}), and (\ref{eq:rho_d0_denm_real_var}),
the correlation between $R_{t, t+\tilde{H}}$ and $RV_{t-\tilde{H},t}$ can be expressed as
\begin{eqnarray}
&&\corr\left(R_{t, t+\tilde{H}}, RV_{t-\tilde{H},t}\right) \notag \\
 &=& \frac{\displaystyle\sum_{L=1}^h\cov(r_L, r^2_0)L}{(\tilde{H}m)\sqrt{\left[\var(r_0)+2\displaystyle \sum_{L=1}^h\left(1-\frac{L}{\tilde{H}m}\right)\cov(r_0,r_L)\right]
     \left[\var(r^2_0)+2\displaystyle \sum_{L=1}^h\left(1-\frac{L}{\tilde{H}m}\right)\cov(r^2_0, r^2_L)\right]}} \notag\\
     \label{eq:rho_d0_ratio}
\end{eqnarray}
\\
Lemma $\ref{lem:gammaz_fun}$ shows that when $d=0$, $\forall r \in \Rset$, $\gammaz(r) > 0$. Hence by Equations (\ref{eq:cov_ro_ret_L}) and (\ref{eq: cov_sq_ret_L}), for $L \in \Nset$, $\cov(r_0, r_L) > 0$ and $\cov(r^2_0, r^2_L) > 0$.
\\
Therefore, for sufficiently large $\tilde{H}$, if $d=0$,
\[
\var(r_0)+2\displaystyle \sum_{L=1}^h\left(1-\frac{L}{\tilde{H}m}\right)\cov(r_0,r_L) > \var(r_0)\;,
\]
and
\[
\var(r^2_0)+2\displaystyle \sum_{L=1}^h\left(1-\frac{L}{\tilde{H}m}\right)\cov(r^2_0, r^2_L) > \var(r^2_0) \;.
\]
Hence,
\[\sqrt{\left[\var(r_0)+2\displaystyle \sum_{L=1}^h\left(1-\frac{L}{\tilde{H}m}\right)\cov(r_0,r_L)\right]
     \left[\var(r^2_0)+2\displaystyle \sum_{L=1}^h\left(1-\frac{L}{\tilde{H}m}\right)\cov(r^2_0, r^2_L)\right]}
\] is positive and bounded.
\\
\\
From Equation (\ref{eq:rho_d0_ratio}), as $\tilde{H} \to \infty$, $\corr\left(R_{t, t+\tilde{H}}, RV_{t-\tilde{H},t}\right) \to 0 $.

\end{proof}

\begin{lemma}\label{lem:sample_mean_ct_covergence d0}
If $d=0$, the sample mean $\overline{\Delta N_n}=n^{-1}\sum_{t=1}^n \Delta N(t)$  of the counts process is a consistent estimator of $\mu_{\Delta N}=\esp[\Delta N(t)]$.
\end{lemma}

\begin{proof}
By Equation (\ref{eq:autocov_count}) and Lemma \ref{lem:exp_gammaz_d0}, there exists $h \in \Nset$ such that when $d=0$, $\forall L > h$, $\cov(\Delta N(L), \Delta N(0)) = 0$. Hence as $n \to \infty$, $\var(\overline{\Delta N_n}) \to 0 $. By Chebyshev's inequality, $\forall \epsilon > 0$,
$\pr(|\overline{\Delta N_n} - \mu_{\Delta N}| \ge \epsilon) \le \frac{\var(\overline{\Delta N_n})}{\epsilon^2} \to 0$, as $n \to \infty$.
Therefore, $\overline{\Delta N_n} \stackrel{p}{\rightarrow} \mu_{\Delta N} $.
\end{proof}

\begin{lemma}\label{lem:sample_mean_ret_convergence_d0}
If $d=0$, the sample mean $\bar{r}_n$ of the return process is a consistent estimator of $\mu_r = \esp[r_t]$.
\end{lemma}

\begin{proof}
By Equation (\ref{eq:cov_ro_ret_L}) and Lemma \ref{lem:exp_gammaz_d0}, there exists $h \in \Nset$ such that when $d=0$, $\forall L > h$, $\cov(r_0, r_L) = 0$.
Therefore, as $n \to \infty$, $\var(\bar{r}_n ) \to 0$. By Chebyshev's inequality, $\forall \epsilon > 0$, $\pr(|\bar{r}_n  - \mu_r|\ge \epsilon) \le \frac{\var(\bar{r}_n ) }{\epsilon^2} \to 0$, as $n \to \infty$.  Therefore, $\bar{r}_n \stackrel{p}{\rightarrow} \mu_r$.
\end{proof}

\begin{lemma}\label{lem:sample_mean_rsq_convergence_d0}
If d=0, the sample variance $\hat{\sigma}_r^2$ of the return process is a consistent estimator of ${\sigma}_r^2$.
\end{lemma}

\begin{proof}
The sample variance $\hat{\sigma}_r^2$ can be expressed as
\[
\hat{\sigma}_r^2 = \frac{n}{n-1}\left[\frac{1}{n}\sum_{t=1}^n (r_t - \mu_r)^2 - (\bar{r}_n - \mu_r)^2  \right] \;.
\]
From Lemma \ref{lem:sample_mean_ret_convergence_d0}, if $d=0$, $\bar{r}_n -\mu_r \stackrel{p}{\rightarrow} 0$. It follows $(\bar{r}_n - \mu_r)^2 \stackrel{p}{\rightarrow} 0$. Moreover, $Z_H(t)=B_H(ct)- B_H(ct-1)$ is the increment of fractional Brownian Motion. When $d=0$, there exists $h \in \Nset$ such that $\forall L > h$, $Z_H(t)$ and $Z_H(t+L)$ are independent.
Since the covariance function of $r_t$ is a function of $\gammaz$, $\forall L > h$ $r_t$ and $r_{t+L}$ are independent. Hence $r^2_t$ and $r^2_{t+L}$ are independent and $\cov(r^2_t, r^2_{t+L}) = 0 \; \forall L \ge h$.
Let $Y_t = (r_t - \mu_r)^2$, $\bar{Y}_n = \frac{1}{n} \sum_{t=1}^n Y_t$, and $\esp[Y_t]=\sigma^2_r$.
When $d=0$, as $n \to \infty$, $\var(\bar{Y}_n) \to 0$. By Chebyshev's inequality, $\forall \epsilon > 0$,
$\pr(|\bar{Y}_n - \sigma^2_r|\ge \epsilon)\le \frac{\var(\bar{Y}_n)}{\epsilon^2}  \to 0$.  Hence $\bar{Y}_n \stackrel{p}{\rightarrow} \sigma^2_r $.
Combining the convergence results of $(\bar{r}_n - \mu_r)^2$ and $\frac{1}{n}\sum_{t=1}^n (r_t-\mu_r)^2$, we obtain
\[
\frac{n}{n-1}\left[\frac{1}{n}\sum_{t=1}^n (r_t-\mu_r)^2 - (\bar{r}_n - \mu_r)^2  \right] \stackrel{p}{\rightarrow} \sigma^2_r \;.
\]
\end{proof}

\begin{lemma}\label{lem:sample_mean_ct_sq_convergence_d0}
If d=0, the sample variance $\hat{\sigma}_{\Delta N}^2$ of the counts process is a consistent estimator of ${\sigma}_{\Delta N}^2$.
\end{lemma}

\begin{proof}
The sample variance of $\Delta N(t)$ can be expressed as
\[
\hat{\sigma}_{\Delta N}^2 = \frac{n}{n-1}\left[\frac{1}{n}\sum_{t=1}^n (\Delta N(t) - \mu_{\Delta N})^2 - (\overline{\Delta N_n} - \mu_{\Delta N})^2  \right] \;,
\]
From Lemma \ref{lem:sample_mean_ct_covergence d0}, if $d=0$, $\overline{\Delta N_n} - \mu_{\Delta N} \stackrel{p}{\rightarrow} 0 $. It follows $(\overline{\Delta N_n}- \mu_{\Delta N})^2 \stackrel{p}{\rightarrow} 0$. Since the covariance function of $\Delta N$ is a function of $\gammaz$, if $d=0$,
there exists $L \ge h$ such that $\Delta N(t)$ and $\Delta N(t+L)$ are independent. Hence $(\Delta N(t))^2$ and $(\Delta N(t+L))^2$ are independent.
Let $\omega_t = (\Delta N(t)-\mu_{\Delta N})^2$, $\bar{\omega}_n = \frac{1}{n}\sum_{t=1}^n (\Delta N(t)-\mu_{\Delta N})^2$, and $\mu_{\omega}=\esp[\omega_t] = \sigma^2_{\Delta N}$. When $d=0$, as $n \to \infty$, $\var(\bar{\omega}_n) \to 0$.
By Chebyshev's inequality, $\forall \epsilon > 0$,
$\pr(|\bar{\omega}_n - \sigma^2_{\Delta N}|\ge \epsilon) \le \frac{\var(\bar{\omega}_n)}{\epsilon^2} \to 0$.  Hence $\bar{\omega}_n \stackrel{p}{\rightarrow} \sigma^2_{\Delta N}$.
Combining the convergence results of $(\overline{\Delta N_n} - \mu_{\Delta N})^2$ and $\frac{1}{n}\sum_{t=1}^n (\Delta N(t)-\mu_{\Delta N})^2$, we obtain
\[
\frac{n}{n-1}\left[\frac{1}{n}\sum_{t=1}^n (\Delta N(t)-\mu_{\Delta N})^2 - (\overline{\Delta N_n} - \mu_{\Delta N})^2  \right] \stackrel{p}{\rightarrow} \sigma^2_{\Delta N} \;.
\]
\end{proof}

\begin{lemma}\label{lem:sample_autocov_ct_convergence_d0}
If d=0, the sample lag-L autocovariance $\hat{\gamma}_{\Delta N}(L)$ of the counts is a consistent estimator of $\gamma_{\Delta N}(L)$.
\end{lemma}

\begin{proof}
By definition the sample lag-L autocovariance of $\Delta N(t)$ is
\begin{eqnarray*}
\hat{\gamma}_{\Delta N}(L) = \frac{1}{n}\sum_{t=1}^{n-L} \left[\Delta N(t) - \overline{\Delta N(t)})(\Delta N(t+L) - \overline{\Delta N(t+L)}\right]
\end{eqnarray*}
where $\overline{\Delta N_{1, n-L}} = \frac{1}{n-L}\sum_{t=1}^{n-L}\Delta N(t)$ and $\overline{\Delta N_{L+1, n}} = \frac{1}{n-L}\sum_{t=L+1}^{n}\Delta N(t)$.
\\
Hence $\hat{\gamma}_{\Delta N}(L)$ can be expressed as
\begin{eqnarray}
\hat{\gamma}_{\Delta N}(L) 
         &=& \frac{n-L}{n}\left\{\frac{1}{n-L}\sum_{t=1}^{n-L} \left[(\Delta N(t) - \mu_{\Delta N})(\Delta N(t+L) - \mu_{\Delta N})\right] \right. \notag \\
                        &-& \frac{1}{n-L}(\overline{\Delta N_{1, n-L}} - \mu_{\Delta N})\sum_{t=1}^{n-L}(\Delta N(t+L)-\mu_{\Delta N}) \notag\\
                        &-& \frac{1}{n-L}(\overline{\Delta N_{L+1, n}} - \mu_{\Delta N})\sum_{t=1}^{n-L}(\Delta N(t)-\mu_{\Delta N}) \notag \\
                        &+& \left.(\overline{\Delta N_{1, n-L}} - \mu_{\Delta N})(\overline{\Delta N_{L+1, n}} - \mu_{\Delta N})  \right\} \notag\\
         &=&  \frac{n-L}{n}\left\{\frac{1}{n-L}\sum_{t=1}^{n-L} \left[(\Delta N(t) - \mu_{\Delta N})(\Delta N(t+L) - \mu_{\Delta N})\right] \right. \notag\\
         &-&   \left.(\overline{\Delta N_{1, n-L}} - \mu_{\Delta N})(\overline{\Delta N_{L+1, n}} - \mu_{\Delta N})  \right\} \;.
\end{eqnarray}
By Lemma \ref{lem:sample_mean_ct_covergence d0},
$\overline{\Delta N(t)} \stackrel{p}{\rightarrow} \mu_{\Delta N}$ and $\overline{\Delta N(t+L)} \stackrel{p}{\rightarrow} \mu_{\Delta N}$.
Hence $(\overline{\Delta N(t)}-\mu_{\Delta N})(\overline{\Delta N(t+L)}-\mu_{\Delta N}) \stackrel{p} {\rightarrow} 0$.
\\
Let $\tilde{W}_t = (\Delta N(t)-\mu_{\Delta N})(\Delta N(t+L)-\mu_{\Delta N})$ and $\mu_{\tilde{W}} = \esp[\tilde{W}_t]=\gamma_{\Delta N}(L)$.
Define $\tilde{W}_n = \frac{1}{n-L}\sum_{t=1}^{n-L}\tilde{W}_t$.
When $d=0$, there exists $ L > h$ such that $\Delta N(t)\Delta N(t+L)$ and $\Delta N(t+i)\Delta N(t+i+L)$ are independent $\forall$ $i > L+h$.
Hence $\cov(\tilde{W}_t, \tilde{W}_{t+i})=0$ $\forall$ $i > L+h$. When $d=0$, as $n \to \infty$, $\var(\tilde{W}_n) \to 0$.
By Chebyshev's inequality, $\pr(|\tilde{W}_n - \mu_{\tilde{W}}| \ge \epsilon )\le \frac{\var(\tilde{W}_n)}{\epsilon^2} \to 0$.
Combining the convergence results of $(\overline{\Delta N(t)}-\mu_{\Delta N})(\overline{\Delta N(t+L)}-\mu_{\Delta N}) \stackrel{p} {\rightarrow} 0$ and
$\pr(|\tilde{W}_n - \mu_{\tilde{W}}| \ge \epsilon ) \to 0$, we obtain
$\hat{\gamma}_{\Delta N}(L) \stackrel{p} \rightarrow \gamma_{\Delta N}(L)$.
\end{proof}
\noindent We will use Lemmas \ref{lem:sample_mean_ct_covergence d0}, \ref{lem:sample_mean_ret_convergence_d0}, \ref{lem:sample_mean_rsq_convergence_d0},
\ref{lem:sample_mean_ct_sq_convergence_d0}, and \ref{lem:sample_autocov_ct_convergence_d0} for the proof of Theorem \ref{thm:consistent_estimators}.

\subsection*{Proof of Theorem \ref{thm:consistent_estimators}}
\begin{proof}
By Lemmas \ref{lem:sample_mean_ret_convergence_d0} and \ref{lem:sample_mean_ct_covergence d0}, $\bar{r}_n \stackrel{p}{\rightarrow} \mu_r$ and $\overline{\Delta N_n} \stackrel{p}{\rightarrow} \mu_{\Delta N}$.
Hence the joint vector
$\left(\bar{r}_n\;, \overline{\Delta N_n} \right) \stackrel{p}{\rightarrow} \left( \mu_r\;, \mu_{\Delta N} \right)$.
\\
Define $g(\bar{r}_n\;, \overline{\Delta N_n})=\frac{\bar{r}_n}{\overline{\Delta N_n}}$. It follows that
$g(\bar{r}_n\;, \overline{\Delta N_n})\stackrel{p}{\rightarrow} g(\mu_r\;, \mu_{\Delta N})$.  Therefore, $\hat{\mu}\stackrel{p}{\rightarrow} \mu$.
\\
\\
Since $\hat{\lambda} = \frac{\overline{\Delta N_n}}{e^{\frac{1}{2}}}$, by Lemma \ref{lem:sample_mean_ct_covergence d0}, $\hat{\lambda}\stackrel{p}{\rightarrow}
\frac{\mu_{\Delta N}}{\rme^{\frac{1}{2}}}=\lambda$.
\\
\\
By Lemmas \ref{lem:sample_mean_rsq_convergence_d0} and \ref{lem:sample_mean_ct_sq_convergence_d0}, $\hat{\sigma}^2_r \stackrel{p}{\rightarrow}\sigma^2_r $
and $\hat{\sigma}^2_{\Delta N} \stackrel{p}{\rightarrow}\sigma^2_{\Delta N}$.
From (\ref{eq:est_sigma_e}), $\hat{\sigma}^2_e$ is a function of $\hat{\sigma}^2_r$, $\hat{\sigma}^2_{\Delta N}$, $\hat{\mu}$, and $\hat{\lambda}$.
It follows that $\hat{\sigma}^2_e \stackrel{p}{\rightarrow} \sigma^2_e$.
\\
\\
We substitute $\hat{\lambda}$, $d=0$ and $L=1$ into Equation (\ref{eq:autocov_count}) to evaluate $\gamma_{\Delta N}(1)$ under $d=0$ and estimate $c$ by solving the equation
\begin{equation}
\gamma_{\Delta N}(1)-\hat{\gamma}_{\Delta N}(1) = 0 \;,
\end{equation}
where $\hat{\gamma}_{\Delta N}(1)$ is the sample lag-one autocovariance of the count. Therefore, we can express $\hat{c}$ as a function of $\hat{\lambda}$ and $\hat{\gamma}_{\Delta N}(1)$. Let $\hat{c} = g_c(\hat{\lambda},\hat{\gamma}_{\Delta N}(1))$ and $c = g_c(\lambda,\gamma_{\Delta N}(1))$.  From Lemma \ref{lem:sample_autocov_ct_convergence_d0},
$\hat{\gamma}_{\Delta N}(1) \stackrel{p}{\rightarrow} \gamma_{\Delta N}(1)$. Using the fact $\hat{\lambda} \stackrel{p}\rightarrow \lambda$, we obtain $g_c(\hat{\lambda},\hat{\gamma}_{\Delta N}(1)) \stackrel{p}{\rightarrow}g_c(\lambda, \gamma_{\Delta N}(1))$.

\end{proof}

\begin{lemma}\label{lem:cov_R_H}
Let $\gamma_R(L) = \cov(R_{\tilde{t}, \tilde{t}+\tilde{H}}\;, R_{\tilde{t}+L,\; \tilde{t}+L+\tilde{H}})$ and
$\gamma_r(k) = \cov(r_t, r_{t+k})$. If $d=0$,
$\forall\; 0 < L \le \tilde{H}$, \\
i)if $\tilde{H} - L $ is odd,
\begin{eqnarray}\label{eq:cov_R_H_odd}
\gamma_R(L) &=& \left[(\tilde{H}-L)m\right]\left[\var(r_t) + 2\sum_{k=1}^{m}\gamma_r(k)\right] + 2\sum_{k=m+1}^{(\tilde{H}-L)m}[(\tilde{H}-L+1)m-k]\gamma_r(k) \notag \\
 &+& \sum_{k=(\tilde{H}-L)m+1}^{(\tilde{H}+L)m-1}[(\tilde{H}+L)m-k]\gamma_r(k)\;.
\end{eqnarray}
ii) if $\tilde{H}-L$ is even,
\begin{eqnarray}\label{eq:cov_R_H_even}
\gamma_R(L) &=& \left[(\tilde{H}-L)m\right]\left[\var(r_t) + 2\sum_{k=1}^{(\tilde{H}-L)m}\gamma_r(k)\right] + \sum_{k=(\tilde{H}-L)m+1}^{(\tilde{H}+L)m-1}[(\tilde{H}+L)m-k]\gamma_r(k)\;. \notag \\
\end{eqnarray}
It follows $\gamma_R(L) = O(\tilde{H})$.
\end{lemma}

\begin{proof}
\begin{eqnarray}
\cov(R_{\tilde{t}, \tilde{t}+\tilde{H}}\;, R_{\tilde{t}+L,\; \tilde{t}+L+\tilde{H}}) &=& \sum_{j=\tilde{t}m+1}^{(\tilde{t}+\tilde{H})m} \sum_{i=\tilde{t}m+1}^{(\tilde{t}+\tilde{H})m} \cov(r_j, r_{i+Lm}) \notag \\
&=& \sum_{j=1}^{\tilde{H}m}\sum_{i=1}^{\tilde{H}m}\cov(r_j, r_{i+Lm}) \label{eq:cov_R_H}
\end{eqnarray}
Define $k=|i+Lm-j|$. If $\tilde{H}-L$ is odd, then (\ref{eq:cov_R_H}) has the expression of (\ref{eq:cov_R_H_odd}).
If $\tilde{H}-L$ is even, then (\ref{eq:cov_R_H}) has the expression of (\ref{eq:cov_R_H_even}).
\end{proof}

\begin{corollary}\label{coro:d0_cov_R_formula}
If $L=\tilde{H}$, then $\gamma_R(\tilde{H}) = \sum_{k=1}^{2\tilde{H}m-1} [2\tilde{H}m-k]\gamma_r(k)$.
When $d=0$, $r_t$ is $h$-dependent, and thus $\gamma_R(\tilde{H}) = \sum_{k=1}^{h} [2\tilde{H}m-k]\gamma_r(k)$.
Hence $R_{\tilde{t}, \tilde{t}+\tilde{H}}$ is $\tilde{H}+h/m$ dependent.
\end{corollary}

\begin{lemma}\label{lem:var_R_square}
If $d=0$,
\begin{align*}
\var \left(R^2_{\tilde{t}, \tilde{t}+\tilde{H}}\right) & \sim \tilde{H}^3 C^{\ast} \;, \\
\var \left(\tilde{R}^2_{\tilde{t}, \tilde{t}+\tilde{H}}\right) & \sim  \tilde{H}^3 C^{\ast\ast}
\end{align*}
where $C^{\ast}$ and $C^{\ast\ast}$ are positive constants.
\end{lemma}
\begin{proof}
Let $\R^{\ast} = \R - \esp[\R] = \sum_{\tilde{t} m+1}^{(\tilde{t}+\tilde{H})m}(r_t - \esp[r_t])$.
Let $r^{\ast}_t = r_t - \esp[r_t]$. Then ${\R^{\ast}}^2 = \sum_{i=\tilde{t} m+1}^{(\tilde{t}+\tilde{H})m}\sum_{j=\tilde{t} m+1}^{(\tilde{t}+\tilde{H})m}r^{\ast}_i r^{\ast}_j$.
Hence
\begin{equation}
\var(\R^2) = \var(\R^{\ast2} ) + 4\esp[\R]\cov(\R, \R^2) + 4(\esp[\R])^2\var(\R)\;,
\end{equation}
Applying Proposition 3.2.1 by Peccati and Taqqu(2011), we obtain
\begin{align*}
\var(\R^{\ast2} ) &= \cov\left(\R^{\ast2}, \R^{\ast2}\right) \\
  & = \sum_{i=\tilde{t} m+1}^{(\tilde{t}+\tilde{H})m}\sum_{j=\tilde{t} m+1}^{(\tilde{t}+\tilde{H})m}\sum_{u=\tilde{t} m+1}^{(\tilde{t}+\tilde{H})m}
\sum_{v=\tilde{t}m+1}^{(\tilde{t}+\tilde{H})m}
\cov(r^{\ast}_ir^{\ast}_j,\; r^{\ast}_u r^{\ast}_v) \\
  & = \sum_{i=\tilde{t} m+1}^{(\tilde{t}+\tilde{H})m}\sum_{j=\tilde{t} m+1}^{(\tilde{t}+\tilde{H})m}\sum_{u=\tilde{t} m+1}^{(\tilde{t}+\tilde{H})m}
      \sum_{v=\tilde{t}m+1}^{(\tilde{t}+\tilde{H})m} \cum(r^{\ast}_i,r^{\ast}_j, r^{\ast}_u, r^{\ast}_v)  \\
  & + 2 \sum_{i=\tilde{t} m+1}^{(\tilde{t}+\tilde{H})m}\sum_{j=\tilde{t} m+1}^{(\tilde{t}+\tilde{H})m}\sum_{u=\tilde{t} m+1}^{(\tilde{t}+\tilde{H})m}
      \sum_{v=\tilde{t}m+1}^{(\tilde{t}+\tilde{H})m}\cov(r^{\ast}_i,r^{\ast}_v )\cov(r^{\ast}_u,r^{\ast}_j) \\
  & = \sum_{i=\tilde{t} m+1}^{(\tilde{t}+\tilde{H})m}\sum_{j=\tilde{t} m+1}^{(\tilde{t}+\tilde{H})m}\sum_{u=\tilde{t} m+1}^{(\tilde{t}+\tilde{H})m}
      \sum_{v=\tilde{t}m+1}^{(\tilde{t}+\tilde{H})m}\cum(r_i,r_j, r_u, r_v) \tag{cum1} \\
  & + 2 \sum_{i=\tilde{t} m+1}^{(\tilde{t}+\tilde{H})m}\sum_{j=\tilde{t} m+1}^{(\tilde{t}+\tilde{H})m}\sum_{u=\tilde{t} m+1}^{(\tilde{t}+\tilde{H})m}
         \sum_{v=\tilde{t}m+1}^{(\tilde{t}+\tilde{H})m}\cov(r_i,r_v )\cov(r_u,r_j)\tag{cov1}
\end{align*}
When $d=0$, $\{r_t\}$ is h-dependent. For a fixed $i$, if $|i-v| > h$, then $r_i$ is independent of $r_v$. Similarly, for a fixed $u$, if
$|u-j| > h$, then $r_u$ is independent of $r_j$.
\\
Hence
\begin{align*}
\mathrm{cov1} & = \sum_{i=\tilde{t}m+1}^{(\tilde{t}+\tilde{H})m}\sum_{v=\tilde{t}m+1}^{(\tilde{t}+\tilde{H})m}\cov(r_i, r_v)
\sum_{u=\tilde{t}m+1}^{(\tilde{t}+\tilde{H})m}\sum_{j=\tilde{t}m+1}^{(\tilde{t}+\tilde{H})m}\cov(r_u, r_j) \\
 &= \sum_{i=\tilde{t}m+1}^{(\tilde{t}+\tilde{H})m}\sum_{|i-v| \le h}\cov(r_i, r_v)
\sum_{u=\tilde{t}m+1}^{(\tilde{t}+\tilde{H})m}\sum_{|u-j| \le h}\cov(r_u, r_j) \\
 & \sim \tilde{H}^2 m \sum_{|i-v| \le h}\sum_{|u-j| \le h}\cov(r_i, r_v)\cov(r_u, r_j)\;.
\end{align*}
If any group of ${r_i, r_j, r_u, r_v}$ is independent of the remaining elements, $\cum(r_i,r_j,r_u,r_v)=0$.
Hence for fixed $i$, if $|i-v| \leq h$, $|i-j| \leq h$, and $|i-u| \leq h$, $\cum(r_i,r_j,r_u,r_v) \neq 0$.
Thus
\[
\mbox{cum1} \sim \tilde{H}m \sum_{|i-j| \le h}\sum_{|i-u| \le h}\sum_{|i-v| \le h}\cum(r_i,r_j, r_u, r_v)\;.
\]
Combining these results, $\var({\R^{\ast}}^2 ) \sim \tilde{H}^2 C_1$, where
\begin{align*}
C_1 & = m\sum_{|i-v| \le h}\sum_{|u-j| \le h}\cov(r_i, r_v)\cov(r_u, r_j)\\
  & + m\sum_{|i-j| \le h}\sum_{|i-u| \le h}\sum_{|i-v| \le h}\cum(r_i,r_j,r_u,r_v) \;.
\end{align*}

Similarly,
\begin{align*}
\cov\left(\R, \R^2 \right) & = \sum_{i=\tilde{t}m+1}^{(\tilde{t}+\tilde{H})m}\sum_{j=\tilde{t}m+1}^{(\tilde{t}+\tilde{H})m}\sum_{u=\tilde{t}m+1}^{(\tilde{t}+\tilde{H})m}
\cov(r_i, r_jr_u) \\
& = \sum_{i=\tilde{t}m+1}^{(\tilde{t}+\tilde{H})m}\sum_{j=\tilde{t}m+1}^{(\tilde{t}+\tilde{H})m}\sum_{u=\tilde{t}m+1}^{(\tilde{t}+\tilde{H})m}
\cum(r_i, r_j, r_u) \tag{cum2}\\
& + 2 \sum_{i=\tilde{t}m+1}^{(\tilde{t}+\tilde{H})m}\sum_{j=\tilde{t}m+1}^{(\tilde{t}+\tilde{H})m}\sum_{u=\tilde{t}m+1}^{(\tilde{t}+\tilde{H})m}
\esp[r_v]\cov(r_i, r_j) \tag{cov2}
\end{align*}
Applying the similar argument as for cum1 and cov1, we obtain
\begin{align*}
\mbox{cum2} & \sim \tilde{H}m \sum_{|i-u| \le h}\sum_{|i-v| \le h}\cum(r_i, r_j, r_u)  \\
\mbox{cov2} & = \tilde{H}m\esp[r_v]\sum_{i=\tilde{t}m+1}^{(\tilde{t}+\tilde{H})m}\sum_{j=\tilde{t}m+1}^{(\tilde{t}+\tilde{H})m}\cov(r_i, r_j)\ind{|i-j|\leq h} \\
            & = \tilde{H}m\esp[r_v]\sum_{j=\tilde{t}m+1}^{\tilde{t}m+h}\left((\tilde{t}+\tilde{H})m - |j-h|\right)\cov(r_0, r_{j-\tilde{t}m-h}) \\
            & \sim \tilde{H}^2 C_2  \;,
\end{align*}
where $C_2$ is a positive constant.
Hence $\cov\left(\R, \R^2 \right) \sim \tilde{H}^2 C_2$.
Combining the results for $\var({\R^{\ast}}^2 )$ and $\cov\left(\R, \R^2 \right)$ as well as the results that $\esp[\R] = \tilde{H}m\esp[r_t]$ and $\var(\R)\sim \tilde{H} C_3$ when $d=0$, where $C_3$ is a constant (see (\ref{eq:var_agg_Rt_expression})), it is straightforward to show that $\var \left(R^2_{\tilde{t}, \tilde{t}+\tilde{H}}\right) \sim \tilde{H}^3 C^{\ast}$, where $C^{\ast}$ is a constant.
With similar arguments, it can be shown $\var \left(\tR^2 \right) \sim \tilde{H}^3 C^{\ast\ast}$.

\end{proof}

\begin{lemma}\label{lem:var_RV_square}
If $d=0$,
\[
\var \left({RV}^2_{\tilde{t}-\tilde{H}, \tilde{t}}\right) \sim \tilde{H}^3 D^{\ast}
\]
where $D^{\ast
}$ is a positive constant.
\end{lemma}
\begin{proof}
We skip the proof since it is similar to the proof of \ref{lem:var_R_square}.
\end{proof}

\begin{corollary}\label{coro:cov_R_RV_square}
For every $0 < L \le \tilde{H}$,
\[
\cov \left(R^2_{\tilde{t}, \tilde{t}+\tilde{H}}, \; R^2_{\tilde{t}+L, \tilde{t}+L+\tilde{H}}\right) \sim D_1\tilde{H}^3
\]
\[
\cov \left(RV^2_{\tilde{t}-\tilde{H}, \tilde{t}}, \; RV^2_{\tilde{t}+L-\tilde{H}, \tilde{t}+L}\right) \sim D_2 \tilde{H}^3
\]
where $D_1$ and $D_2$ are positive constants.
\end{corollary}

\begin{lemma}
\begin{equation}\label{eq:var_rt_sq}
\var(r^2_t) =  \esp[r^4_t] -(\var(r_t))^2 - 2(\esp[r_t])^2 \var(r_t) - (\esp[r_t])^4
\end{equation}
where
$\var(r_t)$ is given by (\ref{eq:var_r0}),
\begin{equation}
\esp[r^4_t] = (\mu^4+6\mu^2\sigma^2_e+3\sigma_e^4)\esp[\nu] + (7\mu^4 + 18\mu^2\sigma_e^2)\esp[\nu^2] + (6\mu^4+6\mu^2\sigma_e^2)\esp[\nu^3] + \mu^4 \esp[\nu^4] \;,
\end{equation}
$\nu = \esp[\Delta N(t)|\Lambda]$, and $\esp[\nu]$, $\esp[\nu^2]$, $\esp[\nu^3]$, and $\esp[\nu^4]$ are given by (\ref{eq:nu}), (\ref{eq:nu_sq}), (\ref{eq:nu_cube}), and (\ref{eq:nu_quad}).
\end{lemma}
\begin{proof}
Applying Proposition 3.2.1 by Peccati and Taqqu(2011), we can express $\var(r_t^2)$ as
\begin{equation}
 \var(r^2_t) 
   = \cum(r_t, r_t, r_t, r_t) + 4 \esp[r_t]\cum(r_t, r_t, r_t) + 2\cum(r_t, r_t)\cum(r_t, r_t) + 4(\esp[r_t])^2\cum(r_t, r_t)
\end{equation}
\\
Let $r^{\ast}_t = r_t - \esp[r_t]$. Then $\cum(r_t, r_t, r_t, r_t) = \cum(r^{\ast}_t, r^{\ast}_t, r^{\ast}_t, r^{\ast}_t)$, and $\cum(r_t, r_t, r_t) = \cum(r^{\ast}_t, r^{\ast}_t, r^{\ast}_t)$. Therefore,
\begin{equation}\label{eq:var_r_sq}
\var(r^2_t)  = \cum(r^{\ast}_t, r^{\ast}_t, r^{\ast}_t, r^{\ast}_t)
    + 4 \esp[r_t]\cum(r^{\ast}_t, r^{\ast}_t, r^{\ast}_t)
    + 2 (\var(r_t))^2
    + 4(\esp[r_t])^2\var(r_t) \;.
\end{equation}
 where
 \begin{align*}
 \cum(r^{\ast}_t, r^{\ast}_t, r^{\ast}_t, r^{\ast}_t) & = -3 \esp[(r^{\ast}_t)^2]^2 + \esp[(r^{\ast}_t)^4] \\
 \cum(r^{\ast}_t, r^{\ast}_t, r^{\ast}_t) &= \esp[(r^{\ast}_t)^3]
 \end{align*}
\\
Let $\mu_r = \esp[r_t]$.  Then
\begin{align}
\esp[(r^{\ast}_t)^4] &= \esp[r^4_t] - 4\mu_r\esp[r^3_t] + 6\mu^2_r\esp[r^2_t] - 3\mu_r^4 \label{eq:r_cen_4th_moment} \\
\esp[(r^{\ast}_t)^3]  &= \esp[r^3_t] - 3\mu_r\esp[r^2_t] + 2\mu_r^3 \label{eq:r_cen_3rd_moment}\\
\esp[(r^{\ast}_t)^2]  &= \var(r_t)\label{eq:r_cen_2nd_moment}
\end{align}
Plugging (\ref{eq:r_cen_4th_moment}), (\ref{eq:r_cen_3rd_moment}), and (\ref{eq:r_cen_2nd_moment}) into (\ref{eq:var_r_sq}), we obtain
\begin{equation}\label{eq:var_r_sq_expression}
\var(r^2_t) = -(\var(r_t))^2 + \esp[r^4_t] - 2\mu^2_r \var(r_t) - \mu^4_r
\end{equation}
where
\begin{align}
\esp[r^4_t] &= \esp\left[\mu\Delta N(t)+\sum_{k=N(t-1)+1}^{N(t)}e_k \right]^4 \notag \\
            &= \mu^4 \esp[(\Delta N(t))^4] + 6\mu^2\sigma_e^2 \esp[(\Delta N(t))^3] + 3\esp[\Delta N(t)]\sigma_e^4 \label{eq:r_4th_moment}
\end{align}
\\
\\
Conditional on $\Lambda$, $\Delta N(t)$ is Poisson process. Let $\nu = \esp[\Delta N(t)|\Lambda]$. $\nu = \lambda \int_{t-1}^t \rme^{Z_H(s)}\rmd s$.
Hence
\begin{align}
\esp[(\Delta N(t))^4] & = \esp[\esp[(\Delta N(t))^4|\Lambda]] = \esp[\nu] + 7\esp[\nu^2] + 6 \esp[\nu^3] + \esp[\nu^4] \\
\esp[(\Delta N(t))^3] & = \esp[\esp[(\Delta N(t))^3|\Lambda]] = \esp[\nu] + 3\esp[\nu^2] + \esp[\nu^3]
\end{align}
and thus $\esp[r^4_0]$ can be represented as
\begin{equation}
\esp[r^4_t] = (\mu^4+6\mu^2\sigma^2_e+3\sigma_e^4)\esp[\nu] + (7\mu^4 + 18\mu^2\sigma_e^2)\esp[\nu^2] + (6\mu^4+6\mu^2\sigma_e^2)\esp[\nu^3] + \mu^4 \esp[\nu^4]\;,
\end{equation}
where
\begin{align}
\esp[\nu] &= \lambda \rme^{1/2} \label{eq:nu}\\
\esp[\nu^2] &=  \lambda^2 \int_{-1}^0 \int_{-1}^0 \esp[\rme^{Z_H(u)}\rme^{Z_H(v)}]\rmd u\rmd v =  \lambda^2 \int_{-1}^0 \int_{-1}^0 \rme^{1+\gammaz(t-s)}\rmd t \rmd s \label{eq:nu_sq}\\
\esp[\nu^3] &=  \lambda^3 \int_{-1}^0 \int_{-1}^0 \int_{-1}^0 \esp[\rme^{Z_H(t)}\rme^{Z_H(s)}\rme^{Z_H(u)}]\rmd t \rmd s\rmd u  \notag \\
            &= \lambda^3 \int_{-1}^0 \int_{-1}^0 \int_{-1}^0  \rme^{\frac{3}{2}+\gammaz(t-s)+\gammaz(t-u)+\gammaz(s-u)}\rmd t \rmd s \rmd u \label{eq:nu_cube}\\
\esp[\nu^4] &=  \lambda^4 \int_{-1}^0 \int_{-1}^0 \int_{-1}^0 \int_{-1}^0\esp[\rme^{Z_H(t)}\rme^{Z_H(s)}\rme^{Z_H(u)}\rme^{Z_H(u)}]\rmd t \rmd s \rmd u\rmd v \notag \notag\\
            &= \lambda^4 \int_{-1}^0 \int_{-1}^0 \int_{-1}^0 \int_{-1}^0  \rme^{2+\gammaz(t-s)+\gammaz(t-u)+\gammaz(t-v)\gammaz(s-u)+\gammaz(s-v)+\gammaz(u-v)}\rmd t \rmd s \rmd u \rmd v \label{eq:nu_quad}
\end{align}

\end{proof}

\begin{lemma}\label{lem:convg_sample_mean}
Define
\begin{eqnarray*}
\overline{R_{\tilde{H}}} &=& \frac{1}{\tilde{T}-2\tilde{H}}\sum_{\tilde{t}=\tilde{H}+1}^{\tilde{T}-\tilde{H}} R_{\tilde{t}, \tilde{t}+\tilde{H}} \\
\overline{RV_{\tilde{H}}} &=& \frac{1}{\tilde{T}-2\tilde{H}}\sum_{\tilde{t}=\tilde{H}+1}^{\tilde{T}-\tilde{H}} RV_{\tilde{t}-\tilde{H}, \tilde{t}}\;,
\end{eqnarray*}
where $\tilde{H} = \tilde{T}^{\kappa}$, $0 < \kappa < 1$.
If $d=0$, for every $\epsilon^{\prime} > 0$,
\begin{eqnarray*}
\frac{\overline{R_{\tilde{H}}}}{\esp[R_{\tilde{t}, \tilde{t}+\tilde{H}}]} &=& 1 + o_p(\tilde{T}^{-1/2+\epsilon^{\prime}/2})\;, \\
\frac{\overline{RV_{\tilde{H}}}}{\esp[RV_{\tilde{t}-\tilde{H}, \tilde{t}}]} &=& 1 + o_p(\tilde{T}^{-1/2+\epsilon^{\prime}/2})\;.
\end{eqnarray*}

\end{lemma}

\begin{proof}
The expected value of $\frac{\overline{R_{\tilde{H}}}}{\esp[R_{\tilde{t}, \tilde{t}+\tilde{H}}]}$ is equal to
\[
\esp \left[\frac{\overline{R_{\tilde{H}}}}{\esp[R_{\tilde{t}, \tilde{t}+\tilde{H}}]} \right] = 1\;,
\]
and its variance is equal to
\[
\var\left(\frac{\overline{R_{\tilde{H}}}}{\esp[R_{\tilde{t}, \tilde{t}+\tilde{H}}]} \right)
=\frac{\var\left(\overline{R_{\tilde{H}}}\right)}{\left(\esp[R_{\tilde{t}, \tilde{t}+\tilde{H}}]\right)^2}
= \frac{1}{m^2 (\esp[r_t])^2}\left[\frac{\var\left(\overline{R_{\tilde{H}}}\right)}{\tilde{H}^2}   \right] \;,
\]
where
\[
\var\left(\overline{R_{\tilde{H}}}\right) = \frac{1}{\left(\tilde{T}-2\tilde{H}\right)^2}\left[\var\left(\sum_{\tilde{t}=\tilde{H}}^{\tilde{T}-\tilde{H}}R_{\tilde{t}, \tilde{t}+\tilde{H}} \right)\right]\;,
\]
and
\begin{eqnarray*}
\var\left(\sum_{\tilde{t}=\tilde{H}}^{\tilde{T}-\tilde{H}}R_{\tilde{t}, \tilde{t}+\tilde{H}} \right)
 &=& (\tilde{T}-2\tilde{H})\var(R_{\tilde{t}, \tilde{t}+\tilde{H}})
 + \mathop{\sum_{\tilde{t}=\tilde{H}+1}^{\tilde{T}-\tilde{H}} \sum_{\tilde{k}=\tilde{H}+1}^{\tilde{T}-\tilde{H}}}_{\tilde{t}\ne \tilde{k}}\cov\left(R_{\tilde{t}, \tilde{t}+\tilde{H}},\; R_{\tilde{k}, \tilde{k}+\tilde{H}}   \right)\\
 &=& (\tilde{T}-2\tilde{H})\var(R_{\tilde{t}, \tilde{t}+\tilde{H}})
 + 2 \sum_{0 < L < \tilde{T}-2\tilde{H}}(\tilde{T}-2\tilde{H}-L)\;\cov \left(R_{\tilde{t},\tilde{t}+\tilde{H}},R_{\tilde{t}+L,\; \tilde{t}+\tilde{H}+L}\right)
\end{eqnarray*}
Let $\gamma_R(L) = \cov(R_{\tilde{t}, \tilde{t}+\tilde{H}}\;, R_{\tilde{t}+L, \tilde{t}+L+\tilde{H}})$.
\\
From Corollary \ref{coro:d0_cov_R_formula}, if $d=0$, $R_{\tilde{t}, \tilde{t}+\tilde{H}}$ is $\tilde{H}+h/m$ dependent.
Hence we can express
\begin{equation}
\sum_{0 < L < \tilde{T}-2\tilde{H}}(\tilde{T}-2\tilde{H}-L)\;\cov \left(R_{\tilde{t},\tilde{t}+\tilde{H}},R_{\tilde{t}+L,\; \tilde{t}+\tilde{H}+L}\right)
=\sum_{L=1}^{(\tilde{H}+h/m)\; \wedge \; (\tilde{T}-2\tilde{H})}(\tilde{T}-2\tilde{H}-L)\gamma_R(L)\;.
\end{equation}
By Lemmas \ref{lem:convergence_agg_ret} and \ref{lem:cov_R_H}, $\var(R_{\tilde{t}, \tilde{t}+\tilde{H}}) \sim C_{\mu, \lambda}  \tilde{H}$,
and $\forall \; 0 < L \le \tilde{H}$, $\gamma_R(L) \sim C_4 \tilde{H})$, where both $C_{\mu, \lambda} $ and $C_4$ are positive constants.
\\
\\
If $(\tilde{H}+h/m) < \tilde{T}-2\tilde{H}$,
\[
\var\left(\sum_{\tilde{t}=\tilde{H}}^{\tilde{T}-\tilde{H}}R_{\tilde{t}, \tilde{t}+\tilde{H}} \right) = O((\tilde{T}-2\tilde{H})\tilde{H}^2)
\]
\\
and thus
\begin{eqnarray*}
\var\left(\overline{R_{\tilde{H}}}\right)
& = &  O\left(\frac{\tilde{H}^2}{\tilde{T}}\right)\;.
\end{eqnarray*}
\\
It follows that $\var\left(\frac{\overline{R_{\tilde{H}}}}{\esp[R_{\tilde{t}, \tilde{t}+\tilde{H}}]} \right) = O(\tilde{T}^{-1})$.
Note that $\tilde{H}+h/m < \tilde{T}-2\tilde{H} $ implies $0 < \kappa < \frac{\log(\tilde{T}-h/m)-\log 3}{ \log\tilde{T}}$, which converges to $1$ as
$\tilde{T} \to \infty$.
\\
By Chebyshev's inequality, $ \forall\; \epsilon^{\prime} > 0$, we obtain
\[
P \left(\frac{\left|\frac{\overline{R_{\tilde{H}}}}{\esp[R_{\tilde{t}, \tilde{t}+\tilde{H}}]} -1 \right|}{\tilde{T}^{-1/2 +\epsilon^{\prime}/2}} \ge \epsilon  \right)
\le \frac{(\tilde{T}^{1-\epsilon^{\prime}})\var\left(\frac{\overline{R_{\tilde{H}}}}{\esp[R_{\tilde{t}, \tilde{t}+\tilde{H}}]} \right)}{\epsilon^2}
 = O(\tilde{T}^{-\epsilon^{\prime}}) \rightarrow 0 \;,
\]
as $\tilde{T} \to \infty$.
\\
\\
The proof of $RV_{\tilde{t}-\tilde{H}, \tilde{t}}$ follows from a similar argument.
\end{proof}

\begin{corollary}\label{coro:convergence_sample_mean}
If $d=0$, for $0 < \kappa < 1/2$,
\begin{eqnarray*}
\overline{R_{\tilde{H}}} -  \esp[\R]  &\stackrel{p}{\longrightarrow}& 0 \\
\overline{RV_{\tilde{H}}} -  \esp[\RV] &\stackrel{p}{\longrightarrow}& 0
\end{eqnarray*}
\end{corollary}


\begin{remark}\label{rmk:convergence_sample_mean_tilde}
Follow the proof of Lemma \ref{lem:convg_sample_mean} and Corollary \ref{coro:convergence_sample_mean}, the same limiting results apply to $\overline{\tilde{R}_{\tilde{H}}}$.
\end{remark}

\begin{lemma}\label{lem:convg_sample_mean_sqRet}
When $d=0$,
for $0 < \kappa < 1$ and for every $\epsilon^{\prime} > 0$,
\[
\frac{\frac{1}{\tilde{T}-2\tilde{H}}\sum_{\tilde{t}=\tilde{H}+1}^{\tilde{T}-\tilde{H}}R^2_{\tilde{t}, \tilde{t}+\tilde{H}}}{\esp[R^2_{\tilde{t},\tilde{t}+\tilde{H}}]}
=1 + o_p(\tilde{T}^{-1/2+\epsilon^{\prime}/2}) \;.
\]
\end{lemma}
\begin{proof}
The variance of the ratio is equal to
\[
\var \left(\frac{\frac{1}{\tilde{T}-2\tilde{H}}\sum_{\tilde{t}=\tilde{H}+1}^{\tilde{T}-\tilde{H}}R^2_{\tilde{t}, \tilde{t}+\tilde{H}}}{\esp[R^2_{\tilde{t},\tilde{t}+\tilde{H}}]} \right)
= \frac{1}{(\esp[R^2_{\tilde{t}, \tilde{t}+\tilde{H}}])^2(\tilde{T}-2\tilde{H})^2}\; \var \left(\sum_{\tilde{t}=\tilde{H}+1}^{\tilde{T}-\tilde{H}} R^2_{\tilde{t}, \tilde{t}+\tilde{H}} \right)
\]
Note that
\[
\esp[R^2_{\tilde{t}, \tilde{t}+\tilde{H}}] = \var(R_{\tilde{t}, \tilde{t}+\tilde{H}}) + (\esp[R_{\tilde{t}, \tilde{t}+\tilde{H}}])^2
\]
Since $\var(R_{\tilde{t}, \tilde{t}+\tilde{H}}) \sim C_{\lambda, \mu}\tilde{H}$ and $\esp(R_{\tilde{t}, \tilde{t}+\tilde{H}})$ is proportional to $\tilde{H}$,
we have $\esp[R^2_{\tilde{t}, \tilde{t}+\tilde{H}}] = O(\tilde{H}^2)$.
From Lemma \ref{lem:var_R_square}, we obtain
\begin{equation}\label{eq:var_sum_bigR_square_1}
\sum_{\tilde{t}=\tilde{H}+1}^{\tilde{T}-\tilde{H}}\var \left(R^2_{\tilde{t}, \tilde{t}+\tilde{H}} \right) = O((\tilde{T}-2\tilde{H})\tilde{H}^3)\;.
\end{equation}
\\
\\
By Corollary \ref{coro:cov_R_RV_square}, if $d=0$, for every $0 < L \le \tilde{H}$, $\cov(R^2_{\tilde{t}, \tilde{t}+\tilde{H}},\;R^2_{\tilde{t}+L, \tilde{t}+L+\tilde{H}})\sim D_1\tilde{H}^3$.
Hence, if $\tilde{H}+h/m < \tilde{T}-2\tilde{H}$, which implies that $0 < \kappa < 1$,
\begin{equation}\label{eq:var_sum_bigR_square_2_a}
\sum_{L=1}^{\tilde{H}+h/m\;\wedge \tilde{T}-2\tilde{H}}\left(\tilde{T}-2\tilde{H}-L \right)\cov \left(R^2_{\tilde{t}, \tilde{t}+\tilde{H}}\;, R^2_{\tilde{t}+L, \tilde{t}+L+\tilde{H}}\right) = O((\tilde{T}-2\tilde{H})\tilde{H}^4)\;.
\end{equation}
Combining (\ref{eq:var_sum_bigR_square_1}) and (\ref{eq:var_sum_bigR_square_2_a}), we obtain
\\
\[
\var\left(\sum_{\tilde{t}=\tilde{H}+1}^{\tilde{T}-\tilde{H}} R^2_{\tilde{t}, \tilde{t}+\tilde{H}} \right) = O((\tilde{T}-2\tilde{H})\tilde{H}^4)\;,
\]
and thus
\[
\var \left(\frac{\frac{1}{\tilde{T}-2\tilde{H}}\sum_{\tilde{t}=\tilde{H}+1}^{\tilde{T}-\tilde{H}}R^2_{\tilde{t}, \tilde{t}+\tilde{H}}}{\esp[R^2_{\tilde{t},\tilde{t}+\tilde{H}}]} \right) = O(\tilde{T}^{-1})\;.
\]
\\
By Chebysev's inequality, for every $\epsilon^{\prime} > 0$,
\[
 P \left(\frac{\left| \frac{\frac{1}{\tilde{T}-2\tilde{H}}\sum_{\tilde{t}=\tilde{H}+1}^{\tilde{T}-\tilde{H}}R^2_{\tilde{t}, \tilde{t}+\tilde{H}}}{\esp[R^2_{\tilde{t},\tilde{t}+\tilde{H}}]} -1 \right|}{\tilde{T}^{-1/2+\epsilon^{\prime}/2}} \ge \epsilon^{\prime}  \right)\le \frac{(\tilde{T}^{1/2-\epsilon^{\prime}})\var \left(\frac{\frac{1}{\tilde{T}-2\tilde{H}}\sum_{\tilde{t}=\tilde{H}+1}^{\tilde{T}-\tilde{H}}R^2_{\tilde{t}, \tilde{t}+\tilde{H}}}{\esp[R^2_{\tilde{t},\tilde{t}+\tilde{H}}]} \right)}{{\epsilon^{\prime}}^2} = O(\tilde{T}^{-1/2-\epsilon^{\prime}})
\]
which converges to 0 as $\tilde{T} \to \infty$.

\end{proof}

\begin{remark}\label{rmk:convg_sample_mean_sqRV_tildeR}
  Follow the proof of Lemma \ref{lem:convg_sample_mean_sqRet}, the same limit results apply to $\tR^2$ and $\RV^2$.
\end{remark}

\begin{lemma}\label{lem:convg_ratio_S_R_varR}

Let
\begin{eqnarray}
S^2_R &=& \frac{1}{\tilde{T}-2\tilde{H}}\sum_{\tilde{t}=\tilde{H}+1}^{\tilde{T}-\tilde{H}}\left(R_{\tilde{t},\tilde{t}+\tilde{H}}-\overline{R_{\tilde{H}}}\right)^2 \label{eq:S_R}\\
S^2_{RV} &=& \frac{1}{\tilde{T}-2\tilde{H}}\sum_{\tilde{t}=\tilde{H}+1}^{\tilde{T}-\tilde{H}}\left(RV_{\tilde{t}-\tilde{H}, \tilde{t}}-\overline{RV_{\tilde{H}}}\right)^2 \label{eq:S_RV}
\end{eqnarray}
When $d=0$,
if $0 < \kappa < 1$,
\begin{eqnarray*}
\frac{S^2_R}{\var(R_{\tilde{t}, \tilde{t}+\tilde{H}})} &=& 1 + o_p(\tilde{T}^{-1/2+\epsilon^{\prime}/2})\\
\frac{S^2_{RV}}{\var(RV_{\tilde{t}\tilde{H}, \tilde{t}})} &=& 1 + o_p(\tilde{T}^{-1/2+\epsilon^{\prime}/2})
\end{eqnarray*}
\end{lemma}

\begin{proof}
We can express
\begin{eqnarray*}
\frac{S^2_R}{\var(R_{\tilde{t}, \tilde{t}+\tilde{H}})}
&=& \frac{\frac{1}{\tilde{T}-2\tilde{H}}\sum_{\tilde{t}=\tilde{H}+1}^{\tilde{T}-\tilde{H}} R^2_{\tilde{t}, \tilde{t}+\tilde{H}}}{\esp[R^2_{\tilde{t}, \tilde{t}+\tilde{H}}]}\left( \frac{\esp[R^2_{\tilde{t}, \tilde{t}+\tilde{H}}]}{\var(R_{\tilde{t}, \tilde{t}+\tilde{H}})} \right)
 - \frac{(\overline{R_{\tilde{H}}})^2}{(\esp[R_{\tilde{t}, \tilde{t}+\tilde{H}}])^2}\left(\frac{(\esp[R_{\tilde{t}, \tilde{t}+\tilde{H}}])^2}{\var(R_{\tilde{t}, \tilde{t}+\tilde{H}})}\right) \\
&=&   \frac{\frac{1}{\tilde{T}-2\tilde{H}}\sum_{\tilde{t}=\tilde{H}+1}^{\tilde{T}-\tilde{H}} R^2_{\tilde{t}, \tilde{t}+\tilde{H}}}{\esp[R^2_{\tilde{t}, \tilde{t}+\tilde{H}}]} \left(1 + \frac{(\esp[R_{\tilde{t}, \tilde{t}+\tilde{H}}])^2}{\var(R_{\tilde{t}, \tilde{t}+\tilde{H}})}   \right)
- \frac{(\overline{R_{\tilde{H}}})^2}{(\esp[R_{\tilde{t}, \tilde{t}+\tilde{H}}])^2}\left(\frac{(\esp[R_{\tilde{t}, \tilde{t}+\tilde{H}}])^2}{\var(R_{\tilde{t}, \tilde{t}+\tilde{H}})}\right)
\end{eqnarray*}
From Lemma \ref{lem:convg_sample_mean_sqRet} as well as Lemma \ref{lem:convg_sample_mean},
if $ 0 < \kappa < 1$ with $\epsilon^{\prime}$ chosen so that $-1/2 +\epsilon^{\prime}/2 < 0$,
\begin{eqnarray*}
\frac{S^2_R}{\var(R_{\tilde{t}, \tilde{t}+\tilde{H}})} &=& (1+o_p(\tilde{T}^{-1/2+\epsilon^{\prime}/2})) \left(1 + \frac{(\esp[R_{\tilde{t}, \tilde{t}+\tilde{H}}])^2}{\var(R_{\tilde{t}, \tilde{t}+\tilde{H}})}   \right) -
(1+o_p(\tilde{T}^{-1/2+\epsilon^{\prime}/2})\left(\frac{(\esp[R_{\tilde{t}, \tilde{t}+\tilde{H}}])^2}{\var(R_{\tilde{t}, \tilde{t}+\tilde{H}})}\right)\\
&=& 1 + o_p(\tilde{T}^{-1/2+\epsilon^{\prime}/2})\;.
\end{eqnarray*}
The proof for $\frac{S^2_{RV}}{\var(RV_{\tilde{t}\tilde{H}, \tilde{t}})}$ is similar to $S^2_R$ and thus is skipped.

\end{proof}

\begin{lemma}\label{lem:convergence_norm_S_R}
If $d=0$,
\begin{eqnarray}
\frac{S^2_R}{\tilde{H}} & \convprob m \left(\var(r_0) + 2 \sum_{k=1}^h \gamma_r(k)\right) \label{eq:S_R_convergence}\\
\frac{S^2_{RV}}{\tilde{H}} & \convprob m \left(\var(r^2_0) + 2 \sum_{k=1}^h \tilde{\gamma}_r(k)\right) \label{eq:S_RV_convergence}
\end{eqnarray}
where $\gamma_r(k)= \cov(r_t, r_{t+k)}$ and $\tilde{\gamma}_r(k)=\cov(r^2_{t}, r^2_{t+k})$.
\end{lemma}
\begin{proof}
When $d=0$, $r_t$ is $h$-dependent, hence (\ref{eq:var_agg_Rt_expression}) is equal to
\[
\var(R_{\tilde{t}, \tilde{t}+\tilde{H}}) = \tilde{H}m \left[\var(r_0) + 2\sum_{k=1}^h \left(1-\frac{k}{\tilde{H}m} \right)\gamma_r(k) \right] \;.
\]
Therefore,
\begin{eqnarray*}
\frac{S^2_R}{\var(R_{\tilde{t}, \tilde{t}+\tilde{H}})} =
\frac{S^2_R}{\tilde{H}}\;\frac{1}{m[\var(r_0) + 2\sum_{k=1}^h(1-\frac{k}{\tilde{H}m})\gamma_r(k)]}
\end{eqnarray*}
As $\tilde{H} \to \infty$, $\sum_{k=1}^h \left(1-\frac{k}{\tilde{H}m} \right)\gamma_r(k) \to \sum_{k=1}^h \gamma_r(k)$.
By Lemma \ref{lem:convg_ratio_S_R_varR}, for every $\tilde{H}+h/m < \tilde{T}-2\tilde{H}$ and every $\epsilon^{\prime} > 0$,
\[
\frac{S^2_R}{\tilde{H}}\;\frac{1}{m[\var(r_0) + 2\sum_{k=1}^h\gamma_r(k)]} = 1 + o_p(\tilde{T}^{-1/2+\epsilon^{\prime}/2})
\]
Therefore,
\[
\frac{S^2_R}{\tilde{H}} \convprob m \left(\var(r_0) + 2 \sum_{k=1}^h \gamma_r(k)\right)\;.
\]
\\
The steps to obtain the limiting result for $S^2_{RV}$ is similar to $S^2_R$ and thus is skipped here.
\end{proof}

\begin{remark}\label{rmk:convg_norm_tilde_S_R}
Follow Remark \ref{rmk:convg_sample_mean_sqRV_tildeR} and the proofs of Lemmas \ref{lem:convg_ratio_S_R_varR} and \ref{lem:convergence_norm_S_R}, the same limiting results apply to $S^2_{\tilde{R}}$, where $S^2_{\tilde{R}}=\frac{1}{\tilde{T}-2\tilde{H}}\sum_{\tilde{t}=\tilde{H}+1}^{\tilde{T}-\tilde{H}}\left(\tR-\overline{\tilde{R}_{\tilde{H}}}\right)^2$.
\end{remark}

\subsection*{Proof of Theorem \ref{thm:asym_normal_dist_norm_rho}}

\begin{proof}

\noindent Let
\begin{equation}
\tilde{C}_{\tilde{T}}=\frac{1}{\tilde{T}-2\tilde{H}}\sum_{\tilde{t}=\tilde{H}+1}^{(\tilde{T}-\tilde{H})}\left(\tR-\overline{\tilde{R}_{\tilde{H}}}\right)
\left(RV_{\tilde{t}-\tilde{H},\tilde{t}}-\overline{RV_{\tilde{H}}}\right)\;,
\end{equation}
and the theoretically-centered sample covariance as
\[
\tilde{C}^{\ast}_{\tilde{T}} = \frac{1}{\tilde{T}-2\tilde{H}}\sum_{\tilde{t}=\tilde{H}+1}^{\tilde{T}-\tilde{H}}\left(\tilde{R}_{\tilde{t}, \tilde{t}+\tilde{H}}-\esp[\tilde{R}_{\tilde{t},\tilde{t}+\tilde{H}}]\right)\left(RV_{\tilde{t}-\tilde{H},\tilde{t}}-\esp[RV_{\tilde{t}-\tilde{H},\tilde{t}}]\right) \;.
\]
Let
\[
S^2_{\tilde{R}} = \frac{1}{\tilde{T}-2\tilde{H}}\sum_{\tilde{t}=\tilde{H}+1}^{\tilde{T}-\tilde{H}}\left(\tR -\overline{\tilde{R}_{\tilde{H}}}\right)^2 \;,
\]
and
\[
S^2_{RV} = \frac{1}{\tilde{T}-2\tilde{H}}\sum_{\tilde{t}=\tilde{H}+1}^{\tilde{T}-\tilde{H}}\left(RV_{\tilde{t}-\tilde{H},\tilde{t}}-\overline{RV_{\tilde{H}}}\right)^2 \;.
\]
We can then express
\begin{equation}\label{eq:tilde_rho_expression}
\tilde{\rho}_{\tilde{H}, \tilde{T}} = \frac{\tilde{C}_{\tilde{T}}}{S_{\tilde{R}} S_{RV}} = \frac{\tilde{C}_{\tilde{T}}- \Tilde{C}^{\ast}_{\tilde{T}}+\tilde{C}^{\ast}_{\tilde{T}}}{S_{\tilde{R}} S_{RV}}
= \frac{\tilde{C}^{\ast}_{\tilde{T}}}{S_{\tilde{R}} S_{RV}} +  \frac{\tilde{C}_{\tilde{T}}-\tilde{C}^{\ast}_{\tilde{T}}}{S_{\tilde{R}} S_{RV}}
\end{equation}

\noindent By Lemma \ref{lem:asm_var_Ct} and Theorem \ref{thm:CLT}, for every $\kappa \in (0, 1/3)$
\begin{equation}
\sqrt{\frac{\tilde{T}}{\tilde{H}^3}}\Tilde{C}^{\ast}_{\tilde{T}} \convdistr N\left(0\;, \frac{2}{3}m^2 \sum_{|u|\leq h}\sum_{|v|\leq h} \cov(r_0,r_u)  \cov(r_0^2,r_v^2) \right)\;.
\end{equation}
From Lemma \ref{lem:convergence_norm_S_R} and Remark \ref{rmk:convg_norm_tilde_S_R}, 
when $d=0$,
$\frac{S^2_{\tilde{R}}}{\tilde{H}}\stackrel{p}{\longrightarrow} m \left(\var(r_0) + 2 \sum_{k=1}^h \gamma_r(k)\right)$ and
\\
$\frac{S^2_{RV}}{\tilde{H}} \stackrel{p}{\longrightarrow} m \left(\var(r^2_0) + 2 \sum_{k=1}^h \tilde{\gamma}_r(k)\right)$.
Therefore, by Slutsky's Theorem, it follows that for every $\kappa \in (0, 1/3)$,
\begin{equation}\label{eq:1st_term_asym_dist_rho_hat}
\ddfrac{\left(\sqrt{\frac{\tilde{T}}{\tilde{H}^3}}\right)\Tilde{C}^{\ast}_{\tilde{T}}}{\sqrt{\frac{S_{\tilde{R}}^2}{\tilde{H}}\frac{S^2_{RV}}{\tilde{H}}}}= \sqrt{\tilde{T}^{1-\kappa}}\;\frac{\Tilde{C}^{\ast}_{\tilde{T}}}{S_{\tilde{R}} S_{RV}} \convdistr N \left(0, \frac{S^2}{A_1A_2} \right)\;,
\end{equation}
where $S^2=\frac{2}{3}  \sum_{|u|\leq h}\sum_{|v|\leq h} \cov(r_0,r_u)  \cov(r_0^2,r_v^2)$, $A_1=\var(r_0) + 2 \sum_{k=1}^h \cov(r_0, r_k)$, and $A_2= \var(r^2_0) + 2 \sum_{k=1}^h \cov(r^2_0, r^2_k)$.
\\
\\
Next, we show for $0 < \kappa < 1/3$, $\sqrt{\tilde{T}^{1-\kappa}}\left(\frac{\tilde{C}_{\tilde{T}}-\Tilde{C}^{\ast}_{\tilde{T}}}{S_{\tilde{R}} S_{RV}}\right) \to 0$.
\\
We can express
\begin{eqnarray*}
\tilde{C}_{\tilde{T}}&-& \Tilde{C}^{\ast}_{\tilde{T}} \\
&=& \frac{1}{\tilde{T}-2\tilde{H}}\sum_{\tilde{t}=\tilde{H}+1}^{\tilde{T}-\tilde{H}}\left[(\tR-\overline{\tR})(\RV-\overline{RV_{\tilde{H}}})  - (\tR-\esp[\tR])(\RV-\esp[\RV])\right] \\
&=& \left(\overline{RV_{\tilde{H}}}-\esp[\RV]\right)\left(\esp[\tR]- \overline{\tR} \right)
\end{eqnarray*}
\\
From Lemma \ref{lem:convg_sample_mean} and Remark  \ref{rmk:convergence_sample_mean_tilde}, when $d=0$, as $\tilde{T} \to \infty$, $\forall\; 0 < \kappa < 1/3$,
\begin{eqnarray*}
\overline{\tilde{R}_{\tilde{H}}} &=& \esp[\tR] + \esp[\tR]o_p(\tilde{T}^{-1/2+\epsilon^{\prime}/2}) = \esp[\tR] + o_p(\tilde{T}^{\kappa-1/2+\epsilon^{\prime}/2}) \\
\overline{RV_{\tilde{H}}} &=& \esp[\RV] + \esp[\RV]o_p(\tilde{T}^{-1/2+\epsilon^{\prime}/2}) = \esp[\RV] + o_p(\tilde{T}^{\kappa-1/2+\epsilon^{\prime}/2})
\end{eqnarray*}
It follows
\begin{equation}\label{eq:diff_tilde_C_tilde_C_ast}
\tilde{C}_{\tilde{T}}-\Tilde{C}^{\ast}_{\tilde{T}}  = \left(\overline{RV_{\tilde{H}}}-\esp[RV]\right)\left(\esp[\tR]- \overline{\tilde{R}_{\tilde{H}}} \right) = o_p(\tilde{T}^{2\kappa-1+\epsilon^{\prime}})\;.
\end{equation}
\\
Since
\begin{eqnarray}
\sqrt{\tilde{T}^{1-\kappa}}\left(\frac{\tilde{C}_{\tilde{T}}-\Tilde{C}^{\ast}_{\tilde{T}}}{S_{\tilde{R}} S_{RV}}\right)
&=& \sqrt{\tilde{T}^{1-\kappa}}\left(\frac{\tilde{C}_{\tilde{T}}-\Tilde{C}^{\ast}_{\tilde{T}}}{\sqrt{\frac{S_{\tilde{R}}^2}{\tilde{H}}\frac{S^2_{RV}}{\tilde{H}}}}\right)\left(\frac{1}{\tilde{H}}\right)
\end{eqnarray}
From Remark \ref{rmk:convg_norm_tilde_S_R}, for $0 < \kappa < 1/3$, $\frac{S^2_{\tilde{R}} S^2_{RV}}{\tilde{H}^2} = A_1 A_2 + o_p(\tilde{T}^{-1/2+\epsilon^{\prime}/2})$,
thus we obtain
\begin{equation}\label{eq:asym_diff_Ct}
\sqrt{\tilde{T}^{1-\kappa}}\left(\frac{\tilde{C}_{\tilde{T}}-\Tilde{C}^{\ast}_{\tilde{T}}}{S_{\tilde{R}} S_{RV}}\right) = o_p(\tilde{T}^{\kappa/2-1/4+3\epsilon^{\prime}/4})\;,
\end{equation}
which converges to zero.
Therefore, by (\ref{eq:tilde_rho_expression}),(\ref{eq:1st_term_asym_dist_rho_hat}), and (\ref{eq:asym_diff_Ct}), for every $\kappa \in (0, 1/3)$,
\[
\sqrt{\tilde{T}^{1-\kappa}}\; \tilde{\rho}_{\tilde{H}, \tilde{T}} \convdistr N\left(0,\; \frac{S^2}{A_1 A_2} \right).
\]
\end{proof}

\begin{lemma}\label{lem:asym_cov_RRV}
\noindent Let ${\tilde{R}}_k^{\ast}$ and ${RV}^{\ast}_k$ be the centered versions of $\tilde{R}_k$ and $RV_k$:
${\tilde{R}}_k^{\ast} = \tilde{R}_k - \esp[\tilde{R}_k]$, ${RV}_k^{\ast} = {RV}^{\ast}_k - \esp[{RV}_k]$, where $\tilde{R}_k = \sum_{j=km+h+1}^{(k+\tilde{H})m}r_j $
and ${RV}_k = \sum_{j=km+1}^{(k+\tilde{H})m}r^2_j$.
Let $r^{\ast}_i=r_i-\esp[r_i]$ and $s_i^{\ast}=r_i^2-\esp[r_i^2]$.
\\
(i) if $k > \tilde{H}$, then
\begin{equation}
 \cov({\tilde{R}}_0^{\ast} {RV}^{\ast}_0,\; {\tilde{R}}_k^{\ast} {RV}^{\ast}_k ) = 0 \;.
\end{equation}
\\
(ii) if $0 \leq k \leq \tilde{H} $, then
\begin{equation}
\cov({\tilde{R}}_0^{\ast} {RV}^{\ast}_0,\; {\tilde{R}}_k^{\ast} {RV}^{\ast}_k ) = c_k + O(1) \;,
\end{equation}
where
$c_k$ is given by (\ref{eq:cov1_ret}).
\end{lemma}

\begin{proof}
Applying Proposition 3.2.1 in Peccati and Taqqu (2011) and properties of cumulants, we show that
for $k \geq 0$,
  \begin{align*}
    \cov({\tilde{R}}_0^{\ast} {RV}^{\ast}_0,\; {\tilde{R}}_k^{\ast} {RV}^{\ast}_k )
    & = \sum_{i=h+1}^{\tilde{H}m} \sum_{j=-\tilde{H}m+1}^{0}\sum_{u=km+h+1}^{(k+\tilde{H})m} \sum_{v=(k-\tilde{H})m+1}^{k}\cov(r_i^{\ast}s_j^{\ast}, r_u^{\ast}s_v^{\ast}) \\
    & = \sum_{i=h+1}^{\tilde{H}m} \sum_{j=-\tilde{H}m+1}^{0}\sum_{u=km+h+1}^{(k+\tilde{H})m} \sum_{v=(k-\tilde{H})m+1}^{k} \cum(r_i, r_j^2, r_u, r_v^2) \tag{cum} \\
    & + \sum_{i=h+1}^{\tilde{H}m} \sum_{j=-\tilde{H}m+1}^{0}\sum_{u=km+h+1}^{(k+\tilde{H})m} \sum_{v=(k-\tilde{H})m+1}^{k} \cov(r_i ,r_u) \; \cov(r_j^2,r_v^2) \tag{cov1} \\
    & + \sum_{i=h+1}^{\tilde{H}m} \sum_{j=-\tilde{H}m+1}^{0}\sum_{u=km+h+1}^{(k+\tilde{H})m} \sum_{v=(k-\tilde{H})m+1}^{k} \cov(r_u ,r_j^2) \; \cov(r_i,r_v^2) \; . \tag{cov2}
  \end{align*}
\\
Consider the term (cum).
\\
Since $|i-j| > h$, $|u-j| > h$, and $|u-v| > h$, $r_i$ is independent of $r_j$ and $r_u$ is independent of $r_j$ and $r_v$.
\\
A. If $k > \tilde{H}$, then $r_i$ is independent of $r_u$. Also $r_j$ and $r_v$ are independent of $r_u$, hence the elements of $\{r_i, r_j, r_v\}$ are independent of $r_u$, and
thus $\cum(r_i, r_j^2, r_u, r_v^2)=0$.
Therefore, if $k > \tilde{H}$, the term (cum) is zero.
\\
\\
B. If $0\leq k \leq \tilde{H}$, then $(k-\tilde{H})m+1 \leq 1$. We can split the range of $v$ into two parts:
\\
\\
(i) if $(k-\tilde{H})m+1 \leq v \leq 0$, then $r_i$ is also independent of $r_v$, and the elements of $\{r_i, r_u\}$ are independent of the elements $\{r_j, r_v\}$. Hence $\cum(r_i,r_j^2,r_u,r_v^2)=0$.
\\
\\
(ii) if $ 1 \leq v  \leq km$,
the (cum) term equals
\begin{eqnarray}
  \sum_{i=h+1}^{\tilde{H}m} \sum_{j=-\tilde{H}m+1}^{0}\sum_{u=km+h+1}^{(k+\tilde{H})m} \sum_{v=1}^{km}  \cum(r_i,r_j^2,r_u,r_v^2)\label{eq:cum_1}\;.
\end{eqnarray}
\\
There are terms in (\ref{eq:cum_1}) with $\cum(r_i,r_j^2,r_u,r_v^2)=0$.
\\
If $ -\tilde{H}m+1 \leq j \leq -h$ or $v \geq h+1$, then $r_j$ is independent of $r_v$, and
the elements of $\{r_i, r_u, r_v\}$ are independent of $\{r_j\}$. Thus $\cum(r_i,r_j^2,r_u,r_v^2)=0$.
Therefore, after excluding those terms with $\cum(r_i,r_j^2,r_u,r_v^2)=0$, (\ref{eq:cum_1}) becomes
\begin{eqnarray}
 \sum_{i=h+1}^{\tilde{H}m} \sum_{j=-h+1}^{0}\sum_{u=km+h+1}^{(k+\tilde{H})m} \sum_{v=1}^{h}  \cum(r_i,r_j^2,r_u,r_v^2) \label{eq:cum_2}\;.
\end{eqnarray}
\\
Next we check if the elements of $\{r_i, r_u\}$ are independent of the elements of $\{r_j, r_v\}$ in (\ref{eq:cum_2}).
If $u > 2h $, then $r_u$ is independent of $r_v$ and $r_j$. Also if $i > 2h$, then $r_i$ is independent of $r_j$ and $r_v$.
It follows $\{r_i, r_u\}$ is independent of $\{r_j, r_v\}$ and $\cum(r_i,r_j^2,r_u,r_v^2)=0$. Hence we have
\begin{eqnarray}
 \sum_{i=2h+1}^{\tilde{H}m} \sum_{j=-h+1}^{0}\sum_{u=(km+h+1)\vee(2h+1)}^{(k+\tilde{H})m} \sum_{v=1}^{h}  \cum(r_i,r_j^2,r_u,r_v^2) = 0 \;. \label{eq:cum_3}
\end{eqnarray}
\\
In the following, we discuss two cases regarding the relationship between $km+h+1$ and $2h+1$:
\\
\\
(B.1). If $(km+h+1)\vee (2h+1) = 2h+1$, we can express (\ref{eq:cum_2}) as
\begin{eqnarray}
&& \sum_{i=h+1}^{\tilde{H}m} \sum_{j=-h+1}^{0}\sum_{u=km+h+1}^{2h} \sum_{v=1}^{h}  \cum(r_i,r_j^2,r_u,r_v^2)\label{eq:cum_4} \\
&+&\sum_{i=h+1}^{\tilde{H}m} \sum_{j=-h+1}^{0}\sum_{u=2h+1}^{(k+\tilde{H})m} \sum_{v=1}^{h}  \cum(r_i,r_j^2,r_u,r_v^2)\label{eq:cum_5}
\end{eqnarray}
After excluding the terms in (\ref{eq:cum_3}) from (\ref{eq:cum_5}), (\ref{eq:cum_2}) becomes
\begin{eqnarray}
&& \sum_{i=h+1}^{\tilde{H}m} \sum_{j=-h+1}^{0}\sum_{u=km+h+1}^{2h} \sum_{v=1}^{h}  \cum(r_i,r_j^2,r_u,r_v^2)\label{eq:cum_6} \\
&+& \sum_{i=h+1}^{2h} \sum_{j=-h+1}^{0}\sum_{u=2h+1}^{(k+\tilde{H})m} \sum_{v=1}^{h}  \cum(r_i,r_j^2,r_u,r_v^2)\label{eq:cum_7}\;.
\end{eqnarray}
In (\ref{eq:cum_6}), if $i > 3h$ then $r_i$ is independent of $r_u$ and the elements of $\{r_i, r_j\}$ are independent of the elements of $\{r_u, r_v\}$.
Also in (\ref{eq:cum_7}), if $u > 3h$, then $r_i$ is independent of $r_u$ and the set $\{r_i, r_j, r_v\}$ is independent of $r_u$ .
After excluding these terms with $\cum(r_i,r_j^2,r_u,r_v^2)=0$, sum of (\ref{eq:cum_6}) and (\ref{eq:cum_7}) become
\begin{eqnarray}\label{eq:cum_term_re1}
&& \sum_{i=h+1}^{3h} \sum_{j=-h+1}^{0}\sum_{u=km+h+1}^{2h} \sum_{v=1}^{h}  \cum(r_i,r_j^2,r_u,r_v^2)\notag\\
&+& \sum_{i=h+1}^{2h} \sum_{j=-h+1}^{0}\sum_{u=2h+1}^{(k+\tilde{H})m\wedge 3h} \sum_{v=1}^{h}  \cum(r_i,r_j^2,r_u,r_v^2)\notag \\
&=& O(1)\;,
\end{eqnarray}
uniformly \wrt\ $k$ and $\tilde{H}$.
\\
(B.2). Similarly, if $(km+h+1)\vee (2h+1) = km+h+1$. Hence we can express (\ref{eq:cum_3}) as
\begin{equation}\label{eq:cum_8}
\sum_{i=2h+1}^{\tilde{H}m} \sum_{j=-h+1}^{0}\sum_{u=km+h+1}^{(k+\tilde{H})m} \sum_{v=1}^{h}  \cum(r_i,r_j^2,r_u,r_v^2)
\end{equation}
After excluding the terms in (\ref{eq:cum_8}) from (\ref{eq:cum_2}), (\ref{eq:cum_2}) becomes
\begin{equation}
\sum_{i=h+1}^{2h} \sum_{j=-h+1}^{0}\sum_{u=km+h+1}^{(k+\tilde{H})m} \sum_{v=1}^{h}  \cum(r_i,r_j^2,r_u,r_v^2) \label{eq:cum_9}
\end{equation}
If $u > 3h$, then $r_u$ is independent of $r_i$ and $r_v$, and the elements of $\{r_i, r_j, r_v\}$ are independent of $r_u$.
Therefore, (\ref{eq:cum_9}) becomes
\begin{equation}\label{eq:cum_term_re2}
\sum_{i=h+1}^{2h} \sum_{j=-h+1}^{0}\sum_{u=km+h+1}^{(k+\tilde{H})m \wedge 3h} \sum_{v=1}^{h}  \cum(r_i,r_j^2,r_u,r_v^2) = O(1)\;,
\end{equation}
From the results of (\ref{eq:cum_term_re1}) and (\ref{eq:cum_term_re2}), it follows $\mbox{(cum)}=O(1)$, uniformly \wrt\ $k$ and $\tilde{H}$.
\\
\\
Next, we consider the sum (cov2). Since $u > km+h+1$ and $j \leq 0$, $r_u$ is independent of $r_j$, and thus $\cov(r_u, r^2_j)=0$.
It follows
\begin{equation}\label{eq:cov2_ret}
\mbox{cov2} = \sum_{i=h+1}^{\tilde{H}m} \sum_{v=(k-\tilde{H})m+1}^{km} \cov(r_i,r_v^2) \sum_{j=-\tilde{H}m+1}^{0}\sum_{u=km+h+1}^{(k+\tilde{H})m} \cov(r_u ,r_j^2) =0 \;.
\end{equation}
\\
Finally, consider the sum (cov1).
\\
If $k > \tilde{H}$, then $r_i$ and $r_u$ are independent and thus $\cov(r_i ,r_u) = 0$. It follows that $\mbox{cov1}=0$.
Also (cov2) and (cum) are both zero if $k > \tilde{H}$.  Therefore, if $k > \tilde{H}$, $\cov({\tilde{R}}_0^{\ast} {RV}^{\ast}_0,\; {\tilde{R}}_k^{\ast} {RV}^{\ast}_k ) =0$.
\\
If $0 \leq k \leq \tilde{H}$, since $\{r_i\}$ are $h$-dependent, we can express
\begin{align}\label{eq:cov1_ret}
 \mathrm{cov1}
  & = \sum_{i=h+1}^{\tilde{H}m}  \sum_{j=-\tilde{H}m+1}^{0}\sum_{u=km+h+1}^{(k+\tilde{H})m} \sum_{v=(k-\tilde{H})m+1}^{km} \cov(r_i,r_u)\cov(r_j^2,r_v^2) \notag \\
  & =  \sum_{u=(k-\tilde{H})m+h+1\vee(-h)}^{(k+\tilde{H})m-h-1\wedge h}\sum_{i=h+1}^{\tilde{H}m} \ind{km+h+1 \leq i+u\leq (k+\tilde{H})m} \; \cov(r_0,r_{u}) \notag \\
  & \phantom{ = } \times \sum_{v=(k-\tilde{H})m+1\vee(-h)}^{(k+\tilde{H})m-1\wedge h}\sum_{j=-\tilde{H}m+1}^{0}  \ind{(k-\tilde{H})m+1 \leq j+v\leq km} \; \cov(r_0^2,r_v^2) \notag\\
  & =  \sum_{u=(k-\tilde{H})m+h+1\vee(-h)}^{(k+\tilde{H})m-h-1\wedge h}  (\tilde{H}m-h-|km-u|)^+ \cov(r_0,r_{u}) \notag \\
  & \times \sum_{v=(k-\tilde{H})m+1\vee(-h)}^{(k+\tilde{H})m-1\wedge h} (\tilde{H}m-|km-v|)^+ \cov(r_0^2,r_v^2) \;.\notag \\
\end{align}
Let $c_k = \mbox{cov1}$. By the results below (\ref{eq:cum_term_re2}) as well as (\ref{eq:cov2_ret}), it follows that
\[
\cov({\tilde{R}}_0^{\ast} {RV}^{\ast}_0,\; {\tilde{R}}_k^{\ast} {RV}^{\ast}_k ) = c_k + O(1) \;.
\]

\end{proof}

\begin{lemma}\label{lem:asm_var_Ct}
Define $\tilde{C}^{\ast}_{\tilde{T}}$, $V_{\tilde{T}}$,  $W_{\tilde{T}}$ and $\tilde{\rho}_{\tilde{H}, \tilde{T}}$ by
\begin{align*}
  V_{\tilde{T}} & = {\tilde{T}}^{-1}\sum_{k=1}^{\tilde{T}}({\tilde{R}}_k^{\ast})^2 \; , \ \
  W_{\tilde{T}} = {\tilde{T}}^{-1}\sum_{k=1}^{\tilde{T}} ({RV}_k^{\ast})^2 \; , \ \
        \tilde{C}^{\ast}_{\tilde{T}}  = {\tilde{T}}^{-1}\sum_{k=1}^{\tilde{T}} {\tilde{R}}_k^{\ast} {RV}^{\ast}_k \; , \ \
  \tilde{\rho}_{\tilde{H},\tilde{T}}  = \frac{\tilde{C}^{\ast}_{\tilde{T}}}{\sqrt{V_{\tilde{T}}W_{\tilde{T}}}} \; .
\end{align*}
  \begin{align*}
  \lim_{\tilde{H}\to\infty}  \tilde{T} {\tilde{H}}^{-3} \var(\tilde{C}^{\ast}_{\tilde{T}})  = \frac{2}{3}  \sum_{|u|\leq h}\sum_{|v|\leq h} \cov(r_0,r_u) \cov(r_0^2,r_v^2)\;.
  \end{align*}
\end{lemma}

\begin{proof}
By definition,
 \begin{equation}
  \var(\tilde{C}^\ast_{\tilde{T}})  = \frac{1}{\tilde{T}}\var\left(\tilde{R}_0^{\ast} {RV}^{\ast}_0  \right)
                     + \frac{2}{\tilde{T}^2} \sum_{k=1}^{\tilde{T}} (\tilde{T}-k)\cov\left(\tilde{R}_0^{\ast} {RV}^{\ast}_0 , \tilde{R}_k^{\ast} {RV}^{\ast}_k   \right)
 \end{equation}
\\
From Lemma \ref{lem:asym_cov_RRV}, if $k > \tilde{H}$, $\cov\left(\tilde{R}_0^{\ast} {RV}^{\ast}_0 , \tilde{R}_k^{\ast} {RV}^{\ast}_k \right) = 0$.
Hence we can express
  \begin{align}\label{eq:norm_var_Ct}
    T {\tilde{H}}^{-3} \var(\tilde{C}^{\ast}_T)
    & = {\tilde{H}}^{-3} \var({\tilde{R}}_0^{\ast} {RV}^{\ast}_0) + 2 {\tilde{H}}^{-3} \sum_{k=1}^{\tilde{H}}\left (1-\frac{k}{\tilde{T}} \right)(c_k + O(1)) \notag \\
    & = {\tilde{H}}^{-3} \var({\tilde{R}}_0^{\ast} {RV}^{\ast}_0) + 2 {\tilde{H}}^{-3} \sum_{k=1}^{\tilde{H}}\; c_k - 2{\tilde{H}}^{-3}\sum_{k=1}^{\tilde{H}}\frac{k}{\tilde{T}}c_k \notag \\
    & + 2{\tilde{H}}^{-3}\sum_{k=1}^{\tilde{H}}\left (1-\frac{k}{\tilde{T}} \right)O(1)
  \end{align}
\\
From (\ref{eq:cov1_ret}) as well as the results below (\ref{eq:cum_term_re2}) and (\ref{eq:cov2_ret}), if $k=0$,
\[
\lim_{\tilde{H} \to \infty}\var({\tilde{R}}_0^{\ast} {RV}^{\ast}_0) = O({\tilde{H}}^2)
\]
Thus
\begin{equation}\label{eq:eq:decom_ck_1}
{\tilde{H}}^{-3} \var({\tilde{R}}_0^{\ast} {RV}^{\ast}_0) = O({\tilde{H}}^{-1})\;.
\end{equation}
\\
If $1\leq k \leq \tilde{H}$, then
\begin{align}\label{eq:c_k_bound}
|c_k| & = {\tilde{H}^2 m^2} \left|\sum_{|u|\leq h} \left(1-\frac{h}{\tilde{H}m}-\frac{|km-u|}{\tilde{H}m}\right)^+ \cov(r_0,r_{u})
                    \times \sum_{|v|\leq h} \left(1-\frac{|km-v|}{\tilde{H}m}\right)^+ \cov(r_0^2,r_v^2)\right| \notag \\
    & \leq {\tilde{H}^2 m^2} \sum_{|u|\leq h}  \sum_{|v|\leq h}\left| \cov(r_0,r_{u})\cov(r_0^2,r_v^2) \right|
\end{align}
Therefore,
\begin{equation}\label{eq:eq:decom_ck_2}
{\tilde{H}}^{-3}\left|\sum_{k=1}^{\tilde{H}}\frac{k}{\tilde{T}}c_k \right|
= \frac{{\tilde{H}}^{-2}}{\tilde{T}}\left|\sum_{k=1}^{\tilde{H}} \frac{k}{\tilde{H}}c_k \right| \leq \frac{{\tilde{H}}^{-2}}{\tilde{T}}\sum_{k=1}^{\tilde{H}} |c_k| \leq \frac{{\tilde{H}}^{-2}}{\tilde{T}}{\tilde{H}} O({\tilde{H}}^2) = O({\tilde{H}}/\tilde{T})\;.
\end{equation}
Also
\begin{align}\label{eq:eq:decom_ck_3}
{\tilde{H}}^{-3}\sum_{k=1}^{\tilde{H}}\left (1-\frac{k}{\tilde{T}} \right) = \frac{1}{{\tilde{H}}^2} - \frac{1}{2\tilde{T}{\tilde{H}}^2} - \frac{1}{2\tilde{T}{\tilde{H}}} \to 0 \;,
\end{align}
as $\tilde{H} \to \infty$.
Therefore, by the results of (\ref{eq:eq:decom_ck_1}), (\ref{eq:eq:decom_ck_2}) and (\ref{eq:eq:decom_ck_3}), Equation (\ref{eq:norm_var_Ct}) can be represented as
\begin{equation}
  \tilde{T} {\tilde{H}}^{-3} \var(\tilde{C}^{\ast}_{\tilde{T}})  =  2 {\tilde{H}}^{-3} \sum_{k=1}^{\tilde{H}} c_k + O({\tilde{H}}^{-1} + {\tilde{H}}/\tilde{T})\;.
\end{equation}
Consider the term
\begin{align*}
    {\tilde{H}}^{-3} \sum_{k=1}^{\tilde{H}} c_k
     & = {\tilde{H}}^{-3} \sum_{|u|\leq h}\sum_{|v|\leq h} \cov(r_0,r_u)  \cov(r_0^2,r_v^2)  \sum_{k=1}^{\tilde{H}}  (\tilde{H}m-h-|km-u|)^+(\tilde{H}m-|km-v|)^+
\end{align*}
Then
  \begin{align*}
    \lim_{\tilde{H}\to\infty} {\tilde{H}}^{-3} \sum_{k=1}^{\tilde{H}} c_k
    & = \lim_{\tilde{H}\to\infty}  m^2  {\tilde{H}}^{-1} \sum_{k=1}^{\tilde{H}} \sum_{|u|\leq h}\sum_{|v|\leq h} \left(1-\frac{|k-u|}{\tilde{H}}\right)^+ \left(1-\frac{|k-v|}{\tilde{H}}\right)^+ \cov(r_0,r_u)  \cov(r_0^2,r_v^2) \\
    & = m^2 \lim_{\tilde{H}\to\infty}  {\tilde{H}}^{-1} \sum_{k=1}^{\tilde{H}}  \left(1-\frac{k}{\tilde{H}}\right)^2 \sum_{|u|\leq h}\sum_{|v|\leq h} \cov(r_0,r_u)  \cov(r_0^2,r_v^2) \\
    & \rightarrow m^2 \int_0^1 (1-t)^2\rmd t \sum_{|u|\leq h}\sum_{|v|\leq h} \cov(r_0,r_u)  \cov(r_0^2,r_v^2) \\
    & = \frac{1}{3} m^2 \sum_{|u|\leq h}\sum_{|v|\leq h} \cov(r_0,r_u)  \cov(r_0^2,r_v^2) \; .
   \end{align*}
 \\
 It follows
 \begin{align*}
  \lim_{\tilde{H}\to\infty}  \tilde{T} {\tilde{H}}^{-3} \var(\tilde{C}^{\ast}_{\tilde{T}})  = \frac{2}{3} m^2 \sum_{|u|\leq h}\sum_{|v|\leq h} \cov(r_0,r_u)  \cov(r_0^2,r_v^2)\;.
  \end{align*}

\end{proof}

\begin{theorem}\label{thm:CLT}
Let $T_k = {\tilde{R}}_k^{\ast} {RV}^{\ast}_k = \sum_{i=km+h+1}^{(k+\tilde{H})m} (r_i-\esp[r_0])\sum_{j=(k-\tilde{H})m+1}^{km} (r_j^2-\esp[r_0^2])$, where ${\tilde{R}}_k^{\ast}$
and ${RV}^{\ast}_k$ are defined in Lemma \ref{lem:asym_cov_RRV}, and $1\leq k \leq \tilde{T}$.
If $d=0$ and $\tilde{H}^3/\tilde{T}\to 0$, then
  \begin{align*}
      (\tilde{T}\tilde{H}^3)^{-1/2}\sum_{k=1}^{\tilde{T}} T_k \convdistr N\left(0\;,\;\frac23 m^2 \sum_{|u|\leq h}\sum_{|v|\leq h}\cov(r_0,r_u)\cov(r_0^2,r_v^2)\right) \; .
  \end{align*}
\end{theorem}

\begin{proof}
In order to prove the central limit theorem, we apply the theorem of \cite{berk:1973}.
Let $cst$ be a generic positive constant. We must check the following conditions
  \begin{enumerate}[(i)]
  \item \label{item:berk-i} There exists an integer $p>2$ such that for $1 \leq k \leq \tilde{T}$, $\esp[|T_k|^p] \leq \constant\; \tilde{H}^p$ and $\tilde{H}^{(2p-2)/(p-2)}/\tilde{T} \to 0$.
  \item \label{item:berk-ii}$\var(T_1+\cdots+T_k) \leq \constant \;  k {\tilde{H}}^3$, for $1 \leq k \leq \tilde{T}$.
  \item \label{item:berk-iii}
    $\lim_{\tilde{T}\to\infty} (\tilde{T}{\tilde{H}}^3)^{-1}\var(T_1+\cdots+T_{\tilde{T}})=\constant \sum_{|u|\leq h}\sum_{|v|\leq h}\cov(r_0,r_u)\cov(r_0^2,r_v^2)$.
  \end{enumerate}

\noindent To check conditions (ii) and (iii), first we derive the expression for
 \begin{align*}
 \var\left(\sum_{j=1}^k T_j\right) &= k\var(T_0) + 2\sum_{j=1}^{k}(k-j)\cov(T_0, T_j)
\end{align*}
 By Lemma \ref{lem:asym_cov_RRV}, if $k > \tilde{H} $, $\cov(T_0, T_k)=0$.
 If $k=0$, $\var(T_0) = O(\tilde{H}^2)$.
 For $1 \leq k \leq \tilde{H}$, $\cov(T_0, T_k)=c_k + O(1)$. Thus by (\ref{eq:c_k_bound}),
 \begin{align*}
 \sum_{j=1}^k(k-j)\cov(T_0, T_j) & = \sum_{j=1}^{k\wedge \tilde{H}}(k-j)\cov(T_0, T_j) \\
                                 & = \sum_{j=1}^{k\wedge \tilde{H}}(k-j)\left(c_j + O(1) \right) \\
                                 & =  k \sum_{j=1}^{k \wedge \tilde{H}} \left(1-\frac{j}{k}\right)c_k + O(k \wedge \tilde{H}) \\
                                 & \leq  k \sum_{j=1}^{k \wedge \tilde{H}} |c_j| + O(k \wedge \tilde{H})\\
                                 & \leq  k\;O(\tilde{H}^3)\;.
 \end{align*}

\noindent It follows that
 \begin{equation}
  \var\left(\sum_{j=1}^k T_j\right) \leq k\;O(\tilde{H}^3)\;,
 \end{equation}
which proves condition (\ref{item:berk-ii}).
\\
Let $\tilde{C}^{\ast}_{\tilde{T}} = \frac{1}{\tilde{T}}\sum_{k=1}^{\tilde{T}} T_k$. Since $\var(T_1+T_2+\cdots + T_{\tilde{T}}) = \tilde{T}^2 \var(\tilde{C}^{\ast}_{\tilde{T}})$.
It follows from Lemma \ref{lem:asm_var_Ct},
\begin{align*}
      (\tilde{T}K^3)^{-1}\var\left(\sum_{k=1}^{\tilde{T}} T_k \right) \rightarrow \frac23\sum_{|u|\leq h}\sum_{|v|\leq h}\cov(r_0,r_u)\cov(r_0^2,r_v^2) \; .
  \end{align*}
This proves condition (\ref{item:berk-iii}).
\\
Next, we check condition (i). Since $T_0$ has zero mean, we can express
  \begin{align*}
    \esp[T_0^4] = \cum(T_0,T_0,T_0,T_0)+3\var^2(T_0) \; .
  \end{align*}
  Thus it suffices to prove that $\cum(T_0,T_0,T_0,T_0)=O(K^4)$. This will show that
  \eqref{item:berk-i} holds with $p=4$ and the condition on $K$ is $K^{3}/n\to 0$. Applying the properties of cumulants, we have
  \begin{align*}
    \cum(T_0,T_0,T_0,T_0) & = \sum_{i=h+1}^{K} \sum_{j=-K+1}^{0}\sum_{k=h+1}^{K} \sum_{\ell=-K+1}^{0}
                            \sum_{u=h+1}^{K} \sum_{v=-K+1}^{0}\sum_{w=h+1}^{K} \sum_{z=-K+1}^{0}
                            \cum(\bar{r}_i\bar{s}^2_j,\bar{r}_k\bar{s}^2_\ell,\bar{r}_u\bar{s}^2_v,\bar{r}_w\bar{s}^2_z) \\
  \end{align*}
 By Proposition 3.2.1 by Pecatti and Taqqu (2011)
  and that $r_i$, $r_k$, $r_u$, and $r_w$ are independent of $r_j$, $r_{\ell}$, $r_v$, and $r_z$, we obtain that for $i > h$, $k > h$, $u >h$, and $w > h$,
  \begin{align}
    \cum(T_0,T_0,T_0,T_0) & = (K-h)^2K^2\var^2(r_0)\var^2(r_0^2) + (K-h)K \cum(r_0,r_0,r_0,r_0)   \cum(r_0^2,r_0^2,r_0^2,r_0^2) \notag  \\
                          & \phantom{ = } +  (K-h)^2 K\var^2(r_0)\cum(r_0^2,r_0^2,r_0^2,r_0^2) +  (K-h)K^2\var^2(r_0^2)\cum(r_0,r_0,r_0,r_0) \;.
  \end{align}
By Lemma \ref{lem:rt_finite_moments}, $\cum(T_0,T_0,T_0,T_0)= O(K^4)$.
\end{proof}

\begin{lemma}\label{lem:rt_finite_moments}
For every positive integer $p$, $\esp[r_t^p] < \infty $.
\end{lemma}
\begin{proof}
From the return model, we can express the $p^{th}$ moment of the return as
\[
\esp[r^p_t] = \esp\left[\mu\Delta N(t)+\sum_{k=N(t-1)+1}^{N(t)}e_k \right]^p
\]
The term with the highest moment on the RHS of the above equation is $\esp[(\Delta N(t))^p]$, which can be represented as
$\esp[\esp[(\Delta N(t))^p|\Lambda]]$.
Conditional on $\Lambda$, $\Delta N(t)$ is a Poisson random variable with mean $\esp[\Delta N(t)|\Lambda]=\lambda \int_{t-1}^t \rme^{Z_H(s)}\rmd s$.
Thus $\esp[(\Delta N(t))^p|\Lambda]$ is the $p^{th}$ moment of the Poisson variable $\Delta N(t)|\Lambda$.
\\
Let $\nu = \esp[\Delta N(t)|\Lambda]$ be the mean of $\Delta N(t)|\Lambda$.
Since the non-centered $p^{th}$ moment of a Poisson distribution is a polynomial in its mean with the highest order $p$,
we can express $\esp[(\Delta N(t))^p|\Lambda]$ as the sum of the expected values of $\nu$ with different exponentials.
Hence $\esp[(\Delta N(t))^p]$ involves evaluating the term $\esp[\nu^p]$.
Since $Z_H$ is Gaussian stationary process, it has finite moments.
It follows that $\esp[\nu^p]$ is finite and thus $\esp[(\Delta N(t))^p]$ is finite.

\end{proof}

\end{document}